\documentclass[12pt]{book}
        \usepackage{setspace}

\usepackage{a4wide}
\usepackage{fancyhdr}
\usepackage{amsthm}
\usepackage{amsmath}
\usepackage{amsfonts}
\usepackage{tabularx}
\usepackage{graphicx}
\usepackage[bookmarks]{hyperref}
\usepackage{hyperref}
\usepackage[utf8]{inputenc}
\usepackage{bbm}
\usepackage[ruled,vlined]{algorithm2e}
\usepackage{multirow}
\usepackage{float}
\usepackage{subfigure}

\newcommand{\bra}[1]{\mbox{$\langle #1 |$}}
\newcommand{\ket}[1]{\mbox{$| #1 \rangle$}}

\newtheorem{defn}{Definition}[section]

\newtheorem{prop}{Proposition}[section]
\newtheorem{thrm}{Theorem}[section]

\newcommand{\Tr}{\mathrm{Tr}}

\renewcommand{\chaptermark}[1]%
         {\markboth{\thechapter.\ #1}{}}
\renewcommand{\sectionmark}[1]%
         {\markright{\thesection\ #1}}
\newcommand{\LMUTitle}[9]{
  \thispagestyle{empty}
  \vspace*{\stretch{1}}
  {\parindent0cm
   \rule{\linewidth}{.7ex}}
  \begin{flushright}

    \vspace*{\stretch{1}}
    \sffamily\bfseries\Huge
    #1\\
    \vspace*{\stretch{1}}
    \sffamily\bfseries\large
    #2
    \vspace*{\stretch{1}}
  \end{flushright}
  \rule{\linewidth}{.7ex}
  \vspace*{\stretch{5}}
  \begin{center}
    \includegraphics[width=2in]{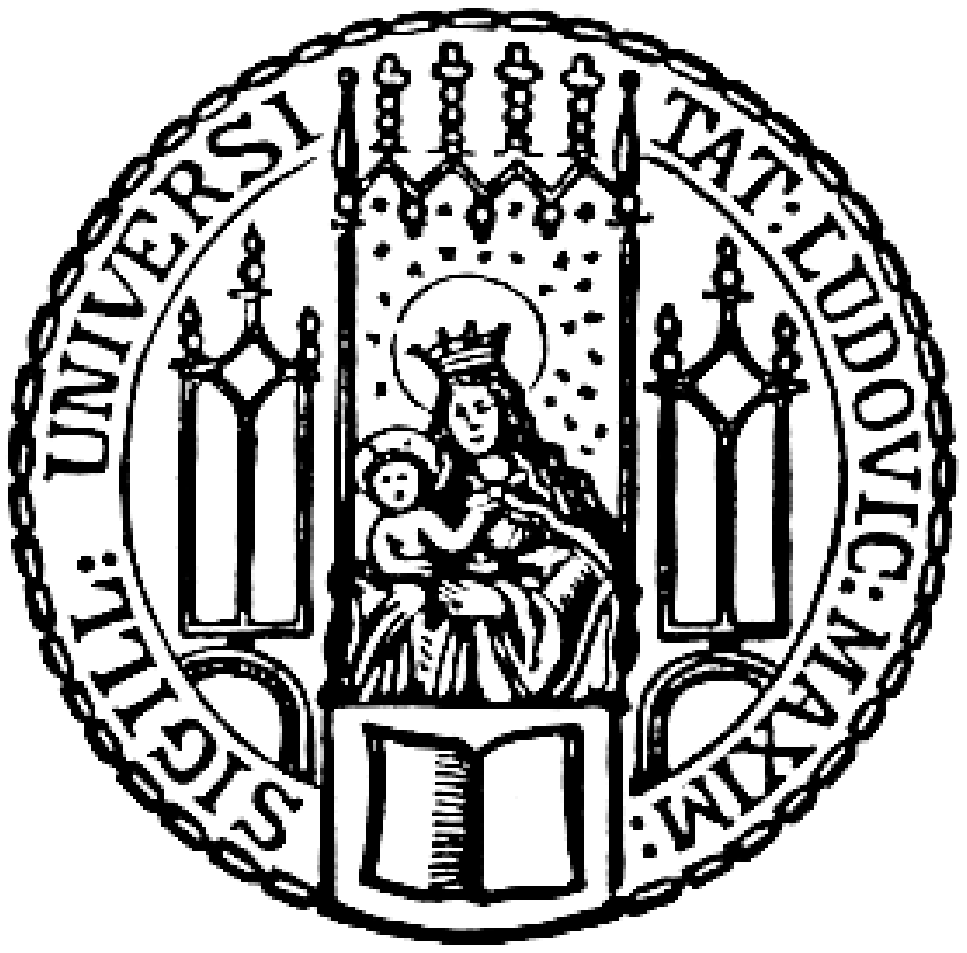}
  \end{center}
  \vspace*{\stretch{1}}
  \begin{center}\sffamily\LARGE{#5}\end{center}
  \newpage
  \thispagestyle{empty}

  \cleardoublepage
  \thispagestyle{empty}

  \vspace*{\stretch{1}}
  {\parindent0cm
  \rule{\linewidth}{.7ex}}
  \begin{flushright}
    \vspace*{\stretch{1}}
    \sffamily\bfseries\Huge
    #1\\
    \vspace*{\stretch{1}}
    \sffamily\bfseries\large
    #2
    \vspace*{\stretch{1}}
  \end{flushright}
  \rule{\linewidth}{.7ex}

  \vspace*{\stretch{3}}
  \begin{center}
    \Large \textbf{Master's Thesis}\\
    \Large 
    \textbf{Ludwig-Maximilians-Universit{\"a}t}\\
    \Large \textbf{M{\"u}nchen}\\
    \vspace*{\stretch{2}}
    \Large \textbf{Munich, #6}\\
    \vspace*{\stretch{2}}
    \Large \textbf{Supervisor:  #7} \\
  \end{center}

  \newpage
  \thispagestyle{empty}

  \vspace*{\stretch{1}}

  \cleardoublepage
}

\begin{document}

  \frontmatter

  \LMUTitle
      {Fermionic Entanglement and Correlation}
      {Lexin Ding}
      {Guangzhou} 
      {Faculty of Physics}                         
      {Munich 2020}                          
      {October 15, 2020}                            
      {Dr. Christian Schilling}                          
      {Zweitgutachter}                         
      {Pr"ufungsdatum}                         

  \mainmatter\setcounter{page}{1}
  \chapter*{}
\section*{Declaration of Authorship}

The author hereby declares that the master thesis has been written solely by himself. Results not obtained by the author have been cited in the bibliography and/or acknowledged.

\bigskip
\bigskip
\bigskip
\bigskip

Lexin Ding

Munich, October 15, 2020

\chapter*{Abstract}

Entanglement is one of the most striking features of quantum mechanics. Although the theory of entanglement for systems with distinguishable particles is well-developed, it is not directly applicable to identical fermions, as the $N$-fermion Hilbert space does not enjoy a tensor product structure. In this thesis, we study the concepts of mode entanglement and particle entanglement in fermionic systems. In particular, in the mode picture, we derived analytic formula for the entanglement between two sites/orbitals, an effective $\mathbb{C}^4 \otimes \mathbb{C}^4$ setting, while respecting the fundamental fermionic superselection rules. Using these results, we quantitatively resolved the correlation paradox in the dissociation limit, and showed that infinitesimal noise completely wipes out all the physical entanglement in the ground state of two dissociating nuclei with marginalized interaction. In molecules, we successfully separated entanglement from the total correlation between molecular orbitals. Our analysis demonstrated the drastic effect of superselection rules on the accessible entanglement between molecular orbitals, while at the same time revealed the mostly classical nature of the correlation shared between them.
  \tableofcontents

  \chapter{Introduction} \label{chap:intro}

Entanglement plays a central role in quantum information theory, where it is regarded as a highly valuable non-local resource. One can harness entanglement to perform various quantum information processing tasks that are beyond local and classical means, such as quantum teleportation  \cite{bennett1993teleporting,bouwmeester1997experimental}, quantum cryptography \cite{ekert1991quantum,bennett1992experimental} and superdense coding \cite{bennett1992communication,mattle1996dense}. Thus successfully quantifying entanglement in an operationally meaningful way is of tremendous importance. This motivation is further bolstered by recent studies that revealed a connection between entanglement and novel physical phenomena in strongly correlated many-body systems\cite{amico2008entanglement}, such as quantum phase transition \cite{vidal2003entanglement,osborne2002entanglement,osterloh2002scaling}, topological order \cite{kitaev2006topological,levin2006detecting} and chemical bonding \cite{boguslawski2013orbital,szalay2017correlation}. This is not surprising as highly entangled and complex ground states are typically responsible for such phenomena. Moreover, the theory of entanglement provides the theoretical foundation and diagnostic tools for numerical methods for solving the ground state problem such as the density matrix renormalization group (DMRG) method \cite{white1992density,legeza2003optimizing,schollwock2011density,stein2016automated}.

Despite the ubiquitous relevance of entanglement in electronic systems, the framework from quantum information theory is not immediately applicable to the fermionic setting. An illustrative example for the questionable application of quantum information theoretical tools is the attempt to quantify the ``correlation'' contained in an $N$-electron quantum state in terms of the one-particle reduced density matrix, e.g., in Refs.~\cite{Ziesche97b,huang2005entanglement}. The common reasoning is the following one: First, one defines the configuration states
\begin{equation}\label{Psiconf}
\ket{\Psi}=f_{\chi_1}^\dagger f_{\chi_2}^\dagger\cdot \ldots \cdot f_{\chi_N}^\dagger\ket{0}
\end{equation}
as being ``uncorrelated''. This seems to be plausible since ground states of \emph{non-interacting} electrons are exactly of that form
\eqref{Psiconf}, exhibiting a product structure of $N$ fermionic creation operators $f^\dagger_{\chi_j}$, populating the $N$ energetically lowest spin-orbitals $\ket{\chi_1},\ldots, \ket{\chi_N}$.

To apply the quantum information theoretical formalism which refers to \emph{distinguishable} subsystems one describes fermions by antisymmetric states within the Hilbert space $\mathcal{H}_1^{\otimes^N}$ of $N$ distinguishable particles  (``first quantization''). By referring to the tensor product $\mathcal{H}_1^{\otimes^N}$, each electron is assigned \emph{its own} one-particle Hilbert space $\mathcal{H}_1$ and  algebra of observables and the notion of reduced density operators follows then accordingly. Yet the unpleasant surprise is that even for an ``uncorrelated'' state \eqref{Psiconf} each of the $N$ electrons is still entangled with the complementary $N-1$ electrons. Indeed, the
von Neumann entropy
\begin{equation}\label{entropy}
S(\gamma)= -\mbox{Tr}[\gamma \log \gamma] = -\sum_{j} \lambda_j \log \lambda_j
\end{equation}
of the one-particle reduced density matrix (1RDM) $\gamma \equiv \Tr_{N-1}[\ket{\Psi}\!\bra{\Psi}]= 1/N\sum_{j=1}^{N}\ket{\chi_j}\!\bra{\chi_j}$ does not vanish. One tries to ``fix'' this issue by normalizing $\gamma$ to the particle number $N$ instead. This has the effect that $\gamma$'s non-vanishing eigenvalues $\lambda_j$ change from $1/N$ to $1$ and $S$ would consequently vanish as desired \cite{Ziesche97}. Yet, the von Neumann entropy \eqref{entropy} has an information theoretical origin and meaning based on probability theory\cite{Jozsa97book,plenio2014introduction} which is now unfortunately lost.

To circumvent this issue, we follow two natural routes and define new correlation quantities that respect the fermionic nature of the system. One route is to define correlation in the mode picture, where we embed the $N$-fermion Hilbert space into the total Fock space of the associated modes. A tensor product structure is naturally recovered when we separate the total set of modes into two subsets. Using this structure, correlation and entanglement between the two subsets of modes can be defined and measured the same way as in between distinguishable systems, as the two subsets of modes are distinct. This process is of course not free of issues. Namely, not all observables on the local Fock spaces are physical due to the superselection rules (SSR), which alters what we perceive as correlation and entanglement \cite{bartlett2003entanglement,banuls2007entanglement}. Thus the notion of entanglement between fermionic modes must be defined with great care. The second route leads us into the particle picture, but without the $N$-fermion Hilbert space embedding into the $N$ distinguishable particle space. Inspired by concepts from resource theory, configuration states are considered ``free'' of the resource of ``particle correlation'', in analogy to uncorrelated states in the distinguishable particle setting, and their convex combinations are denoted as ``quantum-free'', in analogy to separable states. 

Having defined these new concepts of fermionic correlation and entanglement, the underdeveloped practical aspect of entanglement measure theory then becomes the biggest hindrance to direct applications in electronic systems. Many fruitful results were obtained on the formal definitions of different types of operationally meaningful entanglement measures \cite{wootters1998entanglement,vedral1998entanglement,bennett1996mixed}. In practice, however, they can rarely be computed with ease. So far no closed formula for a faithful measure of entanglement for general mixed states is known beyond two-qubit setting \cite{hill1997entanglement,wootters1998entanglement,miranowicz2008closed}, which excludes even the most primitive setting of two electronic orbitals (with a total Hilbert space isomorphic to $\mathbb{C}^4 \otimes \mathbb{C}^4$). To fill this important gap, we seek to calculate the entanglement between orbitals/sites measured by the \textit{relative entropy of entanglement} \cite{vedral1998entanglement}. 

The thesis is structured as follows. In Chapter \ref{chap:foundations} we review important concepts in quantum information theory, including theories of entanglement and its measures for distinguishable particle systems. In Chapter \ref{chap:fermion} we introduce the concepts of mode- and particle- correlation and entanglement for fermions. In Chapter \ref{chap:quantifying} we derive our main results, namely the analytic formula for mode entanglement between two sites/orbitals. These results are applied to concrete systems in Chapter \ref{chap:App}, where we fully resolve the correlation paradox in the dissociation limit in Section \ref{sec:diss}, and study the correlation and entanglement between orbitals in molecular systems in Section \ref{sec:qchem}.
  \chapter{Foundations} \label{chap:foundations}

In this chapter, we will review the fundamental aspects of the theory of correlation and entanglement. Starting from the mathematical definition of quantum states in Section \ref{sec:quant_state} and measurement in Section \ref{sec:meas}, the concepts of bipartite correlation and entanglement in distinguishable systems will be addressed in Section \ref{sec:distinguishable_ent}, as well as various ways of quantifying entanglement in Section \ref{sec:ent_meas}. Additionally, we will also briefly discuss how to identify the quantum and classical part of the correlation in a quantum state in Section \ref{sec:quantvsclass}.

\section{Quantum States} \label{sec:quant_state}

Quantum mechanics postulate that every physical system is associated with a complex vector space, a Hilbert space $\mathcal{H}$\cite{nielsen2002quantum}. The quantum states describing the system are elements (rays) in $\mathcal{H}$, which contain complete information of the system. On one hand, it is remarkable that quantum states should form such high level structure. For one, closure of the Hilbert space under linear combination immediately give rise to superposition, the key ingredient for entanglement. On the other hand, the postulate tells us nothing about \textit{how} to identify these quantum descriptors. We know that the content of a quantum state is two-fold: 1) It is the end results of a sequence of operations, \textit{a preparation}, that represents physical manipulation of the system given the initial condition. 2) It contains all information regarding the probabilistic distributions of outcomes regarding any physically implementable measurement. As many preparations can lead to the same state of the system, we should be able to talk about a quantum state as a mathematical object without referring to its preparation. This object should serve as an oracle, a map, that contains answers to all the expectation values of physical observables, and the higher order moments (leading to knowledge of variance and so on). 

The existence of a representation of quantum states as vectors in Hilbert spaces is proved by the so called Gelfand-Naimark-Segal (GNS) construction\cite{gelfand1943imbedding,segal1947irreducible}. More precisely, a representation is first established for physical observables which then allows for representation of states. Since one can add two physical observables or measure them in sequence, this give rise to an algebraic structure among them. That is, sums and products of physical observables should also be physical observables. The closure of the set of physical observables under addition and multiplication is called the \textit{algebra of observables} $\mathcal{A}$. Additionally each element $A$ in $\mathcal{A}$ is associated with an \textit{adjoint} (Hermitian conjugation) denoted as $A^\dagger$ which is also contained in $\mathcal{A}$. We also assume $\mathcal{A}$ is unital, i.e. it contains the identity element $\mathbbm{1}$ as it corresponds to doing nothing to the quantum state. A quantum state $\omega$ is a map from the algebra of observables $\mathcal{A}$ to the complex numbers $\mathbb{C}$
\begin{equation}
\begin{split}
\omega : \mathcal{A} &\longrightarrow \mathbb{C}
\\
A &\longmapsto \omega(A),
\end{split}
\end{equation}

that satisfies the following conditions:
\begin{enumerate}
\item  $\omega(\mathbbm{1}) = 1$. (Normalization)
\item  $\omega(A+B) = \omega(A) + \omega(B)$ for all $A,B \in \mathcal{A}$. (Linearity)
\item  $\omega(A^\dagger A) \geq 0$ for all $A \in \mathcal{A}$. (Positivity)
\end{enumerate}

Using this map $\omega$, we can define an \textit{inner product} $\langle A, B \rangle = \omega(A^\dagger B)$ for $A,B \in \mathcal{A}$. This allows for an identification of $\mathcal{A}$ as a Hilbert space $\mathcal{H}_{\omega}$. In case $\omega$ has non-trivial kernel, the identification takes the general form
\begin{equation}
    \begin{split}
        \pi : \mathcal{A}/\mathcal{I} &\longrightarrow \mathcal{H}_\omega,
        \\
        [A] & \longmapsto |[A]\rangle,
    \end{split}
\end{equation}
where the ideal $\mathcal{I} = \{ A \in \mathcal{A}\,|\, \omega(A^\dagger A) = 0\}$. The equivalence classes in the quotient $\mathcal{A}/\mathcal{I}$\footnote{To be precise this would be the Cauchy completion of $\mathcal{A}/\mathcal{I}$.} are denoted as $[A]$ for $A \in \mathcal{A}$ and $[A] = [B]$ if $A + \mathcal{I} = B + \mathcal{I}$. In this case the inner product becomes $\langle [A], [B]\rangle = \omega(A^\dagger B)$. We define the \textit{action} of an element of the algebra $A$ on a vector $|[B]\rangle$ as $A|[B]\rangle = |[AB]\rangle$. Then $\mathcal{A}$ can be mapped to the algebra of endomorphisms on $\mathcal{H}$ with the map $\Pi: \mathcal{A} \rightarrow \mathrm{End}(\mathcal{H}_\omega)$. Finally the state $\omega$ can be rewritten as
\begin{equation}
A \longmapsto \omega(A) = \omega(\mathbbm{1}A) = \langle [1], [A]\rangle = \langle [\mathbbm{1}], A [\mathbbm{1}]\rangle,
\end{equation}
and we identify the vector $|\Psi\rangle \equiv |[\mathbbm{1}]\rangle$ as the representation of $\omega$ in $\mathcal{H}_\omega$, and the value $\omega(A)$ are given by the expectation $\langle \Psi, A \Psi\rangle$.

The purpose of this grossly abbreviated version of GNS construction is not only to introduce the abstract definition of quantum states, but also to stress that even though our view of quantum mechanics usually revolves around the notion of Hilbert space, and more than often this is useful and constructive, the starting point of the theory is actually much earlier, namely from the algebra of observables. Later we shall see that changes in the algebra of observables can lead to ambiguities and subtleties in our way of describing quantum states, and have drastic effects on the notion of entanglement.

That being said, the abstract definition of quantum state is equivalent to the density matrix formalism, where the axiomatic conditions are translated to
\begin{enumerate}
\item  $\Tr[\rho] = 1$. (Normalization)
\item  $\Tr[\rho \, O] \geq 0$ for any positive matrix $O$. (Positivity)
\end{enumerate}
The linearity condition is omitted as the map $\Tr[\rho \cdot ] : \mathcal{A}\cong \mathcal{B}(\mathcal{H}) \rightarrow \mathbb{C}$ is explicitly linear.

Before we move on, a few remarks on some useful properties of the set of quantum states are due. The set of all quantum states $\mathcal{D}$ is convex. That is, if $\omega_1$ and $\omega_2$ is convex, so is their arbitrary convex combination $p \omega_1 + (1-p) \omega_2$ where $p \in [0,1]$. The boundary of $\mathcal{D}$ are the states $\omega$ with non-trivial kernels. These states are represented by rank-deficit density matrices. The extreme points of $\mathcal{D}$ are the states that cannot be written as a non-trivial ($p \neq 0,1$) convex combinations of other states. Such states are called \textit{pure} states, and are represented by rank-1 density matrices. These extreme points generate $\mathcal{D}$ as any quantum states can be decomposed as convex sums of pure states via spectral decomposition. In other words $\mathcal{D}$ is the convex hull of the set of pure states.

\section{Quantum Measurement} \label{sec:meas}

The notion of measurement and its relation to information play an important role in concepts of different type of correlations. In this section we will go over some basic concepts on quantum measurement, including projective measurement, and the more general case of positive operator-valued measurement (POVM).

\subsection{Projective Measurement}

Projective measurement is commonly known and used due to its connection to physical observables. Given a quantum state represented by a density operator $\rho$, one can obtain information regarding a physical observable $M$ by performing a projective measurement with respect to $M$. If all possible measurement outcomes form the set $\{m\}$, then one can write $M$ as its spectral decomposition $M = \sum_m m P_m$ where $P_m$ are projections onto the eigen-sector labeled by the eigenvalue $m$. The projective measurement associated with $M$ is characterized by the set of orthogonal projectors $\{P_m\}$ which satisfy the completeness relation $\sum_m P_m = \mathbbm{1}$ and orthogonality $P_m P_n = \delta_{m,n} P_m$.

Provided the set of projective operators $\{P_m\}$, we can calculate the probability of any measurement outcome by
\begin{equation}\label{eqn:prob}
    p(m) = \Tr[\rho P_m].
\end{equation}
The completeness relation guarantees the total probability of any possible measurement outcome occurring to be $1$. The expectation value of the associated observable $M$ can be recovered as $\langle M \rangle_\rho = \sum_m m p(m) = \Tr[\rho \sum_m m P_m] = \Tr[\rho M]$. Another key property of these projective operators are their idempotency. Applying the same projective measurement for the second time would only lead to the same result. The final quantum state after a projective measurement with outcome $m$ is an eigenstate of $P_m$.

\subsection{POVM}

Apart from projective measurement, there exists a more general type of measurement that are not charaterised by a set of orthogonal projectors, but rather a set of positive operators $\{E_m\}$ that satisfy the completeness relation $\sum_m E_m = \mathbbm{1}$, each associating with a measurement outcome $m$. This genearlised form of measurement is called positive operator-valued measurement (POVM). The previously introduced projective measurement of course falls under this umbrealla.  The probability for obtaining the result $m$ is given similarly as \eqref{eqn:prob} by $p(m) = \Tr[\rho E_m]$. 

Due to the relaxed orthogonality restriction, the number of positive operators $E_i$ describing a general measurement can exceed the dimension of the Hilbert space. That is, on the same Hilbert space a POVM can produce in general more number of outcomes than a projective measurement. This special property makes POVM sometimes more suitable for certain tasks. We shall illustrate this with the following example\footnote{This example is taken from the book \textit{Quantum Computation and Quantum information} by Micheal Nielsen and Issac Chuang \cite{nielsen2002quantum}.}.

Suppose Alice prepared two states, $|\psi_1\rangle = |0\rangle$ and $|\psi_2\rangle = (|0\rangle + |1\rangle)/\sqrt{2}$. She randomly gave one of them to Bob, who knew beforehand that the state was one of the two, and asked him to find out which state he was given. Since $|\psi_1\rangle$ and $\psi_2\rangle$ are not orthogonal they cannot be distinguished in a deterministic manner. Furthermore, Bob can never arrive at a definite conclusion on which state he has given any outcome of a projective measurement. However, with a carefully designed POVM, Bob can reliably answer whether his state is $|\psi_1\rangle$ or $|\psi_2\rangle$ when a subset of measurement outcomes 
are obtained. The optimal POVM for this discrimination task is charaterised by the following three positive operators
\begin{equation}
    \begin{split}
        E_1 &= \frac{\sqrt{2}}{1+\sqrt{2}} |1\rangle \langle 1 |,
        \\
        E_2 &= \frac{\sqrt{2}}{1+\sqrt{2}} \frac{(|0\rangle - |1\rangle)(\langle 0| - \langle 1 |)}{2},
        \\
        E_3 &= \mathbbm{1} - E_1 - E_2.
    \end{split}
\end{equation}
Because $E_1$ is orthogonal to $|\psi_1\rangle \langle \psi_1|$, the probability of obtaining result $E_1$ is given by $p_1 = |\langle 0 | 1\rangle|^2 = 0$ if Bob's state is $|\psi_1\rangle$. In case of outcome $|\psi_1\rangle$ Bob can safely conclude that his state is $|\psi_2\rangle$. Likewise, if Bob's measurement result is $E_2$, he knows for certain that his state is $|\psi_1\rangle$. The only downside to this procedure appears when Bob obtains the outcome $E_3$. The probability of obtaining this outcome is $p_3 = 1/2$ for both $|\psi_1\rangle$ and $|\psi_2\rangle$. In this case Bob cannot infer anything from the result.

Each of the positive operators comprising a POVM can be expressed as a product $E_i = K_i^\dagger K_i$, due to the positivity of $E_i$. These $K_i$'s are the Kraus operators describing this measurement. The final quantum state associating with the outcome $E_i$ is given by
\begin{equation}
    \rho^{(i)} = \frac{K^\dagger_i \rho K_i}{\Tr[K^\dagger_i \rho K_i]},
\end{equation}
occurring with the probability $p_i = \Tr[K^\dagger_i \rho K_i]$. 

As POVMs are not associated with any physical observables like projective measurements, one might wonder its importance or relevance. As a matter of fact, POVMs naturally appears in a subsystem, as the effect of a projective measurement on the total system. The connection between POVMs and projective measurements is precisely described in Naimark's dilation theorem\cite{gelfand1943imbedding}, which states that all POVM on a quantum system can be realised by a projective measurement on a larger system that contains it.
\begin{thrm}[Naimark's Theorem] \label{thrm:naimark}
For any POVM described by $\{E_m\}$ acting on $\mathcal{H}$, there exists an isometry $V: \mathcal{H} \rightarrow \mathcal{H}'$ with $\text{dim}(\mathcal{H}') \geq \text{dim}(\mathcal{H})$ and a projective measurement $\{P_m\}$ acting on $\mathcal{H}'$ such that $E_i = V^\ast P_m V$.
\end{thrm}

To prove Theorem \ref{thrm:naimark}, let $\{E_m\}_{m=1}^M$ be a POVM on system $A$. We will show that $\{E_m\}$ can be realised by the projective measurement $\{ \mathbbm{1}_A \otimes |e_m\rangle \langle e_m|\}$ acting on the composite system $\mathcal{H}_A \otimes \mathcal{H}_B$ where $\mathcal{H}_B \cong \mathbb{C}^M$. We define the isometry $V$ as
\begin{equation}
    V = \sum_m \sqrt{E_m} \otimes |e_m\rangle.
\end{equation}
Here we used the property that every positive operator has a positive square root. First we check that $V$ is indeed an isometry
\begin{equation}
    V^\ast V = \sum_{m,n} \sqrt{E_m}^\dagger \sqrt{E_n} \delta_{m,n} = \sum_m E_m = \mathbbm{1}_A.
\end{equation}
Secondly we check that $\{E_m\}$ is indeed realised by
\begin{equation}
    V^\ast P_m V = \sum_{i,j} \sqrt{E_i}^\dagger \mathbbm{1}_A \sqrt{E_j} \langle e_i|e_m\rangle \langle e_m|e_j\rangle = E_m.
\end{equation}
We would like to remark that such construction is not unique. Namely more than one projective measurements acting on a larger system can realise the same POVM on a subsystem. This is also not to suggest that a POVM cannot be implemented without a projective measurement on the total system.

\section{Bipartite Correlation and Entanglement in Distinguishable Systems\label{sec:distinguishable_ent}}

In this section, we review the concepts of correlation and entanglement in the common context of distinguishable subsystems as studied in quantum information theory. We restrict ourselves to the most important case of bipartite settings and refer the reader to Refs.~\cite{Geza09,Horo09} for an introduction into the concept of multipartite correlation and entanglement.

Let us consider in the following a quantum system which can be split into two subsystems $A$ and $B$. In the common quantum information theoretical formalism those two subsystems are assumed to be distinguishable and its states are described by density operators $\rho_{AB}$ on the total Hilbert space $\mathcal{H}_{AB}\equiv \mathcal{H}_A\otimes \mathcal{H}_B$, where $\mathcal{H}_{A/B}$ denotes the local Hilbert space of subsystem $A/B$. The underlying algebra $\mathcal{A}_{AB}$ of observables of the total system follows in the same way from the local algebras, $\mathcal{A}_{AB}\equiv \mathcal{A}_{A}\otimes \mathcal{A}_{B}$. A particularly relevant class of observables are the local ones, i.e, those of the form $O_A\otimes O_B$. As a matter of fact, they correspond to simultaneous measurements of $O_A$ on subsystem $A$ and $O_B$ on subsystem $B$. To understand the relation between both subsystems, one would be interested in understanding how the respective measurements of both local measurements are correlated. As a matter of definition, they are uncorrelated if the expectation value of $A\otimes B$ factorizes,
\begin{eqnarray}\label{ABzeroC}
\langle O_A\otimes O_B\rangle_{\rho_{AB}} &\equiv&\Tr_{AB}[\rho_{AB}\,O_A\otimes O_B] \nonumber \\
&=& \Tr_{AB}[\rho_{AB}\,O_A\otimes \mathbbm{1}_B]\, \Tr_{AB}[\rho_{AB}\,\mathbbm{1}_A\otimes O_B] \nonumber   \\
&\equiv & \Tr_{A}[\rho_{A}\,O_A]\, \Tr_{B}[\rho_{B}\,O_B] \equiv \langle O_A\rangle_{\rho_{A}} \langle O_B\rangle_{\rho_{B}}.
\end{eqnarray}
In the second line we introduced the identity operator $\mathbbm{1}_{A/B} \in \mathcal{A}_{A/B}$ and the last line gives rise to the reduced density operators $\rho_{A/B}\equiv \Tr_{B/A}[\rho_{AB}]$ of subsystems $A/B$ obtained by tracing out the complementary subsystem $B/A$. To quantify the correlation between the measurements of $O_A$ and $O_B$ one thus introduces the correlation function
\begin{equation}\label{ABcorfunc}
C_{\rho_{AB}}(O_A,O_B) \equiv  \langle O_A\otimes O_B\rangle_{\rho_{AB}}-  \langle O_A\rangle_{\rho_{A}} \langle O_B\rangle_{\rho_{B}}.
\end{equation}
Popular examples are the spin-spin or the density-density correlation functions, i.e., the local operators $A, B$ are given by some spin-component operator ${S}_\tau(\vec{x})$ or the particle density operator ${n}(\vec{x})$ at two different positions $\vec{x}_{A/B}$ in space.

The vanishing of the correlation function for a specific pair of observables $O_A,O_B$ does not imply by any means that the same will be the case for any other pair $O_A',O_B'$ of local observables. One idea would be to determine an average of the correlation function $C_{\rho_{AB}}(O_A,O_B)$  or its maximal possible value with respect to all possible choices of local observables $O_A,O_B$. At first sight, those two possible measures of total correlation seem to be very difficult (if not impossible) to calculate for a given $\rho_{AB}$. To achieve this, we define the uncorrelated states as the following.
\begin{defn}[Uncorrelated States]\label{def:uncorr}
Let $\mathcal{H}_{AB} \equiv \mathcal{H}_A \otimes \mathcal{H}_B$ be the Hilbert space and $\mathcal{A}_{AB}\equiv \mathcal{A}_A\otimes \mathcal{A}_B$ the algebra of observables of a bipartite system $A:B$, with local Hilbert spaces $\mathcal{H}_{A/B}$ and local algebras $\mathcal{A}_{A/B}$. A state $\rho_{AB}$ on $\mathcal{H}_{AB}$ is called uncorrelated, if and only if
\begin{equation}
    \langle O_A \otimes O_B \rangle_{\rho_{AB}} = \langle O_A \rangle_{\rho_A} \langle O_B \rangle_{\rho_B},
\end{equation}
for all local observables $O_A\in \mathcal{A}_A$, $O_B\in \mathcal{A}_B$.
The set of uncorrelated states is denoted by $\mathcal{D}_0$ and states $\rho_{AB}\notin \mathcal{D}_0$ are said to be \textit{correlated}.
\end{defn}
A comment is in order regarding the local algebras $\mathcal{A}_{A/B}$ that playing a crucial role in Definition \ref{def:uncorr}. In the context of distinguishable subsystems one typically assumes that $\mathcal{A}_{A/B}$ comprises all Hermitian operators on the local space $\mathcal{H}_{A/B}$. As a consequence, a state $\rho_{AB}$ is then uncorrelated if and only if it is a product state, $\rho_{AB}= \rho_A \otimes \rho_B$. This conclusion is, however, not true anymore if one would consider in Definition \ref{def:uncorr} smaller sub-algebras\cite{zanardi2001virtual}. Actually, exactly this will be necessary in fermionic quantum systems due to the number parity superselection rule \cite{SSR}.

Once the set of uncorrelated states $\mathcal{D}_0$  is specified, we then define the set of \textit{separable states} as the classical mixtures of uncorrelated states. Mathematically, separable states are convex combinations of uncorrelated states. The set of separable states $\mathcal{D}_\text{sep}$ is the convex hull $\mathrm{Conv}(\mathcal{D}_0)$, which is illustrated in Figure \ref{fig:states}. 

\begin{defn}[Separable States]\label{def:sep}
A state $\rho$ is separable if it can be written as $\rho_{AB} = \sum_i p_i \rho_i$ where $p_i \geq 0$, $\sum_i p_i=1$ and $\rho_i \in \mathcal{D}_0$. The set of separable state is $\mathcal{D}_\text{sep} = \mathrm{Conv}(\mathcal{D}_0)$. States $\rho_{AB} \notin \mathcal{D}_{sep}$ are called entangled.
\end{defn}

While the uncorrelated states are the ones that can be generated using local operators, separable states can be generated with the additional help of classical communication\cite{werner1989quantum}. Local operation and classical communication (LOCC) form an important class of actions that will later play a central role in quantifying entanglement. From Definitions \ref{def:uncorr} and \ref{def:sep}, we see that entanglement is a relative concept. Whether a state is entangled or not depends not only on the particular bipartition, but also on the local algebras of observables $\mathcal{A}_{A/B}$. 

\begin{figure}[ht]
    \centering
    \includegraphics[scale=0.4]{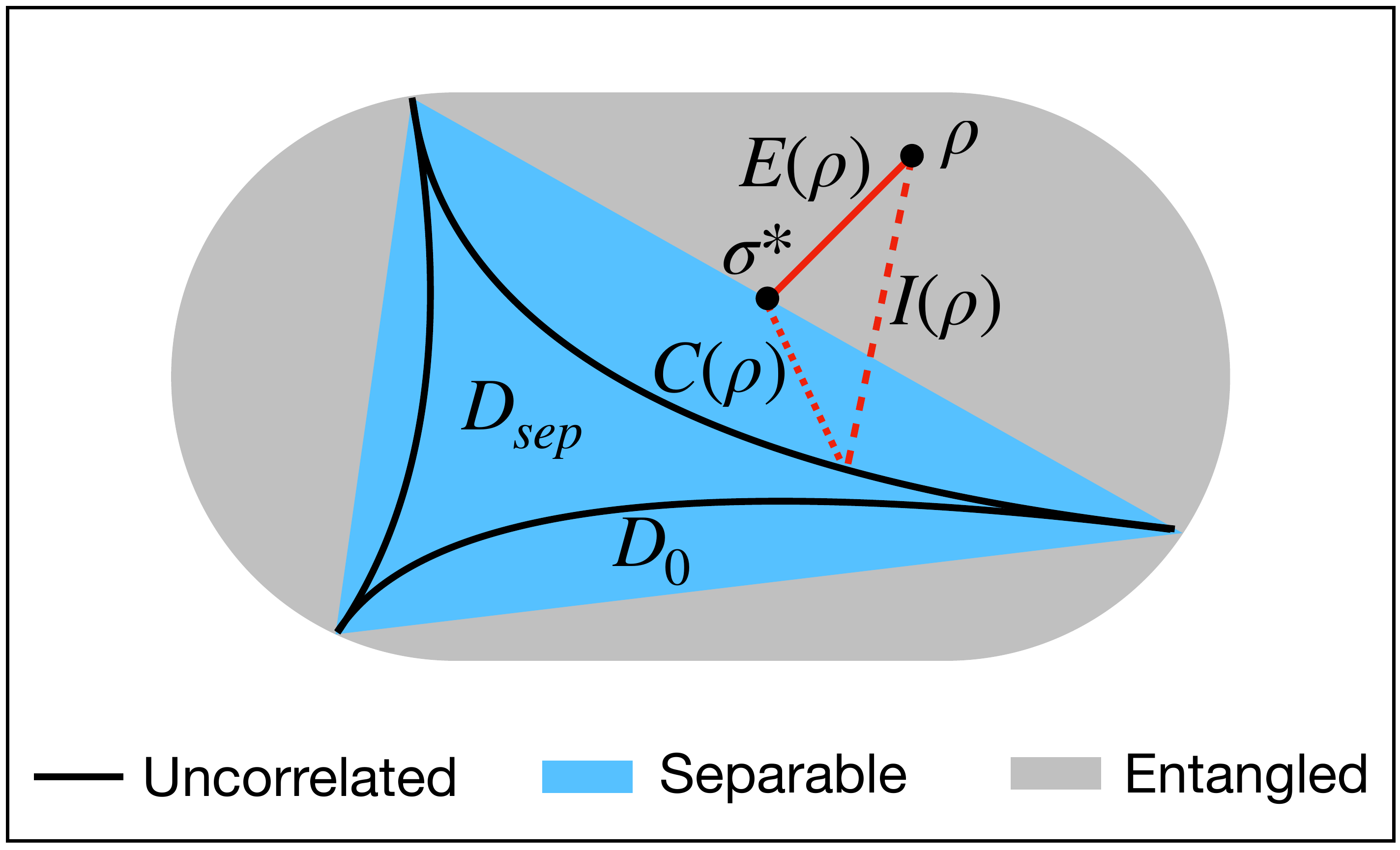}
    \caption{Schematic illustration of the space of quantum states, including the uncorrelated ($\mathcal{D}_0$, black curve) and separable states ($\mathcal{D}_{sep}$, blue). The total correlation $I$ (red dashed) and entanglement $E$ (red solid) of a state $\rho$ is its distances to $\mathcal{D}_0$ and $\mathcal{D}_{sep}$, respectively, measured by the relative entropy. The classical correlation $C$ is the distance from the closest separable state $\sigma^\ast$ to the closest product state (red dotted).}
    \label{fig:states}
\end{figure}

\section{Entanglement Detection}

In the previous section we defined the uncorrelated and separable states for arbitrary local algebras of observables. In this section we will focus on the case where the local algebras of observables are generated by all local Hermitian operators. That is, $\mathcal{A}_{A/B} \cong \mathcal{B}(\mathcal{H}_{A/B})$. In this case the uncorrelated states are exactly the product states of the form $\rho_{AB} = \rho_A \otimes \rho_B$, and the separable states can be written as their convex combinations $\rho_{AB} = \sum_i p_i \rho_A^{(i)} \otimes \rho_B^{(i)}$.

Before we can measure the correlation/entanglement in a state $\rho_{AB}$, the first essential question is, how do we tell whether $\rho_{AB}$ is correlated/entangled or not? The criterion for correlation is rather simple. One can easily detect correlation when a state $\rho_{AB}$ violates the condition
\begin{equation}
    \rho_{AB} = \Tr_B[\rho_{AB}] \otimes \Tr_A[\rho_{AB}].
\end{equation}
Detecting entanglement on the other hand is not so straightforward. Working only with Definition \ref{def:sep}, one would have to compare the state $\rho_{AB}$ with all (infinite) convex combinations of product states. The definition itself as a criterion is only conclusive when a convex decomposition of $\rho_{AB}$ into uncorrelated states is exactly found. The impracticality of the original definition necessitates the need of more practical entanglement/separability criteria. 

Partial transposition is perhaps the best known and often first applied operational separability criterion. Also named after its discoverers as the Peres-Horodecki criterion\cite{peres1996separability,horodecki1997separability}, it utilises the fact that separable states are invariant under partial transposition (on the $B$ subsystem), defined as 
\begin{equation}
\begin{split}
(\cdot)^{T_B} : \mathcal{B}(\mathcal{H}_A) \otimes \mathcal{B}(\mathcal{H}_B) &\longrightarrow \mathcal{B}(\mathcal{H}_A \otimes \mathcal{H}_B)
\\
O_A \otimes O_B &\longmapsto O_A \otimes O_B^{T}.
\end{split}
\end{equation}
The domain can be extended to $\mathcal{B}(\mathcal{H}_A \otimes \mathcal{H}_B)$ by linearity. The partial transposition acts as the identity on $\mathcal{D}_\text{sep}$ as
\begin{equation}
\rho_{AB}^{T_B} = \sum_i p_i \rho_A^{(i)} \otimes (\rho_B^{(i)})^T = \rho_{AB}, \quad \forall \rho_{AB} \in \mathcal{D}_\text{sep}.
\end{equation}
The image of $\mathcal{D}_\text{sep}$ under partial transposition consists of positive operators (or more precisely, quantum states). Therefore the partial transposition (on either subsystem) of a separable state must have a positive spectrum. Although in general the converse is not true, namely positive partial transposition does not guarantee separability, a state is conclusively entangled if the spectrum of its partial transposition contains a negative eigenvalue (see Figure \ref{fig:PPT}). This is truly remarkable given the easy implementation. Moreover, this criterion is even both sufficient and necessary when the dimensions of the Hilbert spaces are lower than or equal to $2\times 3$. 

\begin{figure}[h]
\centering
\includegraphics[scale=0.3]{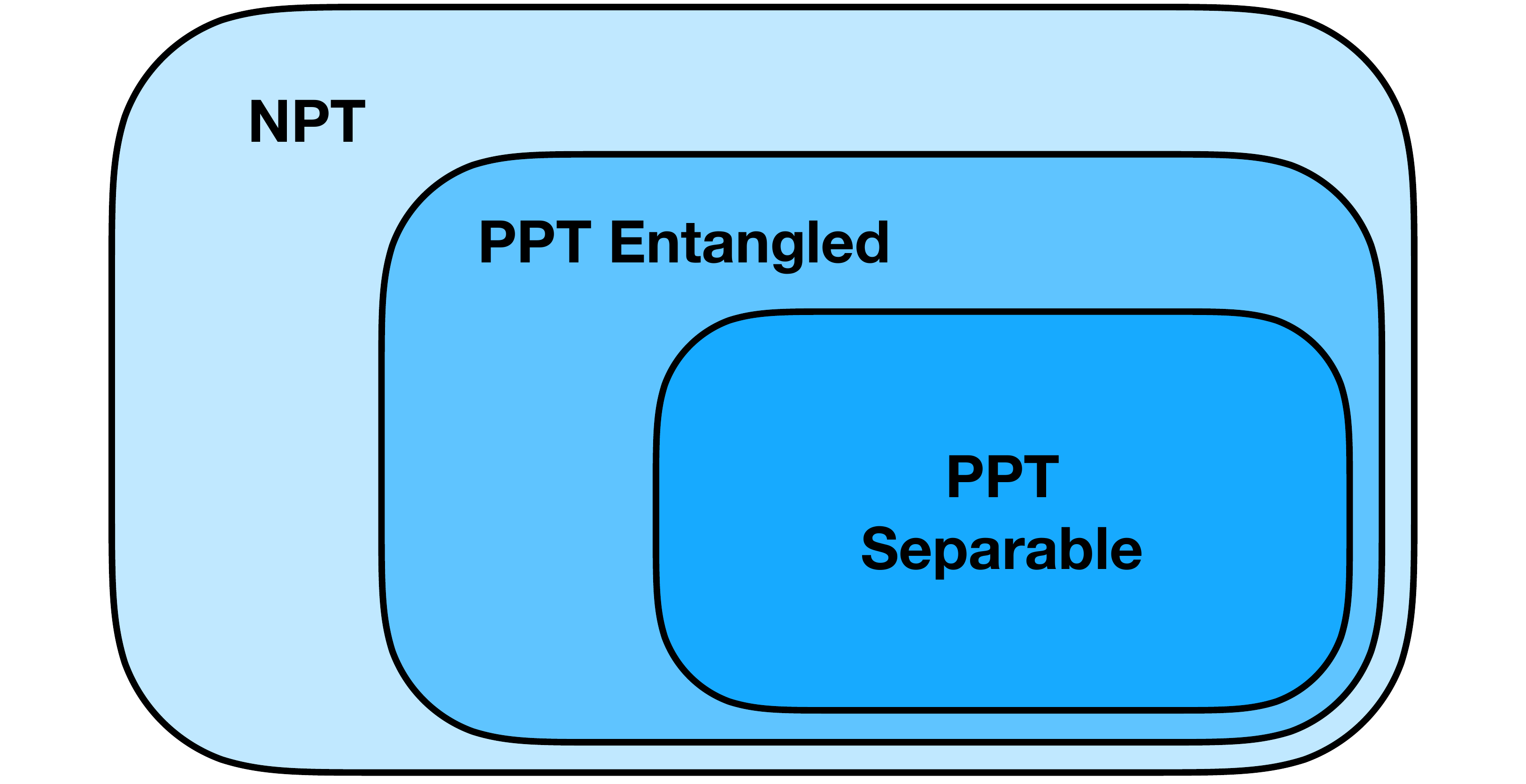}
\caption{Relations among separable states, states with positive partial transposition (PPT) and states with negative partial transposition (NPT). }
\label{fig:PPT}
\end{figure}

There exist other operational ways of detecting entanglement, even though rarely any methods can conclusively detect both separability and entanglement (i.e. a strict binary classification of separable and entangled states), such as matrix realignment\cite{chen2002matrix} and covariance matrix\cite{guhne2007covariance, gittsovich2008unifying}.

These aforementioned separability/entanglement criterion are called \textit{operational}, because they can be implemented independently of the quantum state of interest. There are, however, \textit{non-operational} way of determining the separability of a given quantum state. Entanglement witness, for example, falls under this category. For any entangled state $\rho$, the convexity of $\mathcal{D}_\text{sep}$ allows the existence of a separating hyperplane in the space of density matrices between $\rho$ and $\mathcal{D}_\text{sep}$, mathematically representing a Hermitian operator $W_\rho$ that satisfies\cite{horodecki1996}
\begin{equation}
\Tr[\rho W_\rho] < 0, \quad \mathrm{and\:} \Tr[\sigma W_\rho] \geq 0, \:\forall \sigma \in \mathcal{D}_\text{sep}.
\end{equation}This Hermitian operator $W_\rho$ is therefore called an entanglement witness, as entangled states (not all) are detected by negative expectation values. A prime example of entanglement witness is perhaps the earliest attempt at entanglement detection, namely the Bell inequalities\cite{bell1964einstein,clauser1969proposed}. As Bell inequalities are not violated by any separable states, they are regarded as non-optimal entanglement witnesses\cite{hyllus2005relations}. Although a general entanglement witness can detect more than one entangled states, the analytic form of an effective witness $W_\rho$ for an entangled state $\rho$ has to be determined in a case by case basis. 

Giving a full review over the subject of separability/entanglement criteria would be far beyond the scope of this thesis (for review papers please see Ref.~\cite{plenio2014introduction,bruss2002characterizing,Horo09,terhal2002detecting}). We make one additional remark that the separability problem is actually NP-hard\cite{gurvits2003classical}. This is precisely why separability/entanglement criteria that can be easily or quickly implemented often sacrifice some level of competency. Doherty \textit{et al.}\cite{doherty2002distinguishing} proposed a remarkable hierarchy of separability criteria using the existence of symmetric extension for separable states, that can detect any entangled states after finitely many steps (the number of which depends on the state). If the state is separable however, one can only confirm its separability after infinitely amount of time.   

\section{Entanglement Measures} \label{sec:ent_meas}

In the last section we discussed several methods of entanglement detection, which can actually be seen as a special case of the subject of this section: entanglement measures. As opposed to an entanglement detection method which can be summarised as a function on the space of density matrices with binary outputs, $0$ for separable states and $1$ for entangled ones, entanglement measure goes one step beyond, and assigns different positive values to entangled states.

The initial motivation for entanglement measure was closely linked to a few of the earliest quantum information protocols. In the early 90's researchers found that a Bell state describing two distinguishable spin-$\frac{1}{2}$ particles
\begin{equation}
|\Psi_-\rangle = \frac{1}{\sqrt{2}} (|\!\uparrow\rangle \otimes |\!\downarrow\rangle - |\!\downarrow\rangle \otimes |\!\uparrow\rangle)
\end{equation}
can assist in performing novel information processing tasks involving two distant parties otherwise impossible under the constraint of local operations and classical communication (LOCC), such as teleporting an unknown quantum state of another spin-$\frac{1}{2}$ particle\cite{bennett1993teleporting} and communicating two bits of information by sending through only one spin-$\frac{1}{2}$ particle\cite{bennett1992communication,bennett1996purification}. The entanglement in the state $|\Psi_-\rangle$ is a non-local resource that allows one to overcome the constraint of LOCC, of which the precise quantification is highly instructive.

Formally, an entanglement measure is a function from the space of quantum states to the set of non-negative real numbers $E: \mathcal{D} \rightarrow [0,\infty)$. Additionally, $E$ has to fulfill the following conditions\cite{vedral1997quantifying,plenio2014introduction}:
\begin{enumerate}
\item  \textit{$E(\rho) = 0$ if and only if $\rho$ is separable.} Separable states are the ones that can be generated using LOCC only, and therefore contain no entanglement resource. Sometimes this condition is relaxed to the one that only requires $E(\rho) \geq 0$ for all entangled states $\rho$. An entanglement measure that satisfies the original condition is called \textit{faithful} as it reveals all entanglement.
\item  \textit{$E$ is invariant under local unitary transformations, i.e. $E(\rho) = E(U_A \otimes U_B \rho U_A \otimes U_B)$. }This correspond to a change of local basis which does not affect the results of any local measurements, and leaves the entanglement unchanged.
\item  \textit{$E$ does not increase under LOCC.}  Local operations and classical communication cannot turn separable states into entangled ones, nor can they increase the entanglement resource in a quantum state. For this $E$ is also called an entanglement monotone.
\end{enumerate}

For bipartite pure states, the unique\cite{popescu1997thermodynamics} entanglement measure is the well known entanglement entropy, which is von Neumann entropy on the reduced states
\begin{equation}
E(|\Psi\rangle\langle\Psi|) = S(\rho_{A/B}), \quad \text{where } \rho_{A/B} = \Tr_{B/A} [|\Psi\rangle\langle\Psi|], \label{eqn:vonNeumann}
\end{equation}
and $S(\sigma) \equiv -\Tr[\sigma \log(\sigma)]$. Quantifying the entanglement in a bipartite pure state is equivalent to quantifying the mixedness of its reduced states. In this case choosing subsystem $A$ or $B$ will yield the same result as the respective reduced states are isospectral due to the existence of Schmidt decomposition for bipartite pure states\cite{nielsen2002quantum}. One can easily check that this measure for pure states fulfill all the conditions above. For mixed states, however, \ref{eqn:vonNeumann} is not an appropriate entanglement measure. For one, condition 1 is violated as $E(\rho_A\otimes \rho_B) = S(\rho_A) > 0$ when $\rho_A$
is mixed, even though $\rho_A \otimes \rho_B$ is a product state and contains no correlation at all. 

Quantifying entanglement for mixed states is in general difficult. But before we introduce a suitable general entanglement measure, let us first take a detour and talk about quantifying total correlation. Recall that a state is uncorrelated if it is a product state $\rho_{AB} = \rho_A \otimes \rho_B$. The amount of information contained in the total system but not yet in the two subsystems is the total correlation, quantified by the mutual information\cite{Lindblad73,vedral1997quantifying} defined as
\begin{equation}
I(\rho) = S(\rho_A) + S(\rho_B) - S(\rho). \label{eqn:MI}
\end{equation}Interestingly, the mutual information has a pleasing geometric interpretation. Namely, it coincides with the minimum distance from a quantum state $\rho_{AB}$ to any uncorrelated state measured by the quantum relative entropy. Moreover, the closest uncorrelated state to $\rho_{AB}$ is none other than the tensor product of the two reduced states $\rho_A \otimes \rho_B$. To summarised these points,
\begin{equation}
I(\rho) = \min_{\sigma \in \mathcal{D}_0} S(\rho_{AB}||\sigma) = S(\rho || \rho_A \otimes \rho_B),
\end{equation}
where $S(\cdot||\cdot)$ is the quantum relative entropy defined as
\begin{equation}
S(\rho||\sigma) = \Tr[\rho(\log(\rho)-\log(\sigma))].
\end{equation}
To see this, we derive the following inequality starting from $S(\rho||\sigma_A\otimes \sigma_B)$ where $\sigma_A$ and $\sigma_B$ are general states on the local subsystems.
\begin{equation}\label{eqn:MIproof}
\begin{split}
S(\rho||\sigma_A \otimes \sigma_B) =& \Tr[\rho\log(\rho)] - \Tr[\rho\log(\sigma_A \otimes \sigma_B)]
\\
=& \Tr[\rho\log(\rho)] - \Tr[\rho\log((\sigma_A \otimes \mathbbm{1})(\mathbbm{1}\otimes \sigma_B))] 
 \\
=& \Tr[\rho\log(\rho)] - \Tr[\rho_A \log(\sigma_A)] - \Tr[\rho_B \log(\sigma_B)] 
 \\
=& S(\rho||\rho_A\otimes \rho_B) + S(\rho_A||\sigma_A) + S(\rho_B || \sigma_B) 
\\
\geq &S(\rho||\rho_A\otimes \rho_B) = I(\rho).
\end{split}
\end{equation}Here we used the non-negativity of the relative entropy. \eqref{eqn:MIproof} becomes an equality when $\sigma_{A}=\rho_A$ and $\sigma_B = \rho_B$. Inspired by this geometric property of mutual information as a measure for total correlation, we can quantify entanglement in a similar manner, by defining it as the minimum distance from a quantum state $\rho$ to the set of separable states $\mathcal{D}_\text{sep}$ measured by the relative entropy\cite{vedral1997quantifying},
\begin{equation}
E(\rho) = \min_{\sigma \in \mathcal{D}_\text{sep}} S(\rho||\sigma). \label{eqn:rel_ent}
\end{equation}This quantity is therefore called the relative entropy of entanglement.  This entanglement measure has several remarkable properties. First of all, all conditions for a proper entanglement measures are satisfied. The relative entropy of entanglement is faithful, invariant under local unitary transformation and non-increasing under LOCC\cite{vedral1998entanglement}. Secondly, for pure states it reduces to the von Neumann entropy of the reduced states\cite{vedral1998entanglement}, just like in \eqref{eqn:vonNeumann}. Thirdly, \eqref{eqn:rel_ent} puts entanglement and total correlation on an equal footing, allowing us to identify the entanglement as ``a part of'' the total correlation (see Figure \ref{fig:states}). As $\mathcal{D}_\text{sep} \subseteq \mathcal{D}_0$, the entanglement in a state $\rho$ can never exceed its mutual information
\begin{equation}
E(\rho) = \min_{\rho \in \mathcal{D}_\text{sep}} S(\rho||\sigma) \leq \min_{\sigma \in \mathcal{D}_0} S(\rho||\sigma) = I(\rho).
\end{equation} This unifying perspective on correlation and entanglement is our major reason for picking \eqref{eqn:rel_ent} as our first choice for an entanglement measure. 

This choice is of course not a unique one. There exist different kinds of entanglement measure that are well suited for different purposes. For example, using the convex rule construction, the entanglement of formation\cite{bennett1996mixed} is defined as
\begin{equation}
E_F (\rho) = \min_{\{p_i, |\Psi_i\rangle\}} p_i E(|\Psi_i\rangle \langle \Psi_i|)
\end{equation}where the minimum is taken over all pure state decompositions of $\rho$, and the entanglement of the pure states are measured by the entanglement entropy \eqref{eqn:vonNeumann}. $E_F$ is closely related to the entanglement cost\cite{bennett1996mixed}
\begin{equation}
E_C(\rho) = \inf_\Phi \{ r \, | \, \lim_{n \rightarrow \infty} \Tr[\rho^{\otimes n} - \Phi((|\Psi_-\rangle\langle\Psi_-|)^{\otimes rn})] \},
\end{equation}
where $\Phi$ represents an arbitrary LOCC operation. This quantity $E_c$ is the minimum rate $r$ of converting $nr$ copies of Bell states into $n$ copies of $\rho$ through LOCC. It tells us how expensive it is in terms of the currency of Bell states to create an entangled state $\rho$ via LOCC. Even though it remains an open question, it is strongly believed that $E_F$ is equal to $E_C$, which would significantly simplify the task of calculating the entanglement cost. In fact, it is already proven that at least in the asymptotic limit the equality holds\cite{hayden2001asymptotic}, namely
\begin{equation}
E_F^\infty(\rho) \equiv \lim_{n \rightarrow \infty} \frac{E_F(\rho^{\otimes n})}{n} = E_C. 
\end{equation}For this reason, the calculation of entanglement measure has always been of great interest. However, so far $E_F$ can only be analytically calculated for a general mixed state $\rho$ in a two-qubit system, by relating $E_F$ to Wootter's concurrence\cite{wootters1998entanglement}.

Another example goes in the opposite direction, and asks how many Bell pairs can one extract from an many copies of entangled state $\rho$. This quantity is called the entanglement of distillation\cite{bennett1996concentrating,rains1999rigorous}
\begin{equation}
E_D = \inf_\Phi \{ r \, | \, \lim_{n \rightarrow \infty} \Tr[\Phi(\rho^{\otimes n}) - (|\Psi_-\rangle\langle\Psi_-|)^{\otimes rn}] \},
\end{equation}where the LOCC is applied to $n$ copies of $\rho$. $E_D$ is a quantity of great relevance, especially in realistic experimental setups. Suppose Alice and Bob wish to perform a quantum teleportation protocol. For this they need to share between them beforehand several copies ($n$) of Bell pairs $|\psi_-\rangle$, depending on the size of information. In real life, however, decoherence may turn these $n$ copies of $|\psi_\rangle$ into $n$ copies a mixed state $\rho$ even before the distribution is finished, due to interaction of the qubit with the environment. What Alice and Bob must do is turn these $n$ copies of $\rho$ into $m$ copies of ``concentrated'' entangled states, namely the Bell state $|\Psi_-\rangle$, by performing local operations on their qubits and send each other classical information. This then constitutes an entanglement distillation protocol. The ratio $r = \lim_{n\rightarrow \infty} m / n$ is precisely entanglement of distillation $E_D$. On the other hand, the entanglement of distillation is also a challenging quantity to compute. There is no known analytic formula of $E_D$ for general mixed states.

Entanglement of distillation $E_D$, and its mirroring quantity $E_C$ satisfy $E_D \leq E_C$ for all quantum states\cite{horodecki2000limits}. This striking result has a rather intuitive interpretation. Much like energy, the process of distilling entanglement (using energy to do work) is ``dissipative''. In fact, this is best examplified by the existence of bound entangled states\cite{horodecki1998mixed}. Bound entangled states are inseparable states that cannot be distilled into any copies of Bell states, yet still having finite entanglement cost\cite{vidal2001irreversibility}. Furthermore, $E_C$ and $E_D$ serve as the upper and lower bounds , respectively, for all normalised (scaled to maximum of $1$) entanglement measures that satisfy convexity, in addition to the conditions mentioned above\footnote{For asymptotic quantities additional conditions apply\cite{horodecki2000limits}.}\cite{horodecki2000limits}. For example, if $\mathcal{H}_A \otimes \mathcal{H}_B \cong \mathbb{C}^2 \otimes \mathbb{C}^2$, and we define the relative entropy with $\log_2$ instead of the natural logarithm for normalisation, the relative entropy of entanglement $E$ in \eqref{eqn:rel_ent} is bounded by
\begin{eqnarray}
E_D(\rho) \leq E(\rho) \leq E_C(\rho)
\end{eqnarray}
for all states $\rho$.

The last example we would like to mention is the logarithmic negativity
\begin{equation}
N(\rho) = \log(||\rho^{T_B}||), \label{eqn:negativity}
\end{equation}which quantifies the extend to which $\rho$ violates the Peres-Horodecki separability criterion\cite{peres1996separability,horodecki1997separability}. If the partial transposition $\rho^{T_B}$ has negative eigenvalues, then the sum of the absolute values of its eigenvalues would deviate from $1$. However, \eqref{eqn:negativity} is not faithful, in the sense that entangled state with positive partial transposition would be deemed unentangled by this measure. Moreover, unlike the relative entropy of entanglement $E$ and entanglement of formation $E_F$, \eqref{eqn:negativity} does not reduce to the von Neumann entropy \eqref{eqn:vonNeumann} for pure states. It is also not linked to any operational meaning. That being said, the logarithmic negativity distinguishes itself from other entanglement measures for its easy implementation.

\section{Quantum vs Classical Correlation} \label{sec:quantvsclass}

After establishing the relative entropy of entanglement as our first choice of entanglement measure, and thus putting entanglement and correlation at equal footing, a natural question to ask is whether one can separate it from the total correlation and identify the classical part of correlation at the same time. For a very long time, entanglement was thought to be responsible for all the correlation of quantum mechanical origin. But the discovery of quantum discord reveals that quantum correlation can still be present even if a state is separable\cite{ollivier2001quantum}. But due to the information processing value of entanglement, it remains important to compare the amount of entanglement and classical correlation in the system. In this section we outline two schemes of separating the total correlation into quantum and classical parts, the first one being the division into entanglement and classical correlation, and the second scheme referring to the concept of quantum discord.

Recall that in the geometric picture in Figure \ref{fig:states}, the total correlation and entanglement in a state $\rho$ is defined as the distances to its closest uncorrelated state $\rho_A \otimes \rho_B$ and its closest separable state $\sigma^\ast$, respectively. The geometric classical correlation can be defined as the remaining distance from $\sigma^\ast$ to $\rho_A\otimes \rho_B$, measured by the quantum relative entropy\cite{henderson2001classical} (also see Figure \ref{fig:states})
\begin{equation}
    C(\rho) = S(\sigma^\ast || \rho_A \otimes \rho_B). \label{eqn:class_geo}
\end{equation}
For a Bell state $|\Psi_-\rangle = (|00\rangle - |11\rangle)/\sqrt{2}$, the closest uncorrelated state is the maximally mixed state $\sum_{i,j=0}^1|ij\rangle\langle ij|/4$, and the closest separable state can be calculated as the mixture $\sigma^\ast = |00\rangle\langle00|/2 + |11\rangle\langle11|/2$. From these states we calculate the total correlation, relative entropy of entanglement and classical correlation to be $2\log(2)$, $\log(2)$ and $\log(2)$, respectively. In this case the relation $I(\rho) = C(\rho) + E(\rho)$ is fulfilled, although entanglement and classical correlation defined this way do not sum up to be the total correlation for general mixed states. One of the reasons for this is that for pure states the notion of entanglement and quantum correlation coincide, whereas for mixed states they are distinct concepts. 

Although entanglement is known for being responsible for violating Bell inequalities and enhancing performance in information tasks beyond classical capability, there has been reports of quantum nonlocality beyond entanglement\cite{bennett1999quantum}. This nonlocality can even exists in separable states and be harnessed for quantum speed-up\cite{lanyon2008experimental}. 

This brings us to our second way of dividing the total correlation, namely into quantum discord $Q$ and classical correlation $C'$ (the prime is to distinguish the following measure from the classical correlation defined in \eqref{eqn:class_geo} associated with the geometric picture in Figure \ref{fig:states}). The original definition of quantum discord, first formulated by Olliver and Zurek\cite{ollivier2001quantum}, is defined as the difference
\begin{equation}
    D^{(A)}(\rho) = I(\rho)-C'_\text{op}{}^{(A)}(\rho) \label{eqn:disc_1}
\end{equation}
between the total correlation and the operational classical correlation\cite{henderson2001classical} $C'_\text{op}$ defined as\footnote{In the original paper by Olliver and Zurek the quantity $C'_\text{op}$ was implicitly defined as the maximum taken over all \textit{projective measurements} on the subsystem\cite{ollivier2001quantum}. The maximum was later generalised to be taken over all POVMs on the subsystem by Henderson and Vedral\cite{henderson2001classical}.}
\begin{equation}
    C'_\text{op}{}^{(A)}(\rho) = \max_{B_i^\dagger B_i} S(\rho_A) - p_i S(\rho_A^{(i)}), \label{eqn:henderson}
\end{equation}where $\rho_A$ is the reduced state on subsystem $A$ and
\begin{eqnarray}
p_i = \Tr[B^\dagger_i \rho B_i], \quad \rho_A^{(i)} = \frac{1}{p_i} \Tr_B[\rho_A^{(i)}].
\end{eqnarray}The maximum is taken over all possible POVMs acting on subsystem $B$. It is possible to define $C'_\text{op}{}^{(B)}$ by swapping the two subsystems. To understand why \eqref{eqn:disc_1} quantifies the amount of quantum correlation in a generic quantum state $\rho$, we must first discuss the meaning of its classical counterpart \eqref{eqn:henderson}. As pointed out by Olliver and Zurek in their seminal paper\cite{ollivier2001quantum}, classical information can be obtained locally from a quantum state without disturbing it. If a state contains only classical correlation between the two parties, then one can perform a measurement on only one of the subsystems without altering the state as a whole. This is most certainly not true for arbitrary quantum state, most drastically examplified by the Bell states, e.g. $|\Psi_+\rangle$. A single measurement on the first qubit in the reference basis immediately causes the total state to collapsed into a product state with the same bit values. As it turns out, even some separable states can exhibit this nonlocal effect. For example, the following separable state\cite{datta2008quantum}
\begin{equation}
    \rho = \frac{1}{4}(|\psi_+\rangle\langle \psi_+| \otimes |0\rangle\langle 0| + |\psi_-\rangle \langle \psi_-| + |0\rangle\langle 0| \otimes |\psi_+\rangle\langle \psi_+| + |1\rangle\langle 1| \otimes |\psi_-\rangle\langle \psi_-|)
\end{equation}
is altered by any local measurement (projective or POVM). In fact, in this spirit, the only type of state with non-zero discord ($D^{(A)}(\rho) = D^{(B)}(\rho) = 0$) are those with the following property
\begin{equation}
    \sum_k P_k^{(A)}\otimes \mathbbm{1} \, \rho \, P_k^{(A)}\otimes \mathbbm{1} = \sum_k \mathbbm{1} \otimes P_k^{(B)} \, \rho \, \mathbbm{1} \otimes P_k^{(B)} = \rho \label{eqn:disc_cond}
\end{equation}
for some projective measurements $\{P_k^{(A/B)}\}$ on the local subsystems $A/B$\cite{dakic2010necessary}.

In general $C'_\text{op}{}^{(A)}$ and $C'_\text{op}{}^{(B)}$, however, do not coincide, when $\rho_A \neq \rho_B$. Due to this asymmetry, the quantum discord \eqref{eqn:disc_1} is also asymmetric when the two subsystems are swapped. This is somewhat undesirable, despite the operationally meaningful construction. Moreover, the maximisation over all local POVMs is also difficult to implement in general. \eqref{eqn:disc_cond} provides a strict criterion for detecting non-zero discord, but in terms of the measure \eqref{eqn:disc_1} closed formula is only available for two-qubit systems\cite{dakic2010necessary}.

\begin{algorithm}[t] \label{alg:disc_1}
\SetAlgoLined
\KwResult{$D(\rho)$ and the closest classical state $\sigma_\textrm{cl}^\ast$ to $\rho$}
Initialise $\min$;
\\
$n = 0$;
\\
\While{$n < N$}{
\:\: Generate two set of orthonormalized states $\{|a\rangle\}$ and $\{|b\rangle\}$;
\\
\:\: Construct candidate $\sigma_\text{cl} = \sum_{a,b} \mu_{ab} |a\rangle\langle a| \otimes |b\rangle\langle b|$;
\\
\:\: Evaluate $D = S(\rho||\sigma_\text{cl})$;
\\
\:\: Update $\min = D$ if $D < \min$.
\\
\:\: $n = n + 1$;
 }
 Output $D(\rho) = \min$.
 \caption{Calculating discord (direct search)}
\end{algorithm}

For this reason, another closely related quantity was proposed as an alternative definition of quantum discord, which relies on the concepts of classical states, inspired by the condition \eqref{eqn:disc_cond}, of the form\cite{modi2010unified}
\begin{equation}
    \sigma_\text{cl} = \sum_{a,b} \mu_{ab} |a\rangle\langle a| \otimes |b\rangle\langle b| \label{eqn:closest_cl}
\end{equation}
where $\{|a\rangle\}$ and $\{|b\rangle\}$ are local orthonormal bases on the subsystems. In other words, the eigenstates of a classical state are mutually orthogonal and locally distinguishable product states\cite{horodecki2003local}. Such states no nonlocal properties. The set of classical states is denoted as $\mathcal{D}_\text{cl}$. The quantum discord is defined as the minimum distance to the set of classical states
\begin{equation}
    Q(\rho) = \min_{\sigma \in \mathcal{D}_\text{cl}} S(\rho||\sigma). \label{eqn:disc}
\end{equation} The set of classical states $\mathcal{D}_\text{cl}$ is clearly a proper subset of $\mathcal{D}_\text{sep}$. Therefore the quantum discord in a state $\rho$ is no smaller than its entanglement. Although it also enjoys the property of local unitary invariance, $\mathcal{D}_\text{cl}$ lacks convexity, in contrast to $\mathcal{D}_\text{sep}$. The accompanying classical correlation $C'$ is then defined as the mutual information of the closest classical state $\sigma_\text{cl}^\ast$
\begin{equation}
C'(\rho) = I(\sigma_\text{cl}^\ast).
\end{equation}
It was shown in Ref.\cite{modi2010unified} that the eigenvalues for the closest classical state $\sigma_\text{cl}^\ast$ to $\rho$ are given by
\begin{equation}
    \mu_{ab} = (\langle a |\otimes \langle b|) \rho (|a\rangle \otimes |b\rangle). \label{eqn:disc_eig}
\end{equation}
\eqref{eqn:disc_eig} effectively transformed the task of finding the closest classical state into the task of finding the optimal local bases. This computational feasibility is also why we prefer this particular definition of quantum discord. Here we propose two efficient ways of calculating \eqref{eqn:disc} for general mixed states of arbitrary dimensions in Algorithm \ref{alg:disc_1} and \ref{alg:disc_2}.

Algorithm \ref{alg:disc_1} search for the closest classical state by brute force, by sampling all possible combinations of local eigenstates $\{|a/b\rangle\}$. The eigenvalues of the candidates of the closest classical state is given by \eqref{eqn:disc_eig}. Algorithm \ref{alg:disc_1} is easy to implement and sufficiently fast in the case of two-qubit systems, and accurate with only 100 iterations. When we extend to higher dimensional subsystems however, this method becomes expensive. We can improve Algorithm \ref{alg:disc_1} by utilising the fact that a set of orthonormal basis states are the column vectors of a unitary matrix $U$. It would then suffice to find the optimal unitary matrix. With this in mind, we would like to find a faster path from a random element the group of unitary matrices, to the optimal one, $U$. Inspired by the Markov Chain Monte Carlo (MCMC) methods, we propose the following random walk solution.

\begin{algorithm}[t] \label{alg:disc_2}
\SetAlgoLined
\KwResult{$D(\rho)$ and the closest classical state $\sigma_\textrm{cl}^\ast$ to $\rho$}
Initialise $U$ as random unitary, construct $\sigma_{\textrm{cl}}^{(0)}$;
\\
$D^{(0)}(\rho) = S(\rho||\sigma_{\textrm{cl}}^{(0)})$;
\\
$n = 0$;
\\
\While{$n < N$}{
\:\: $n = n + 1$;
\\
\:\: Generate ramdon Hermitian matrix $H$;
\\
\:\: $U' = \exp( i \eta H) U$;
\\
\:\: Take columns of $U'$ as $\{|a/b\rangle\}$;
\\
\:\: Construct candidate $\sigma_\text{cl}^{(n)} =\sum_{a,b} \mu_{ab} |a\rangle\langle a| \otimes |b\rangle\langle b|$;
\\
\:\: Evaluate $D^{(n)} = S(\rho||\sigma_\text{cl}^{(n)})$;
\\
\:\: Randomise uniformly $p \in [0,1]$;
\\
\:\: \textbf{if} $p < \exp(-\beta(D^{(1)} - D^{(0)}) )$
\\
\quad \quad $U = U'$;
\\
\:\: \textbf{else}
\\
\quad \quad $D^{(n)} = D^{(n-1)}$;
\\
\:\: \textbf{end};}
Output $D = D^{(N)}$, $\sigma_\textrm{cl}^\ast = \sigma_\textrm{cl}^{(N)}$.
\caption{Calculating discord (MCMC)}
\end{algorithm}

We first initiate $U$ as a random unitary, from which we can calculate $D^{(0)}(\rho) = S(\rho||\sigma_{\textrm{cl}}^{(0)})$ where $\sigma_{\textrm{cl}}^{(0)}$ is constructed using Eq.\eqref{eqn:closest_cl} and the column vectors of $U$ as the proposed local eigenstates of $\sigma^{(0)}_\text{cl}$. We randomly generate a ``small'' unitary matrix $V = \exp ( i \eta H )$ where $H$ is a random Hermitian matrix and $\eta$ a small hyper-parameter representing a step size. We then construct a new unitary $U' = VU$ and calculate $D^{(1)} = S(\rho||\sigma_{\textrm{cl}}^{(1)})$ using the column vectors of $U'$ as eigenstates of the new classical state $\sigma_\text{cl}^{(1)}$. $U$ is updated as $U = U'$ if $D^{(1)} < D^{(0)}$, or with the probability given by $p = \exp(-\beta(D^{(1)} - D^{(0)})$ if $D^{(1)} \geq D^{(0)}$, where $\beta > 0$ is an adjustable hyper-parameter. The algorithm then repeats for sufficient ($N$) steps and converges to a closest classical state $\rho_\text{cl}^\ast$. This random walk favours lower values of the discord $D$ at each step while maintaining some level of stochasticity, and hence it is much more efficient than a random search as described in Algorithm \ref{alg:disc_1}. We summarise this approach in Algorithm \ref{alg:disc_2}.

\begin{figure}[h]
    \centering
    \includegraphics[scale=0.4]{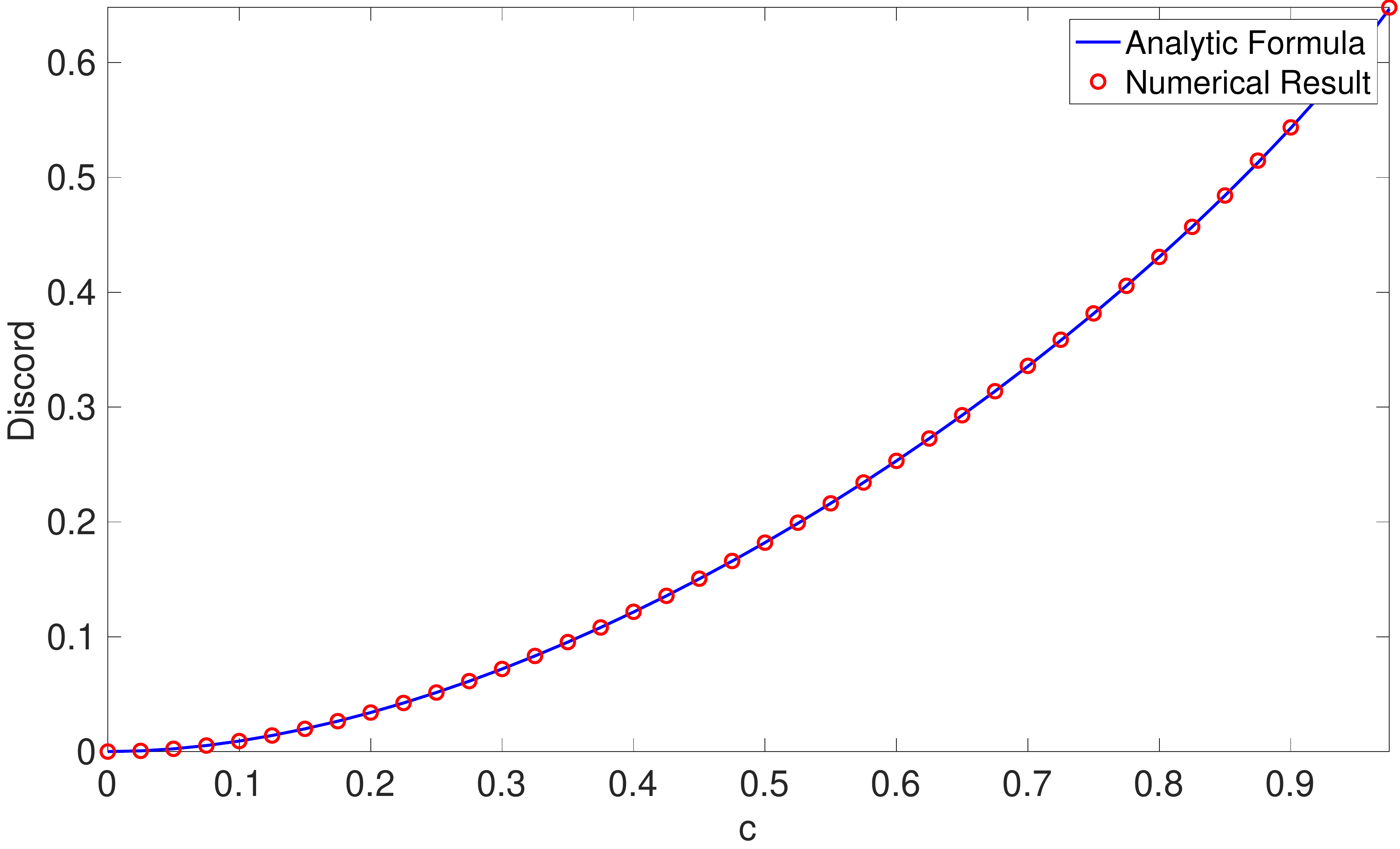}
    \caption{Quantum discord of $\rho(c)$ as a function of $c$, calculated analytically according to Eq.~\eqref{eqn:luo} (blue curve), and numerically with Algorithm \ref{alg:disc_2} (red circles).}
    \label{fig:dicord_c}
\end{figure}

For demonstration, we calculated the discord of the following family of states
\begin{eqnarray}
\rho(c) = (1-c) \frac{\mathbbm{1}}{4} + c |\Psi_-\rangle \langle\Psi_-|, \quad c \in [0,1]
\end{eqnarray}
analytically with the known discord formula for this particular family of two-qubit states in Ref.\cite{luo2008quantum}
\begin{equation}
    D(\rho(c)) = \frac{1-c}{4}\log(1-c) - \frac{1+c}{2}\log(1+c) + \frac{1+3c}{4}\log(1+3c), \label{eqn:luo}
\end{equation}
as well as numerically with the newly introduced Algorithm \ref{alg:disc_2}. As shown in Figure \ref{fig:dicord_c}, the two approaches match perfectly. Moreover, unlike the limited analytic formula in Ref.\cite{luo2008quantum}, our numerical method can be applied to arbitrary mixed states in arbitrary dimensions.

  \chapter{Concepts of Fermionic Correlation and Entanglement} \label{chap:fermion}

In Chapter \ref{chap:foundations} we discussed the notion of entanglement and correlation, and how to quantify them in an operationally meaningful way in systems consisting of distinguishable particles. In this chapter we will turn to systems consisting of identical particles, in particular fermions. We first elucidate its departure from the distinguishable case, and then introduce two paths of redefining the notion of entanglement, namely the notion of mode entanglement and particle entanglement. Partial results of this section has been published/archived on Refs.~\cite{ding2020correlation,ding2020concept,ding2022quantifying}.

\section{Challenges with Fermions}

The concepts of entanglement and correlation, as reviewed in Section \ref{sec:distinguishable_ent}, refer to a well-defined separation of the total system into two (or more) distinguishable subsystems. In the simplest case, this separation emerges naturally from the physical structure of the total system, namely by referring to a possible spatial separation of two subsystems. In that case, it will be also easier to experimentally access both subsystems to eventually extract and utilise the entanglement from their joint quantum state. Nonetheless, the notion of bipartite correlation and entanglement is by no means unique for a given system since one just needs to identify some tensor product structure in the total system's Hilbert space, $\mathcal{H}\equiv \mathcal{H}_A\otimes \mathcal{H}_B$. In the most general approach, one even defines subsystems by choosing two commuting subalgebras $\mathcal{A}_A, \mathcal{A}_B$ of observables\cite{zanardi2001virtual}. This also highlights the crucial fact that entanglement and correlation are relative concepts since they refer to a choice of subsystems/subalgebras of observables.

In case of identical fermions the identification of subsystems is not obvious at all. For instance, how could one decompose the underlying $N$-fermion Hilbert space $\mathcal{H}_N \equiv \wedge^N[\mathcal{H}_1]$ where $\mathcal{H}_1$ is the one-particle Hilbert space, or the Fock space $\mathcal{F}\equiv \oplus_{N\geq 0}\mathcal{H}_N$?
Actually, there are two natural routes to overcome these issues. The first one (presented in Section \ref{sec:mode}) refers naturally to the 2nd quantised formalism and leads to the notion of \emph{mode} (sometimes also called orbital or site) entanglement and correlation \cite{Friis13,Friis16,Sergey17}. The second route (discussed in Section \ref{sec:part}) is related more to first quantisation and defines correlation and entanglement in the \emph{particle} picture.

\section{Mode Picture\label{sec:mode}}

\subsection{Finding Tensor Product}

A natural tensor product structure emerges in the formalism of second quantisation, facilitating a bipartition on the set of spin-orbitals. To explain this, let us fix a reference basis for the one-particle Hilbert space $\mathcal{H}_1$. We then introduce the corresponding fermionic creation $f_i^{\dagger}$ and annihilation operators $f_j$, fulfilling the fermionic commutation relations,
\begin{equation}
    \{ f^{(\dagger)}_i, f^{(\dagger)}_j \} = 0, \quad \{f^\dagger_i, f_j \} = \delta_{ij}. \label{anticommute}
\end{equation}
In the quantum information community the one-particle reference states are often referred to as \textit{modes}, or (lattice) sites by condensed matter physicists. Each spin-orbital or generally mode $i$ can be either empty or occupied by a fermion. In this picture, the quantum states are naturally represented in the occupation number basis. The respective \emph{configuration states}
\begin{equation}\label{config}
    |n_1,n_2,\ldots , n_d \rangle = (f^\dagger_1)^{n_1}(f^\dagger_2)^{n_2} \cdots (f^\dagger_d)^{n_d} |0\rangle
\end{equation}
with $n_1,n_2\ldots,n_d \in \{0,1\}$ form a basis for the Fock space $\mathcal{F}(\mathcal{H}_1)$.
Bipartitions naturally arise as separations of the set of modes $\{1,2,\ldots,d\}$ into two, let's say the first $m$ and the last $d-m$, leading to
\begin{eqnarray}\label{slitconfig}
\lefteqn{|n_1, \ldots, n_m, n_{m+1},\ldots , n_d \rangle} && \nonumber \\
        &\mapsto &|n_1, n_2, \ldots n_{m} \rangle_A \otimes |n_{m+1}, n_{m+2} , \ldots ,n_{d}\rangle_B.
\end{eqnarray}
The total Fock space $\mathcal{F}(\mathcal{H}_1)$ admits then the tensor product structure
\begin{equation}\label{splitFock}
    \mathcal{F}_{AB}\equiv \mathcal{F}(\mathcal{H}_1) = \mathcal{F}(\mathcal{H}_1^{(A)}) \otimes \mathcal{F}(\mathcal{H}_1^{(B)})\equiv \mathcal{F}_A \otimes \mathcal{F}_B,
\end{equation}
where $\mathcal{H}_1^{(A/B)}$ denotes the one-particle Hilbert space spanned by the first $m$ and last $d-m$ modes, respectively. Actually, any splitting of the one-particle Hilbert space into two complementary subspaces, $\mathcal{H}_1 = \mathcal{H}_1^{(A)}\oplus \mathcal{H}_1^{(B)}$, induces a respective splitting \eqref{splitFock} on the Fock space level. Moreover, such a decomposition of the total Fock space into two factors allows us to introduce mode reduced density operators $\rho_{A/B}$ for the respective mode subsystem $A/B$. They are obtained by taking the partial trace of the total state $\rho$ with respect to the complementary factor $\mathcal{F}_{B/A}$. Consequently, $\rho_{A/B}$ is defined as an operator on the local space $\mathcal{F}_{A/B}$ and in general does not have a definite particle number anymore.

It seems that we can now readily apply the common quantum information theoretical formalism referring to distinguishable subsystems. Yet there is one crucial obstacle. Not every Hermitian operator acting on a fermionic Fock space is a physical observable, due to the so-called superselection rules.

\subsection{Superselection Rules} \label{sec:SSRIncorp}

A key ingredient in the physics of fermionic systems is the so-called  {\it parity superselection rule} (P-SSR). In its original form,  P-SSR forbids superpositions of even and odd fermion-numbers states. In a more modern version, P-SSR states that the operators belonging to physically measurable quantities must commute with the particle parity operator. This means they have to be linear combinations of even degree monomials of the creation and annihilation operators. This in turn implies that a superposition of two pure states with even and odd particle numbers cannot be distinguished from  an incoherent classical mixture of those states, thus one recovers the original formulation as a consequence.

The idea that the laws of nature impose P-SSR on fermionic systems was originally derived based on group theoretical arguments.\cite{SSR,wick1970superselection, wick1997intrinsic} However, the importance of P-SSR is also obvious from the fundamental fact that violation of P-SSR would lead to a contradiction to the no-signaling theorem, as we will explain in the following.  The no-signaling theorem states that two spatially separated parties cannot communicate  faster than the speed of light. To relate this to the P-SSR, let us assume that two distant parties Alice and Bob could violate P-SSR locally. That is they are able to locally superpose odd and even parity states. For our argument it is sufficient for Alice and Bob to have each access to one mode (e.g., an atomic spin-orbital).  Their local Fock spaces are thus generated by the fermionic annihilation and creation operators $(f_{A},f_{A}^{\dagger})$ and $(f_{B},f_{B}^{\dagger})$, respectively. Assume now that they share the state $|\psi\rangle_{AB} = \frac{1}{\sqrt{2}}(|0\rangle_{A}|0\rangle_{B} + |0\rangle_{A}|1\rangle_{B})$, which explicitly violates P-SSR. The protocol for Alice to communicate instantaneously one bit $b=0,1$ of classical information to Bob would be the following (see also Figure \ref{nosignalling}): before signaling both of them synchronize the clocks in their labs, and agree upon a signaling time. If Alice wants to communicate $1$, she does nothing (i.e., formally applies the unitary $U_1=\mathbbm{1}$), so $|\psi\rangle_{AB}$ remains unchanged; if she wishes to communicate $0$, Alice applies the unitary  $U_0=i(f^{\dagger}_{A} - f_{A}^{\phantom{\dagger}})$. The shared state then becomes
$|\psi^{\prime}\rangle_{AB} =\frac{i}{\sqrt{2}}(|1\rangle_{A}|0\rangle_{B} + |1\rangle_{A}|1\rangle_{B})$. At the same instance Bob measures the observable $\frac{1}{2}(f_{B} + f_{B}^{\dagger}+\mathbbm{1})$. One easily verifies that the outcome of Bob's measurement is deterministic and will be nothing else than the value of $b$, Alice's message. Hence, this protocol allows Alice to communicate instantaneously one bit ($b$) of information in contradiction to the no-signaling theorem and the laws of relativity.

\begin{figure}[ht]
\centering
\includegraphics[scale=0.4]{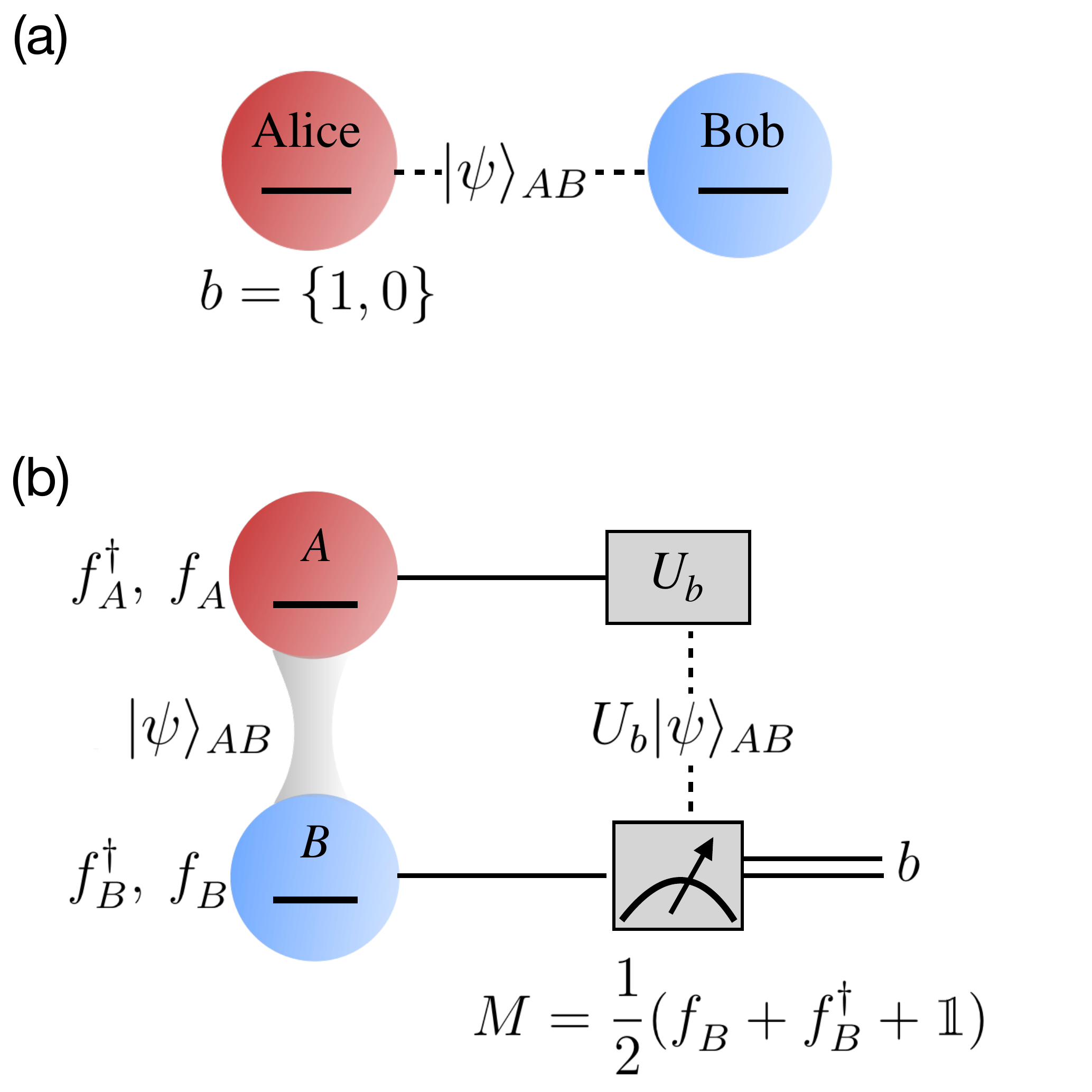}
\caption{(a) Two spatially separated parties Alice and Bob share a quantum state $|\psi\rangle_{AB}$. (b) The protocol showing how superluminal signaling is possible when parity superselection rule is broken: Alice communicates the bit value $b \in \{0,1 \}$ by applying the corresponding unitary $U_{b}$, Bob measures the observable $M$ and obtains instantaneously that bit value, as explained in the text.
}
\label{nosignalling}
\end{figure}

Beside the parity superselection rule, it is often pertinent to consider superselection rules due to some experimental limitations. One such rule is the fermion {\it particle number superselection rule} (N-SSR). Measurable quantities obeying N-SSR must commute with the particle parity operators.\cite{wick1970superselection} This, in the form of lepton number conservation, was once considered to be an exact symmetry of Nature. Recently, however, there have been indications that fundamental Majorana particles may exist which could lead to a violation of the N-SSR. Therefore N-SSR is now conservatively regarded exact only in the low energy regime, where a pair of electrons cannot be spontaneously created. That is to say, in the common settings in condensed matter physics and quantum chemistry, N-SSR is still of compelling relevance.

Having established the fundamental importance of superselection rules, we will now elucidate how they affect our description of quantum states, and consequently change the physically accessible correlation and entanglement in a quantum state. Accordingly, the SSRs will have important consequences for the realization of quantum information processing tasks (e.g., quantum teleportation\cite{debarba2020teleporting,olofsson2020quantum}).

To rephrase and summarise, SSRs are restrictions on local algebras of observables, resulting in physical algebras $\mathcal{A}_A$ and $\mathcal{A}_B$. If the SSR is related to some locally conserved quantity $Q_{A/B}$ (e.g. local parity, local particle number etc.), then local operators must also preserve this quantity. That is, all local observables satisfy
\begin{equation}
\mathcal{A}_{A/B} \ni O_{A/B} = \sum_q P_q O_{A/B} P_q,
\end{equation}
where $q$ ranges over all possible value of $Q_{A/B}$ and $P_q$'s are projectors onto the eigensubspaces, i.e. $O_{A/B}$ are block diagonal in any eigenbasis of $Q_{A/B}$.  It follows that different SSRs will lead to drastically different $\mathcal{A}_{A/B}$. The fact that we cannot physically implement every mathematical operator changes the accessibility of quantum states. The fully accessible states are called the physical states, and they satisfy
\begin{equation}
\rho = \sum_{q,q'} P_q \otimes P_{q'} \rho P_q \otimes P_{q'},
\end{equation}
or equivalently
\begin{equation}
[\rho, Q_A \otimes Q_B] = 0. \label{eqn:SSRcom}
\end{equation}
For a general state $\rho$ which does not satisfy \eqref{eqn:SSRcom}, we can obtain its physical part by the following projection
\begin{equation}
\rho^\textrm{Q} \equiv  \sum_{q,q'} P_q \otimes P_{q'} \rho P_q \otimes P_{q'}. \label{eqn:tilde}
\end{equation}

In Figure \ref{fig:sectors} we illustrate the process of obtaining the physical states under P-SSR and N-SSR from a family of two-orbital quantum states that commute with total particle number and magnetization, a common setting for quantum chemistry calculation which will also be featured in Section \ref{sec:qchem}. The projectors $P_q$'s simply take out coherent terms between sectors with different local parities or particle numbers.
\begin{figure}[!t]
\centering
\includegraphics[scale=0.40]{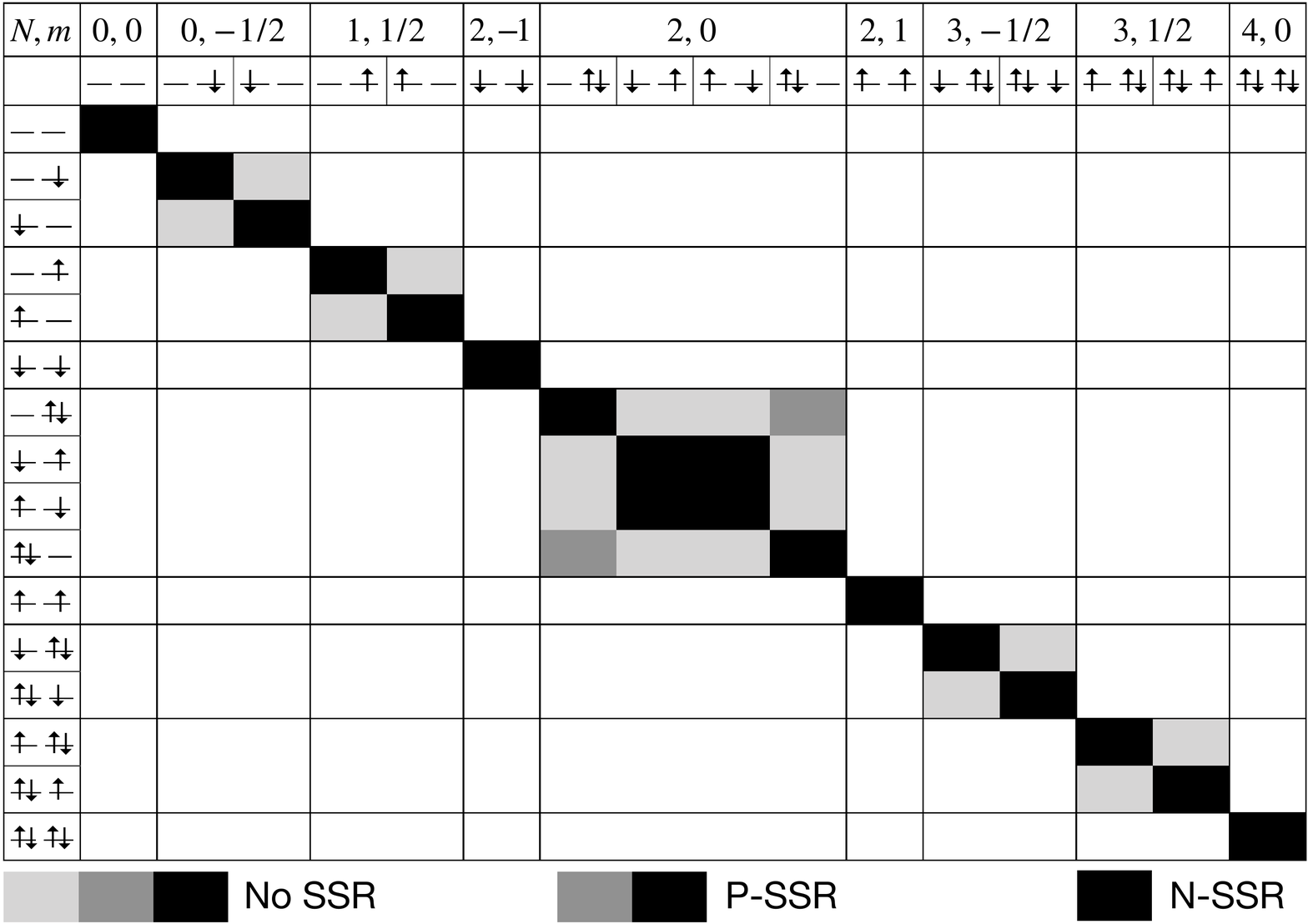}
\caption{Illustration of two-orbital reduced density matrix $\rho_{i,j}$. Most entries vanish as we assume spin and particle symmetry (white). According to \eqref{eqn:tilde} the P-SSR sets light gray entries to zero while N-SSR removes in addition two entries (gray).}
\label{fig:sectors}
\end{figure}

The physical state $\rho^\textrm{Q}$ gives the same expectation value as $\rho$ for all \textit{physical} observables. Therefore we can define a new class of uncorrelated states to be the ones with uncorrelated physical parts with respect to the physical algebra:
\begin{equation}
\begin{split}
\mathcal{D}_0^{\textrm{Q-SSR}} = \{ \rho \, | \,& \langle O_A \otimes O_B \rangle_{\rho} =   \langle O_A  \rangle_{\rho_A}  \langle O_B \rangle_{\rho_B},
\\
&\forall O_A \in \mathcal{A}_A, O_B \in \mathcal{A}_B \}.
\end{split}
\end{equation}
It is clear that the new set of uncorrelated states includes the one of the distinguishable setting, i.e. $\mathcal{D}_0 \subseteq \mathcal{D}^{\textrm{Q-SSR}}_0$. Consequently also more states are deemed separable. Relating to Figure \ref{fig:states}, both the correlation and entanglement measure become smaller in the presence of an SSR. There are two key messages here. First of all, correlation and entanglement are relative concepts. They depend not only on the particular division of the total system into two (or more) subsystems but also on the underlying SSRs, which eventually defines the physical local algebras of observables $\mathcal{A}_{A/B}$ and the global algebra $\mathcal{A}_A \otimes \mathcal{A}_B$. Secondly, by ignoring the fundamentally important SSRs, one may radically overestimate the amount of physical correlation and entanglement in a quantum state.

\begin{figure}[!t]
    \centering
    \includegraphics[scale=0.6]{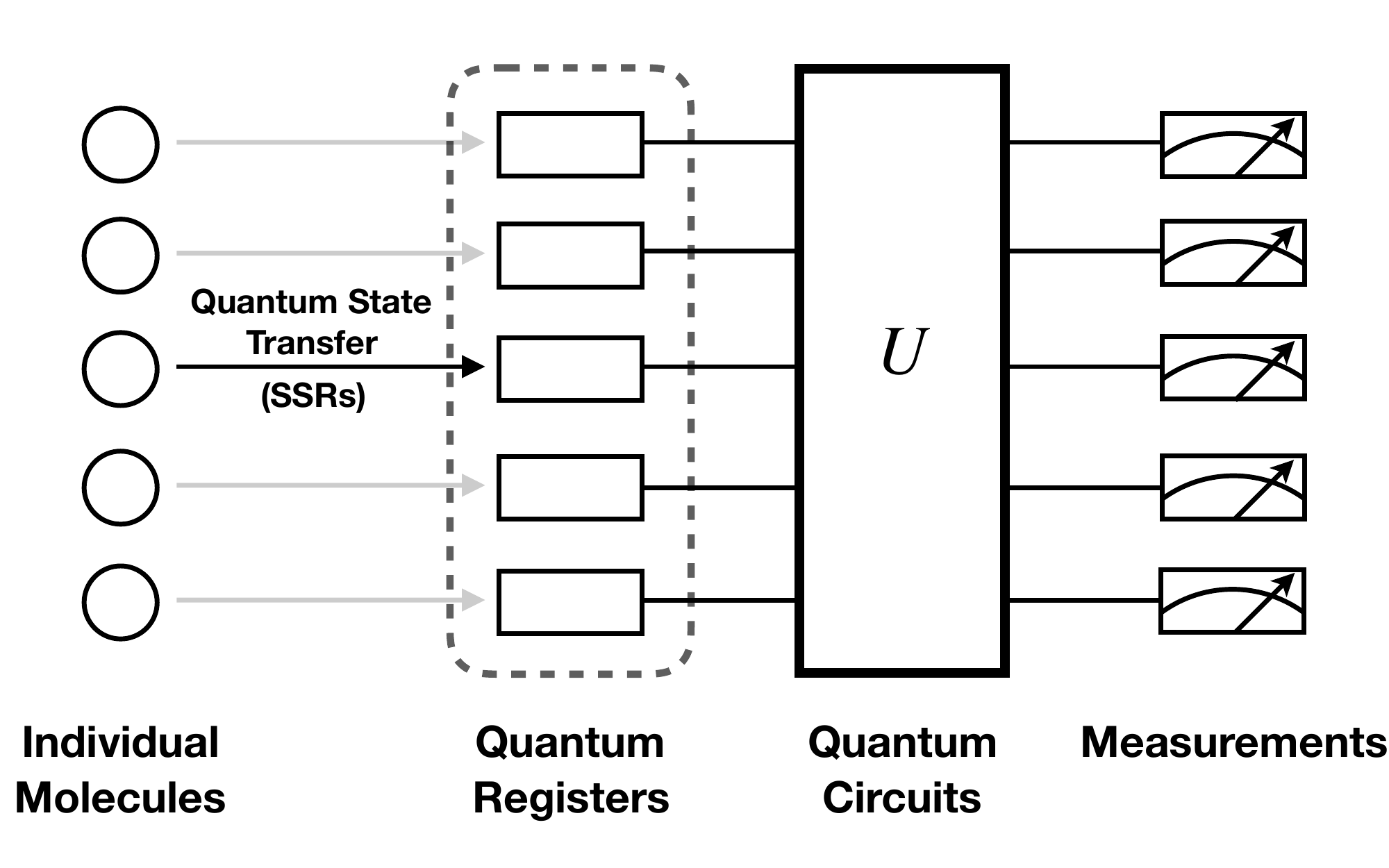}
    \caption{Schematic protocol for utilizing entanglement from molecular systems (see text for more details).}
    \label{fig:QCScheme}
\end{figure}

One of the biggest motivation for correctly identifying the amount of physical correlation and entanglement in a quantum state is its value for information processing tasks. An operationally meaningful quantification of entanglement does not only reveal non-local properties of a quantum state, but should also measure the amount of resource that can be extracted for performing various quantum information tasks mentioned in Chapter \ref{chap:intro}. In Figure \ref{fig:QCScheme} we illustrate the schematic protocol for utilizing entanglement from molecular systems. The quantum states of individual molecules are transferred to SSR-free quantum registers with Hilbert spaces of equal or higher dimensions, through local measurements and classical communication. A quantum circuit represented by a unitary gate $U$ in Figure \ref{fig:QCScheme} then acts on these quantum register states to perform computations. Finally, the end results of the computation are retrieved with carefully designed measurements. The key step that limits the extraction of entanglement is the transferring of the quantum state, which is constrained by the underlying  SSR\cite{bartlett2003entanglement}. What remains on the quantum registers after the transfer are the \emph{physical} parts defined in Eq.~\eqref{eqn:tilde}. From this perspective, the Q-SSR-constrained total correlation, entanglement and classical correlation of a single system in a state $\rho$ follow as
\begin{equation}
\begin{split}
I^{\textrm{Q-SSR}}(\rho) &= I(\rho^\textrm{Q}),
\\
E^{\textrm{Q-SSR}}(\rho) &= E(\rho^\textrm{Q}),
\\
C^{\textrm{Q-SSR}}(\rho) &= C(\rho^{\textrm{Q}}), \label{eqn:SSRmeasures}
\end{split}
\end{equation}
where $I$, $E$ and $C$ are the preferred measures for total correlation, entanglement and classical correlation.

\section{Particle Picture\label{sec:part}}

Besides partitioning modes in the second quantisation picture, the formalism of first quantisation seems to suggest another tensor product structure by exploiting the embedding
\begin{equation}
    \mathcal{H}_N \equiv  \wedge^N(\mathcal{H}_1) \leq \mathcal{H}_1^{\otimes N} ,
\end{equation}
of $N$-fermion Hilbert space into the one of $N$ distinguishable particles. An issue arises as the antisymmetry of $N$-fermion quantum states now erroneously would contribute to this particle correlation/entanglement. To see this, we write down a non-interacting two-fermion state by pseudo-labeling the particles
\begin{equation}
    |\psi\rangle = |\!\uparrow\rangle \wedge |\!\downarrow\rangle = \frac{|\!\uparrow\rangle_1 \otimes |\!\downarrow\rangle_2 - |\!\downarrow\rangle_1 \otimes |\!\uparrow\rangle_2}{\sqrt{2}}. \label{eqn:part_example}
\end{equation}
As the state is a single Slater determinant, it contains no correlation. However, when embedded into $\mathcal{H}^{\otimes 2}$, the state looks manifestly entangled. This confusing situation is caused by the pseudo-labeling. The state \ref{eqn:part_example} actually describe two \textit{distinguishable} particles whose state happens to be antisymmetrised. Therefore using the tensor product structure of $\mathcal{H}^{\otimes 2}$ one effectively overestimates the available particle entanglement by including that arising from the antisymmetrisation.

One the other hand, there is well-defined alternative concept inspired by resource theory\cite{ResT19} which looks rather similar: One defines the configuration states as the distinguished resource-free states
\begin{defn}[Free States]\label{def:uncorrP}
A fermionic state $\rho $ is called free in the particle picture, if and only if it can be represented by a single configuration state, i.e., $\rho\equiv \ket{\Psi}\!\bra{\Psi} $ with,
\begin{equation}\label{config1st}
\ket{\Psi}= \ket{\phi_1} \wedge \ldots \wedge \ket{\phi_N}
\end{equation}
for some (orthonormal) one-fermion states/modes $\ket{\phi_1}, \ldots, \ket{\phi_N} \in \mathcal{H}_1$.
\end{defn}
\noindent Furthermore, in analogy to the separable states one defines
\begin{defn}[Quantum-Free States]\label{def:sepP}
A fermionic state $\rho$ is called quantum-free in the particle picture, if and only if it can be written as a mixture of free states.
\end{defn}
\noindent A few comments are in order. First, the definition of free and quantum-free states could be applied in the context of both fixed particle number ($N$-fermion Hilbert space) and flexible particle number (Fock space). Second, since the definitions of free and quantum-free states look rather similar to those of uncorrelated and non-entangled states, we denote the respective sets by $\mathcal{D}_0^{(p)}$ and $\mathcal{D}_{sep}^{(p)}$, respectively. The superindex `$p$' refers to the particle picture and similarly we will add in the following a superindex `$m$' to the corresponding sets in the mode/orbital picture (as introduced by Definitions \ref{def:uncorr}, \ref{def:sep}) in case both pictures are discussed at the same time. 

\subsection{Quasifree States and Nonfreeness}

Measures for the nonfreeness and its quantum part can then be obtained by determining the minimal distances of a given quantum state $\rho$ to the sets $\mathcal{D}_0^{(p)}$ and $\mathcal{D}_{sep}^{(p)}$, respectively. Due to the close relation of this \emph{(quantum) nonfreeness} to the concepts of total (quantum) correlation we denote the respective measure by $(E^{(p)})\, C^{(p)}$.
Actually, the nonfreeness was first introduced by Gottlieb and Mauser\cite{gottlieb2005new,gottlieb2007properties,Gottlieb14,Gottlieb15} and they  observed\cite{gottlieb2007properties} that using the quantum relative entropy as distance function leads to an analytic formula  (referring to a Fock space-related Definition \ref{def:uncorrP}),
\begin{equation}\label{PartCorr}
    C^{(p)}(\rho) = S(\rho_1)+ S(\mathbbm{1}-\rho_1)-S(\rho).
\end{equation}
In this formula, the 1RDM $\rho_1$ of $\rho$ is trace-normalized to the particle number $N$. In case of pure total states $\rho$, $S(\rho)$ vanishes and the nonfreeness $C^{(p)}$ is nothing else than the particle-hole symmetrized von Neumann entropy of the 1RDM. Since this nonfreeness has a beautiful geometric meaning, the chances for discovering an underlying operational meaning might be better than for the quantity $S(\rho_1)$ as used in most works so far (see, e.g., Refs.~\cite{Ziesche97b,huang2005entanglement,Schoutens08}).

\subsection{Slater Rank and Slater Number}

Deriving an explicit analytic expression for the quantum part $E^{(p)}(\rho)$ of the nonfreeness seems to be a rather hopeless task again (as for the entanglement of mixed states in general). It is thus quite remarkable that at least for the case of two fermions in a four-dimensional one-particle Hilbert space an analytic procedure has been found\cite{schliemann2001quantum} (which, however, does not involve the quantum relative entropy and instead is based on a so-called convex-roof construction).

For a pure $N$-fermion state $|\psi\rangle$, the Slater Rank is defined as the minimum number of Slater determinants one needs to expand $|\psi\rangle$. This is extended to the case a general mixed state $\rho$, by minimizing the maximum Slater rank within a pure state decomposition of $\rho$, over all possible decompositions. This seems to be a daunting task, but the Slater number of an arbitrary two-fermion mixed state $\rho$ can be analytically found. In the first step, one determines the spectral decomposition of the given two-fermion density operator $\rho$ on $\wedge^2[\mathcal{H}_1]$ (here the respective eigenvalues are absorbed into the states $\ket{\Psi_i}$),
\begin{equation}
\rho= \sum_{i=1}^6 \ket{\Psi_i}\!\bra{\Psi_i}.
\end{equation}
By introducing an arbitrary reference basis for $\mathcal{H}_1$, one determines for all six contributions $\ket{\Psi_i}$ the antisymmetric expansion
matrices $w^{(i)}$,
\begin{equation}
    |\Psi_i\rangle = \sum_{a,b=1}^4 w^{(i)}_{ab} f^\dagger_a f^\dagger_b |0\rangle.
\end{equation}
Those are then used to calculate for $i,j=1,\ldots,6$
\begin{equation}
    K_{ij} = \sum_{a,b,c,d=1}^4 \epsilon^{abcd} w^{(i)}_{ab} w^{(j)}_{cd} \label{C_SL}.
\end{equation}
The quantum nonfreeness eventually follows as\cite{schliemann2001quantum}
\begin{equation} \label{PartEnt}
    E^{(p)}(\rho) = 2 \max_i |\kappa_i| - \Tr[|K|],
\end{equation}
where $\{\kappa_i\}$ are the eigenvalues of the matrix $K\equiv(K_{ij})$ and $\Tr[|K|]= \sum_{i=1}^{6}|\kappa_i|$.

  \chapter{Quantifying Mode Entanglement} \label{chap:quantifying}

In this Chapter we focus on a system of two sites/orbitals, and quantify the mode entanglement between these two sites/orbitals. The system can be in principle a subsystem of a larger set of sites/orbitals. Therefore the results are intended for general mixed states on the two sites/orbitals. Moreover, we quantify the accessible entanglement under the restriction of superselection rules according to Section \ref{sec:SSRIncorp}. Our goal is to find the minimizer to the relative entropy of entanglement \eqref{eqn:rel_ent}, i.e. the closest separable state $\sigma^\ast$ to a two-orbital physical state $\rho$ such that
\begin{equation}
    S(\rho||\sigma^\ast) = \min_{\sigma \in \mathcal{D}_\text{sep}} S(\rho||\sigma).
\end{equation}
With the help of symmetry, in some cases $\sigma^\ast$ can be exactly found, as will be shown in Section \ref{sec:analytic}. In case the minimization cannot be analytically solved, we resort to numerical means using semidefinite programming (SDP) in Section \ref{sec:SDP}. 

\section{Analytic Formula for Orbital-Orbital Entanglement} \label{sec:analytic}

\subsection{Symmetry and Entanglement} \label{sec:symm}

Finding the minimizer in \eqref{eqn:rel_ent} is not an easy task, even though from now on we restrict ourselves two-qudit systems where the qudits realized by two sites we labelled $A$ and $B$ that can host up to two electrons (spin $\uparrow$ and $\downarrow$) each. This is a common setting where the sites can be physical lattice sites in a many-electron system, or energy orbitals in a molecule. Even so, the total Fock space of the system is $4\times4$ dimensional, and a general density matrix has $2\sum_{n=1}^{16} n + 16 -1=255$ real-valued degrees of freedom. And more complexity will be introduced by separability constraints which set the boundaries of the range of minimization.

However, if the state of interest $\rho$ exhibits local unitary symmetries, the task of minimization can be greatly simplified. The minimizer $\sigma^\ast$, the closest separable state to $\rho$, shares the same local unitary symmetries as $\rho$\cite{vollbrecht2001entanglement}.

\begin{thrm} \label{thrm:symm}
If $E(\rho) \equiv \min_{\sigma \in D_{sep}} S(\rho||\sigma) = S(\rho||\sigma^\ast)$, and $U^\dagger(g) \rho U(g) = \rho$ for a set of local unitaries $U(g)$ representing elements of a group $g \in G$, then $\sigma^\ast$ satisfies
\begin{equation}
    \sigma^\ast = \mathcal{T}_G(\sigma^\ast),
\end{equation}
where
\begin{equation}
    \mathcal{T}_G(\cdot) = \begin{cases}
    \frac{1}{|G|} \sum_{g \in G} U(g)^\dagger (\cdot) U(g), \quad &G \:\: \textrm{discrete},
    \\
    \int_G \mathrm{d}\mu(g) \, U(g)^\dagger(\cdot) U(g), \quad &G \:\: \textrm{Lie group},
    \end{cases}
\end{equation}
and
\begin{equation}
    U^\dagger(g) \sigma^\ast U(g) = \sigma^\ast, \quad \forall g \in G.
\end{equation}
\end{thrm}
\begin{proof}
We prove the case with discrete $G$. Same reasoning applies to those with Lie groups as well.
\begin{equation}
    \begin{split}
        S(\rho||\sigma^\ast) &= \frac{1}{|G|} \sum_{g \in G} S(U(g)^\dagger \rho U(g) || U(g)^\dagger \sigma^\ast U(g))
        \\
        &= \frac{1}{|G|} \sum_{g \in G} S(\rho || U(g)^\dagger \sigma^\ast U(g))
        \\
        &\geq S\left( \rho \left|\left| \frac{1}{|G|}\sum_{g \in G}\right.\right. U(g)^\dagger\sigma^\ast U(g)  \right)
        \\
        &= S(\rho||\mathcal{T}_G(\sigma^\ast)). \label{app:eq:entropy}
    \end{split}
\end{equation}
In the second line we used the unitary invariance of the relative entropy, and in the third line we used the convexity of the relative entropy. Since $U$ are local unitaries and they do not generate any entanglement, $\mathcal{T}_G(\sigma)$ is still inside the set of separable states $D_{sep}$. If we assume there exists a unique minimizer, namely $\sigma^\ast$, then \eqref{app:eq:entropy} shows a contradiction and $\sigma^\ast = \mathcal{T}_G(\sigma^\ast)$. If we assume the minimizer $\sigma^\ast$ is not unique, we can then replace it with $\sigma^\ast \rightarrow \mathcal{T}_G(\sigma^\ast)$ without altering the relative entropy of entanglement. Since $\mathcal{T}_G(\cdot)$ is a projection, we again arrive at $\sigma^\ast = \mathcal{T}_G(\sigma^\ast)$. We can check that
\begin{equation}
\begin{split}
    U^\dagger \mathcal{T}_G(\sigma^\ast) U &= \frac{1}{|G|} \sum_{U'\in G} U^\dagger U'^\dagger \sigma^\ast U' U = \frac{1}{|G|} \sum_{U''\in G} U''^\dagger \sigma^\ast U'' = \mathcal{T}_G(\sigma^\ast).
\end{split}
\end{equation}
\end{proof}
In fact, Theorem \ref{thrm:symm} can be generalized to non-local unitary group $G$, as long as $\mathcal{T}_G(\mathcal{D}_{sep}) \subseteq \mathcal{D}_{sep}$ is satisfied. That is, $\mathcal{T}_G$ is not entanglement generating.

To illustrate the action of $\mathcal{T}_G$, we turn to a simple unitary group generated by the particle number operator $ {N}$ defined as
\begin{equation}
     {N} = \mathbbm{1} \otimes  {N}_B +  {N}_A \otimes \mathbbm{1}.
\end{equation}
The unitaries generated by $ {N}$ are local and of the form
\begin{equation}
    U = e^{i\alpha  {N}} = e^{i\alpha  {N}_A} \otimes e^{i\alpha  {N}_B}, \quad \alpha \in [0,2\pi).
\end{equation}
We now compute $\mathcal{T}_G(\rho)$ by representing $\rho$ in the eigen-basis of $ {N}$

\begin{equation}
    \begin{split}
        \mathcal{T}_G(\rho) &= \frac{1}{2\pi} \int \mathrm{d}\alpha \, e^{-i\alpha  {N}_A} \otimes e^{-i\alpha  {N}_B} \rho \, e^{i\alpha  {N}_A} \otimes e^{i\alpha  {N}_B}
        \\
        &= \frac{1}{2\pi} \int \mathrm{d}\alpha \, e^{-i\alpha  {N}_A} \otimes e^{-i\alpha  {N}_B} \sum_{n,m,n',m'} \rho_{n,m,n',m'} |n\rangle\langle m| \otimes |n'\rangle \langle m'| \, e^{i\alpha  {N}_A} \otimes e^{i\alpha  {N}_B}
        \\
        &= \sum_{n,m,n',m'} \delta_{n+n',m+m'} \rho_{n,m,n',m'} |n\rangle\langle m| \otimes |n'\rangle \langle m'| = \sum_{N=1}^{D} P_N \rho P_N.
    \end{split}
\end{equation}

The effect of $\mathcal{T}_G$ is block diagonalizing $\rho$ into eigen-sectors of $ {N}$. One can show that with more than one commuting generators, these block will be simultaneous eigen-sectors. This interpretation allows us to quickly write down $\mathcal{T}_G(\rho)$ given the conserved quantities without having to go through the integration.

\subsection{Derivations and Results} \label{sec:derive}
Let us first ask ourselves what symmetries we may assume. The first conserved quantities we have are the local particle numbers ${N}_A,  {N}_B$, as we are for now interested in the N-SSR physical state $\rho^\text{N}$ defined as
\begin{equation}
    \rho^\text{N} = \sum_{m=0}^2 \sum_{n=0}^2 P_m \otimes P_n \, \rho \, P_m \otimes P_n.
\end{equation}
Additionally, there is a common symmetry for electron systems namely the $SU(2)$ symmetry associated with the total electron spin. A state $\rho$ that enjoys this symmetry commutes with two conserved quantities which are eigenvalues of the operators
\begin{equation}
\begin{split}
     S^z &= \mathbbm{1} \otimes S^z_B + S_A^z \otimes \mathbbm{1},
    \\
    \vec{S}^2 &= \frac{1}{2}(S_+S_- + S_-S_+) + (S^z)^2.
\end{split}
\end{equation}
Similar to the unitary operators generated by the particle number operator $N$, the unitary operators generated by $ {S}^z$ are local. However, those generated by $| {\vec{S}}|$ are not. Therefore in order to use this symmetry we need to show that $| {\vec{S}}|$ is not entanglement generating in the presence of other local symmetries mentioned above. To be precise, we state the following proposition.

\begin{prop}\label{prop:totalspin}
If $\rho$ is separable and commutes with $N_A \otimes N_B$ and $S_z$, then $\mathcal{T}_{G}(\rho)$ is separable where the twirl is with respect to the unitary group generated by $\vec{S}^2$.
\end{prop}

A proof for Proposition \ref{prop:totalspin} is included in Appendix \ref{app:totalspin}.

In Table \ref{tab:sym} we listed the simultaneous eigen-states of $N_A\otimes N_B$, $S_z$ and $\vec{S}^2$. In this case all eigen-sectors are one-dimensional, and a state that satisfies $\sigma = \mathcal{T}_G(\sigma)$ can be written as
\begin{equation}
    \sigma = \sum_{i=1}^{16} p_i \sigma_i, \quad p_i \geq 0, \: \sum_i^{16} p_i = 1, \label{eqn:coeff}
\end{equation}
where $\sigma_i = |\Psi_i\rangle \langle \Psi_i|$ are the eigen-states listed in Table \ref{tab:sym}.

\begin{table}[h]
\centering
\begin{tabular}{|c|c|c|c|c|c|}
\hline
$N$                & $S^z$                   & $|\vec{S}|$            & $(N_A,N_B)$    & State                                                                                                                           & Ent.\rule{0pt}{2.6ex}\rule[-1.2ex]{0pt}{0pt} \\ \hline
0                  & 0                       & 0                      & $(0,0)$ & $|\Psi_1\rangle =|\Omega\rangle \otimes |\Omega\rangle$                                                                         & N \rule{0pt}{2.6ex}\rule[-1.2ex]{0pt}{0pt}   \\ \hline
\multirow{4}{*}{1} & \multirow{2}{*}{$1/2$}  & \multirow{2}{*}{$1/2$} & $(0,1)$ & $|\Psi_2\rangle =|\Omega\rangle \otimes |\!\uparrow\rangle$                                                                     & N    \rule{0pt}{2.6ex}\rule[-1.2ex]{0pt}{0pt}
\\ \cline{4-6}
                   &                         &                        & $(1,0)$ & $|\Psi_3\rangle =|\!\uparrow\rangle \otimes |\Omega\rangle$                                                                     & N   \rule{0pt}{2.6ex}\rule[-1.2ex]{0pt}{0pt}
                   \\ \cline{2-6}
                   & \multirow{2}{*}{$-1/2$} & \multirow{2}{*}{$1/2$} & $(0,1)$ & $|\Psi_4\rangle =|\Omega\rangle \otimes |\!\downarrow\rangle$                                                                   & N \rule{0pt}{2.6ex}\rule[-1.2ex]{0pt}{0pt}
                   \\ \cline{4-6}
                   &                         &                        & $(1,0)$ & $|\Psi_5\rangle = |\!\downarrow\rangle \otimes |\Omega\rangle$                                                                  & N    \rule{0pt}{2.6ex}\rule[-1.2ex]{0pt}{0pt}
                   \\ \hline
\multirow{6}{*}{2} & \multirow{4}{*}{$0$}    & \multirow{3}{*}{0}     & $(2,0)$ & $|\Psi_6\rangle =|\!\uparrow\downarrow\rangle \otimes |\Omega\rangle$                                                           & N    \rule{0pt}{2.6ex}\rule[-1.2ex]{0pt}{0pt}
\\ \cline{4-6}
                   &                         &                        & $(0,2)$ & $|\Psi_7\rangle =|\Omega\rangle \otimes |\!\uparrow\downarrow\rangle$                                                           & N   \rule{0pt}{2.6ex}\rule[-1.2ex]{0pt}{0pt}
                   \\ \cline{4-6}
                   &                         &                        & $(1,1)$ & $|\Psi_8\rangle =\frac{|\!\uparrow\rangle|\otimes|\downarrow\rangle - |\!\downarrow\rangle \otimes |\uparrow\rangle}{\sqrt{2}}$ & Y \rule{0pt}{2.6ex}\rule[-1.2ex]{0pt}{0pt}
                   \\ \cline{3-6}
                   &                         & 1                      & $(1,1)$ & $|\Psi_9\rangle =\frac{|\!\uparrow\rangle|\otimes|\downarrow\rangle + |\!\downarrow\rangle \otimes |\uparrow\rangle}{\sqrt{2}}$ & Y \rule{0pt}{2.6ex}\rule[-1.2ex]{0pt}{0pt}
                   \\ \cline{2-6}
                   & $1$                     & 1                      & $(1,1)$ & $|\Psi_{10}\rangle =|\!\uparrow\rangle \otimes |\!\uparrow\rangle$                                                              & N   \rule{0pt}{2.6ex}\rule[-1.2ex]{0pt}{0pt}
                   \\ \cline{2-6}
                   & $-1$                    & 1                      & $(1,1)$ & $|\Psi_{11}\rangle =|\!\downarrow\rangle \otimes |\!\downarrow\rangle$                                                          & N \rule{0pt}{2.6ex}\rule[-1.2ex]{0pt}{0pt}
                   \\ \hline
\multirow{4}{*}{3} & \multirow{2}{*}{$1/2$}  & \multirow{2}{*}{$1/2$} & $(2,1)$ & $|\Psi_{12}\rangle =|\!\uparrow\downarrow\rangle \otimes |\!\uparrow\rangle$                                                    & N \rule{0pt}{2.6ex}\rule[-1.2ex]{0pt}{0pt}
\\ \cline{4-6}
                   &                         &                        & $(1,2)$ & $|\Psi_{13}\rangle =|\!\uparrow\rangle \otimes |\!\uparrow\downarrow\rangle$                                                    & N   \rule{0pt}{2.6ex}\rule[-1.2ex]{0pt}{0pt}
                   \\ \cline{2-6}
                   & \multirow{2}{*}{$-1/2$} & \multirow{2}{*}{$1/2$} & $(2,1)$ & $|\Psi_{14}\rangle =|\!\uparrow\downarrow\rangle \otimes |\!\downarrow\rangle$                                                  & N \rule{0pt}{2.6ex}\rule[-1.2ex]{0pt}{0pt}
                   \\ \cline{4-6}
                   &                         &                        & $(1,2)$ & $|\Psi_{15}\rangle =|\!\downarrow\rangle \otimes |\!\uparrow\downarrow\rangle$                                                  & N \rule{0pt}{2.6ex}\rule[-1.2ex]{0pt}{0pt}
                   \\ \hline
4                  & $0$                     & $0$                    & $(2,2)$ & $|\Psi_{16}\rangle =|\!\uparrow\downarrow\rangle \otimes |\!\uparrow\downarrow\rangle$                                          & N \rule{0pt}{2.6ex}\rule[-1.2ex]{0pt}{0pt}
\\
\hline
\end{tabular}
\caption{Decomposition of the Fock space into eigen-sectors labeled by the tuple $({N}_A, {N}_B,S_z,\vec{S}^2)$.} \label{tab:sym}
\end{table}

Notice most of the eigen-states are already separable. We argue that all the entanglement is confined in the sector
\begin{equation}
    M = \mathrm{Span} \{ |\Psi_8\rangle, |\Psi_9\rangle, |\Psi_{10}\rangle, |\Psi_{11}\rangle \}. \label{eqn:sectorM}
\end{equation}

\begin{prop}\label{prop:M}
$\sigma$ is entangled if and only if $\sigma|_M$ is entangled.
\end{prop}
\begin{proof}
We can write $\sigma$ as the decomposition
\begin{equation}
    \sigma = \sigma|_M \oplus \sigma|_{M^\perp}.
\end{equation}
Assume $\sigma_M$ is entangled. Then its partial transpose (on the subsystem $B$, without loss of generality) $[\sigma|_M]^{T_B}$ necessarily has negative eigenvalues, according to the Peres-Horodecki criterion, which states a state is entangled if, and if and only if, in case of $2\times2$ and $2\times3$ dimensions, its partial transpose has negative eigenvalues. Then $[\sigma|_M]^{T_B}$ also has negative eigenvalues, since the partial trace preserves the sectors $M$ and $M^\perp$. Then $\sigma$ is necessarily entangled.

Now we assume $\sigma$ is entangled, and suppose $\sigma|_M$ is separable. We know that $\sigma|_{M^\perp}$ is separable. Then $\sigma$ is separable, due to the convexity of $D_{G} \cap D_{sep}$. We arrived at a contradiction. Therefore $\sigma|_M$ is entangled.
\end{proof}

Using the Peres-Horodecki criterion, $\sigma|_M$ is entangled if and only if its coefficients in \eqref{eqn:coeff} additionally satisfy
\begin{equation}
    q_{10} q_{11} \geq \left( \frac{q_8-q_9}{2} \right)^2. \label{eqn:separability}
\end{equation}
As $\rho^\text{N}$ and its closest separable state $\sigma^\ast$ can be simultaneously diagonalized, the relative entropy can be written as
\begin{equation}
    S(\rho^\text{N}||\sigma^\ast) = \sum_{i=1}^{16} p_i(\log(p_i)-\log(q_i)),
\end{equation}
where $p_i = \Tr[\rho \sigma_i]$. Our goal is to minimize $S(\rho^\text{N}||\sigma^\ast)$ by varying $\vec{q}=(q_1,q_2,\ldots,q_{16})$ under the constraint of \eqref{eqn:separability}. In the Lagrange multiplier formalism, this is equivalent to minimizing the function
\begin{equation}
        \begin{split}
        F &= -\sum_{i}^{16} p_i \log(q_i) + \lambda \left( \sum_{i=1}^{16} p_i - 1 \right)+ \mu \left( q_{10}q_{11} - \left( \frac{q_8 - q_9}{2} \right)^2 \right). \label{eqn:minF}
    \end{split}
\end{equation}
with respect to $\vec{q}$, $\lambda$ and $\mu$. But before we proceed we would like to show a useful result.

\begin{thrm} \label{thm:trace}
    Let $\rho$ be a density matrix that can be written as $\rho = a \rho_1 \oplus (1-a) \rho_2$ where $\rho_{1,2}$ are density matrices which are block diagonalized into disjoint sectors $M_{1,2}$ that preserve partial transpose (or other appropriate separability criteria), and $a \in (0,1)$. Then the closest separable state $\sigma^\ast$ to $\rho$ over some set of separable state can also be written as $\sigma^\ast = a \sigma^\ast_1 \oplus (1-a) \sigma^\ast_2$, where $\sigma^\ast_{1,2}$ are density matrices in sectors $M_{1,2}$.
\end{thrm}
\begin{proof}
    We know from Theorem \ref{thrm:symm} that the closest separable state $\sigma$ is also block diagonal and can be written as
    \begin{equation}
        \sigma^\ast = b \sigma^\ast_1 \oplus (1-b) \sigma^\ast_2, \quad b \in (0,1).
    \end{equation}
    Then the relative entropy can be written as
    \begin{equation}
    \begin{split}
        S(\rho||\sigma^\ast) &= \Tr[(a\rho_1 + (1-a)\rho_2) (\log(a\rho_1 + (1-a)\rho_2) - \log(b\sigma^\ast_1 + (1-b)\sigma^\ast_2))]
        \\
        &= a S(\rho_1||\sigma^\ast_1) + (1-a) S(\rho_2 ||\sigma^\ast_2)
        \\
        &\quad + a(\log(a)-\log(b))+ (1-a)(\log(1-a)-\log(1-b)).
    \end{split}
    \end{equation}
    For fixed $\sigma_{1,2}^\ast$, minimum is obtained when $b=a$. Therefore the closest separable state must have the same trace restricted to sectors $M_{1,2}$ as $\rho$.
\end{proof}

Theorem \ref{thm:trace} can be generalized to the case with more than two disjoint sectors. A direct consequence is, if we wish to find the closest separable state $\sigma^\ast$ to a state $\rho$ that is block diagonal, and these blocks are not coupled to each other by separability criteria, then minimization can be carried out independently in each sector. In particular, $\rho|_{M'} = \sigma^\ast|_{M'}$ when $M'$ is a one-dimensional sector in a reference tensor product basis. This means in our case we can immediately write down
\begin{equation}
    q_i = p_i, \quad i \neq 8,9,10,11, \label{eqn:Mperp}
\end{equation}
without performing any calculations. Furthermore, the relative entropy can be simplified to
\begin{equation}
    S(\rho||\sigma^\ast) = S(\rho|_M ||\sigma^\ast|_M).
\end{equation}
After this, we are left with a new minimizing function which only concerns sector $M$

\begin{equation}
        \begin{split}
        F &= -\sum_{i=8,9,10,11} p_i \log(q_i) + \lambda \left( \sum_{8,9,10,11} q_i - p_i \right)+ \mu \left( q_{10}q_{11} - \left( \frac{q_8 - q_9}{2} \right)^2 \right). \label{eqn:minF2}
    \end{split}
\end{equation}

\begin{itemize}
    \item \textbf{Special case: $p_{10} = p_{11}$.} We first look at a simpler situation where $p_{10} = p_{11}$. In this case symmetry demands that $q_{10} = q_{11}$ and the quadratic constraint in \eqref{eqn:separability} reduces to linear ones. Assuming $p_8 > p_9 \geq 0$, the coefficients for the closest separable state in sector $M$ are
\begin{equation}
    \begin{split}
        q_8 &= \frac{S}{2}, \quad q_9 = \frac{S}{2(S-p_8)} p_9,
        \\
        q_{10} &= \frac{S}{2(S-p_8)} p_{10},\quad S = p_8 + p_9 + 2p_{10} = \Tr[\rho|_M].
        \label{eqns:soln}
    \end{split}
\end{equation}
For the case $p_9 > p_8 \geq 0$ one simply swaps $p_8 \leftrightarrow p_9$ and $q_8 \leftrightarrow q_9$. This simple solution is due to the fact that the domain of search has linear boundaries and finitely many extreme points. The relative entropy of entanglement is then simplified to
\begin{equation}
    \begin{split}
        S(\rho|\sigma^\ast) &= t \log(t) + (S-t)\log(S-t) -S \log\left(\frac{S}{2}\right),
        \\
        t &= \max\{p_8,p_9\}. \label{eqn:rel_ent_formula}
    \end{split}
\end{equation}
That is, to calculate the entanglement of $\rho$, it suffices to know the values $t = \max\{\Tr[\rho\sigma_8], \Tr[\rho\sigma_9]\}$ and $S = \Tr[\rho|_M]$.

\item \textbf{General case: $p_{10} \neq p_{11}$.} In this general case the solution is more involved and we found the following parameters for the minimizer using software Mathematica. Again, assuming $p_8 > p_9 \geq 0$,
\begin{equation}
    \begin{split}
        q_8 &= \frac{A + B +\sqrt{C}}{4(S-p_9)},
        \\
        q_9 &= \frac{A - B - \sqrt{C}}{4(S-p_8)}
        \\
        q_{10} &= p_{10} + \frac{p_8+p_9-q_8-q_9}{2},
        \\
        q_{11} &= p_{11} + \frac{p_8+p_9-q_8-q_9}{2}, \label{eqns:sol_gen}
    \end{split}
\end{equation}
where $A,B,C$ are polynomial functions of $p_8,p_9,p_{10},p_{11}$
\begin{equation}
\begin{split}
    A &=  S^2 - (p_{10} - p_{11})^2,
    \\
    B &= (p_8-p_9)S,
    \\
    C &= (p_{10}+p_{11})^2(p_8 - p_9)^2 + 8 p_{10} p_{11} (2p_{10}p_{11} + (p_{10}+p_{11})(p_8 + p_9) + 2 p_8 p_9),
    \\
    S &= p_8 + p_9 + p_{10} + p_{11} = \Tr[\rho|_M].
\end{split}
\end{equation}
For the case $p_9 > p_8 \geq 0$ one again swaps $p_8 \leftrightarrow p_9$ and $q_8 \leftrightarrow q_9$. One can check that when $p_{10} = p_{11}$, \eqref{eqns:sol_gen} reduces to \eqref{eqns:soln}.
\end{itemize}

If the N-SSR restriction is relaxed to a P-SSR one, analytic formula for the site-site entanglement can also be obtained, using the previously ignored reflection symmetry $L \leftrightarrow R$ between the two sites. The P-SSR restricted entanglement in $\rho$ is quantified as $E(\rho^\text{P})$ where
\begin{equation}
\rho^\text{P} = \sum_{\tau, \tau' = \text{odd}, \text{even}} P_\tau \otimes P_{\tau'} \, \rho\, P_\tau \otimes P_{\tau'}.
\end{equation}
Following a similar line of reasoning, one can also show that a two-orbital state $\sigma$ that shares the particle number, magnetization, reflection and local parity symmetry can be expanded as \ref{eqn:coeff} with the additional changes to the eigenstates $\sigma_i$'s
\begin{equation}
\Psi_{6/7} \longrightarrow \frac{|\Omega\rangle \otimes |\!\uparrow\downarrow\rangle \mp |\!\uparrow\downarrow\rangle \otimes |\Omega\rangle}{\sqrt{2}}.
\end{equation}
One can prove that the symmetric state $\sigma$ is separable if and only if $\sigma|_M$ and $\sigma|_{M'}$ are individually separable, where $M'=\text{Span}\{|\Psi_1\rangle, |\Psi_6\rangle, |\Psi_7\rangle, |\Psi_{16}\rangle\}$ is the even-even sector (mirroring the odd-odd sector $M$ in \eqref{eqn:sectorM}). One can show that $M$ and $M'$ are preserved by partial transposition. Therefore the entanglement of the state $\rho$ can be divided into the sum of entanglement in sector $M$ and $M'$, by Theorem \ref{thm:trace}. Using again the Peres-Horodecki criterion, we find the P-SSR restricted entanglement as (for particle-hole symmetrized state $\rho$, i.e. $p_{1} = p_{16}$)
\begin{equation}
\begin{split}
E(\rho^\text{P}) = &E(\rho^\text{N}) + t' \log(t') + (s'-t')\log(s'-t') -s' \log\left(\frac{s'}{2}\right), \label{eqn:rel_ent_formula2}
\end{split}
\end{equation}
where  $s' = \Tr[\rho|_{M'}]$, $t' = \max\{p_6,p_7\}$ and $E(\rho^\text{N})$ the N-SSR restricted entanglement calculated with the solutions to closest separable states presented in \eqref{eqns:soln} or \eqref{eqns:sol_gen}. 

\subsection{Symmetry Inheriting}

In Section \ref{sec:analytic}, we derived an analytic formula \eqref{eqns:soln} and \eqref{eqns:sol_gen} for the orbital-orbital entanglement in a two-orbital system, in which the total particle number $N$, total spin $\vec{S}^2$ and magnetization $S^z$ are conserved. In this section we consider the orbital pair as a part of a larger collection of orbitals, and investigate under what condition can \eqref{eqns:soln} and \eqref{eqns:sol_gen} be applied to this orbital pair.

To establish a similar starting point, we consider a system of $K>2$ orbitals, and assume the same conserved quantities as the two-orbital system in Section \ref{sec:symm}, namely $N$, $\vec{S}^2$ and $S^z$. We are interested in the pairwise entanglement between orbital $i$ and $j$ (e.g. $1\leq i < j \leq K$), in the two-orbital reduced state $\rho_{i,j}$ obtained by tracing out all other orbital degrees of freedom in the total state
\begin{equation}
\rho_{i,j} = \Tr_{\setminus \{ i, j \}} [\rho].
\end{equation}
The central question is then what type of symmetry does the two-orbital reduced state $\rho_{i,j}$ exhibit? More precisely and in a more general setting, we wish to know what symmetry does the reduced state $\rho_A$ on a subsystem $A$ (the complementary part denoted as $B$) \textit{inherit} from the total state $\rho$. 

With respect to a bipartition, there are two types of symmetries, local and global. Local symmetries are associated with conserved observables that take the form
\begin{equation}
{Q} = {Q}_A \otimes \mathbbm{1} + \mathbbm{1} \otimes {Q}_B. \label{eqn:local_symm}
\end{equation}
The unitary group generated by ${Q}$ is therefore also local in the sense that its elements are factorized, i.e. $U = U_A \otimes U_B$ where $U_{A/B} = \exp(i\alpha  {Q}_{A/B})$. Then the quantity $Q_A$ is conserved in subsystem $A$ manifested by $U_A \rho_A U_A^\dagger= \rho_A$ which follows directly the assumption $U \rho U^\dagger = \rho$ and the unitary invariance of partial trace. In other words, local symmetries of the total state are naturally inherited by the reduced states, as expected.

The other type of symmetries are associated with conserved quantities that cannot be casted as the form \eqref{eqn:local_symm}, the global ones. These symmetries are in general not inherited by the reduced states. However, we argue in Theorem \ref{thrm:symminherit} that if we further assume that the total state $\rho$ is a singlet, then the global symmetry associated with the conserved total spin $\vec{S}^2$ is also present in the reduced states.

\begin{thrm}\label{thrm:symminherit}
Let $A:B$ be a bipartition of the orbitals. If the total state $\rho$ is a singlet state, namely
\begin{equation}
    \Tr\left[\rho \,  {\vec{S}}^2\right] = 0
\end{equation}
then the reduced state $\rho_A$ satisfies $\left[\rho_A,  {\vec{S}}_A^2\right]  =\left[\rho_{A},  {S}^z_A\right] = 0$. 
\end{thrm}
\begin{proof}
We know that $\rho$ commutes with both $\vec{S}^2$ and $ {S}^z$. For the commutator between the reduced state and local magnetization, we take the partial trace,
\begin{equation}
\begin{split}
        0 = \Tr_B\left[ \left[ \rho,  {S}^z\right] \right] = \Tr_B \left[ \left[ \rho,  {S}^z_A +  {S}^z_B \right] \right] = \Tr_B \left[ \left[ \rho,  {S}^z_A \right] \right] = \left[ \rho_A,  {S}^z_A \right]. \label{eq:Sz}
\end{split}
\end{equation}
$\Tr_B \left[ \left[ \rho,  {S}^z_B \right] \right] = 0$ due to cyclicity of trace. Note that \eqref{eq:Sz} holds even when $\rho$ is not a singlet but still commutes with $\vec{S}_z$. For the total spin, we can write it as,
\begin{equation}
\begin{split}
        \vec{S}^2 =  \sum_{k,l} \vec{S}_k \cdot \vec{S}_l= \vec{S}_A^2 + \vec{S}_B^2 + 2\sum_{k\in A, l \in B} \vec{S}_k \cdot \vec{S}_l,
\end{split}
\end{equation}
where indices $k,l$ are lattice site labels. We then take the partial trace of the commutator $\left[\rho, \vec{S}^2\right]$ on subsystem $B$,
\begin{equation}
\begin{split}
0 = \Tr_B \left[ \left[\rho, \vec{S}^2\right] \right] = \Tr_B \left[ \left[ \rho, \vec{S}_A^2 \right] \right] + \Tr_B \left[ \left[ \rho, \vec{S}_B^2 \right] \right]  + 2\Tr_B \left[ \left[ \rho, \sum_{k\in A, l \in B} \vec{S}_k \cdot \vec{S}_l \right] \right].
\end{split}
\end{equation}
The first term in the last line is the sought after $\left[\rho_A, \vec{S}_A^2\right]$. The second term vanishes again due to the cyclicity of trace. Then we are left with the last term. We rewrite it as
\begin{equation}
    \begin{split}
        \quad \sum_{k\in A, l \in B} \Tr_B \left[ \left[ \rho, \vec{S}_k \cdot \vec{S}_l \right] \right] = \sum_{k\in A, l \in B} \Tr_B \left[ \left[ \rho, S^x_k S^x_l + S^y_k S^y_l + S^z_k S^z_l \right] \right]. \label{eqn:crossterms}
    \end{split}
\end{equation}
Since $\rho$ is a singlet state, we consider $\rho = |\Psi\rangle \langle \Psi|$ where $|\Psi\rangle$ is an eigen-state of $S_z$ with eigenvalue $0$
\begin{equation}
    |\Psi\rangle = \sum_{i} \lambda_i  |s_i;a_i \rangle_A \otimes |\!-\!s_i;b_i \rangle_B.
\end{equation}
$a_i$ and $b_i$ denotes degrees of freedom within the degeneracy classes (e.g., arrangements of spin up and spin down electrons). Then
\begin{equation}
    \begin{split}
        \Tr_B \left[ \rho S^z_B\right] = \sum_{i,j} \lambda_i \overline{\lambda}_j ( - s_i) |s_i;a_i\rangle\langle s_i;a_j| \delta_{s_i,s_j} \delta_{b_i,b_j},
    \end{split}
\end{equation}
which commutes with $ {S}^z_A$. Using
\begin{equation}
    \Tr_B \left[ \left[ \rho, S^z_A S^z_B\right] \right] = \left[ S^z_A, \Tr_B \left[ \rho S^z_B\right] \right]
\end{equation}
we deduce that
\begin{equation}
    \Tr_B \left[ \left[ \rho, S^z_A S^z_B\right] \right] = 0.
\end{equation}
The $x$- and $y$-component spin operator terms in \eqref{eqn:crossterms} vanishes by the same argument, since singlets states are rotationally invariant.
\end{proof}

Apart from symmetries mentioned above, the two-site reduced state $\rho_{i,j}$ of a singlet state also enjoys the spin-flip symmetry, manifested as
\begin{equation}
    \begin{split}
\langle \uparrow\!| \otimes \langle \uparrow\!| \Tr_{\setminus\{i,j\}}[\rho] |\!\uparrow\rangle \otimes |\!\uparrow\rangle &= \langle \uparrow\!| \otimes \langle \uparrow\!| \Tr_{\setminus \{ i, j \}}[(U^\dagger)^{\otimes N}\rho U^{\otimes N}] |\!\uparrow\rangle \otimes |\!\uparrow\rangle
\\
&= \langle \uparrow\!| \otimes \langle \uparrow\!|(U^\dagger)^{\otimes 2} \Tr_{\setminus \{ i, j \}}[\rho] U^{\otimes 2} |\!\uparrow\rangle \otimes |\!\uparrow\rangle
\\
&= \langle \downarrow\!|\otimes \langle \downarrow\!| \Tr_{\setminus \{ i, j \}}[\rho] |\!\downarrow\rangle \otimes |\!\downarrow\rangle,
    \end{split}
\end{equation}
where $U$ is a basis transformation that maps $|\!\uparrow\rangle \otimes |\!\uparrow\rangle$ to $|\!\downarrow\rangle \otimes |\!\downarrow\rangle$ and vice versa. Referring to Table \ref{tab:sym}, this translates to the condition $p_{10} = p_{11}$, which allows us to use the simple formula \eqref{eqns:soln} to calculate the site-site entanglement under N-SSR.

\section{Numerical Method} \label{sec:SDP}

As explain in Section \ref{sec:analytic} calculating the relative entropy of entanglement \eqref{eqn:rel_ent} is in general difficult. However, a few properties of \eqref{eqn:rel_ent} can come to our aid. First of all, the set of separable states $\mathcal{D}_\text{sep}$ is evidently convex. Secondly, the relative entropy of entanglement is convex in both arguments. Therefore we are actually minimizing a convex function over a convex set. If we know the boundary of $\mathcal{D}_\text{sep}$ then \eqref{eqn:rel_ent} can of course be efficiently solved. The complexity in solving \eqref{eqn:rel_ent} originates precisely from the complexity of the boundary of $\mathcal{D}_\text{sep}$. We will show in this section that one can divide this problem into a sequence of convex optimization problem with known optimization boundaries. Then each step can be efficiently solved with the well-developed semidefinite programming\cite{cvx,gb08}. 

\subsection{Semidefinite Programming}

Semidefinite programming (SDP) envelops a wide range of convex optimization problems\cite{vandenberghe1996semidefinite}. Here we present the following form of SDP that is suitable for handling density matrices
\begin{equation}
\begin{split}
    &\mathrm{Minimize} \: F = \Tr[C\rho]
    \\
    &\mathrm{Subject\:to}: \: \Tr[A_i\rho] = b_i, \quad \rho \geq 0, \label{eqn:SDP}
\end{split}
\end{equation}
where $C$ and $A_i$'s are matrices of the same dimension as $\rho$ that encodes the objective function and constraints respectively. $\rho \geq 0$ denotes the condition that $\rho$ is positive semidefinite, namely all eigenvalues of $\rho$ are non-negative.

If we parametrize $\rho$ in terms of its matrix elements, positivity is a highly non-linear constraint. We know, however, that the set of positive semidefinite matrices form a convex set (in fact a convex cone). The additional linear constraints encoded in $A_i$'s are hyperplanes intersecting the set of positive matrices resulting in a smaller convex set. Semidefinite programming makes use of this structure, and search efficiently for the optimal $\rho$ using the so-called interior point method\cite{alizadeh1995interior}. In general when $\rho$ is not restricted to be a quantum state, it can be used to encode non-linear constraints as positivity constraint, e.g.
\begin{equation}
    \rho = \begin{pmatrix}
    A & B\\C&D
    \end{pmatrix}\geq 0 \quad \Leftrightarrow \quad AB-CD \geq 0. 
\end{equation}

The relative entropy as an objective function is slightly problematic. Although it is convex, the matrix logarithm in the relative entropy makes it difficult to be casted as \eqref{eqn:SDP}. Instead, the logarithm is approximated with rational functions to high accuracy\cite{fawzi2019semidefinite}.

\subsection{Algorithm}

We would like to approximate the set of separable states $\mathcal{D}_{\textrm{sep}}$ by relaxing the following condition
\begin{equation}
    \sigma \in \mathcal{D}_{\textrm{sep}} \quad \Leftrightarrow \quad \sigma = \sum_{i}^{\infty} p_i \sigma_A^{(i)} \otimes \sigma_B^{(i)}, \:\: p_i \geq 0, \: \sum_i^{\infty} p_i = 1
\end{equation}
to only finite convex combinations
\begin{equation}
    \sigma \in \mathcal{D}^N_{\textrm{sep}} \quad \Leftrightarrow \quad \sigma = \sum_{i}^{N} p_i \sigma_A^{(i)} \otimes \sigma_B^{(i)}, \:\: p_i \geq 0, \: \sum_i^N p_i = 1.
\end{equation}
That is, $\mathcal{D}_{\textrm{sep}}^N$ contains states that can be written as a finite convex combinations of product states, with at most $N$ positive prefactors $p_i$'s. In fact, in a $D$-dimensional real convex set, any interior points can be written as a convex combination of finitely many points in the generating set of the convex hull.

\begin{thrm}[Carath{\'e}odory]\label{thrm:cara}
Let $P$ be a set in $\mathbb{R}^D$. Let $\Vec{x}$ be a point inside $\mathrm{Conv}(P)$. Then $\Vec{x}$ can be written as a convex combination of at most $D+1$ points in $P$.
\end{thrm}

We apply Theorem \ref{thrm:cara} to the case of quantum states. A general Hermitian matrix of $d\times d$ dimension can be expressed as a real vector in a basis of $D^2$ orthogonal general Gell-Mann matrices $\{G_i\}_{i=0}^{D^2-1}$, with one trace-ful element being the identity matrix, and $D^2-1$ trace-less elements. Then the space of Hermitian operators acting on a Hilbert space of dimension $d$ is isomorphic to $\mathbb{R}^{D^2}$. The set of density matrices is $(D^2-1)$-dimensional, due to the equality constraint of trace,
\begin{equation}
    \rho = \frac{\mathbbm{1}}{d} + \sum_{i=1}^{d^2-1}p_i G_i, \quad 
\end{equation}
For a fixed partition, let $P = \mathcal{D}_0$ be the set of uncorrelated states, i.e. the product states
\begin{equation}
    \mathcal{D}_0 = \{ \sigma \, | \, \sigma = \sigma_A \otimes \sigma_B  \}.
\end{equation}
The set of separable states is $\mathcal{D}_{\textrm{sep}}=\mathrm{Conv}(\mathcal{D}_0)$. We now wish to find the closest separable state $\sigma \in \mathcal{D}_{\textrm{sep}}$ to an entangled state $\rho$. Then by Theorem \ref{thrm:cara} and our observations from above, $\sigma$ can be written as a convex combination of at most $D^2$ product states from $\mathcal{D}_0$.

In fact, we can even improve this further, if the entangled state $\rho$ is real in some reference tensor product basis.

\begin{thrm} \label{thrm:real}
If $\rho$ is real in a reference tensor product basis, then its closest separable state measured by the quantum relative entropy is also real in said basis.
\end{thrm}
\begin{proof}
If $\sigma$ and $\overline{\sigma}$ are both closest separable states to $\rho$, then by convexity of the relative entropy $\text{Re}(\sigma) = (\sigma + \overline{\sigma})/2$ is real and also a closest separable state. Since $\sigma^T = \overline{\sigma}$ for all quantum states $\sigma$, it then suffices to prove that $S(\rho||\sigma) = S(\rho||\sigma^T)$. Then the statement follows the realness of $\rho$ and convexity of the relative entropy on the second argument.
\begin{equation}
    \begin{split}
        S(\rho||\sigma^T) - S(\rho||\sigma) &= \Tr[\rho\log(\sigma)] -  \Tr[\rho^T\log(\sigma^T)]
        \\
        &= \Tr[\rho \log(\sigma)] - \Tr[(\log(\sigma))^T \rho^T]
        \\
        &= 0.
    \end{split}
\end{equation}
In the second equality we used the cyclicity of trace and $\log(\sigma^T) = (\log(\sigma))^T$. In the third equality we used the transposition invariance of trace.
\end{proof}
By Theorem \ref{thrm:real} if $\rho$ is real we can reduce the domain of search to the set of real separable density matrices $\mathcal{R}$ by replacing $P = \mathcal{D}_0$ to $P = \mathcal{D}_0 \cap \mathcal{R}$. In $D$ dimension there are $D(D-1)/2$ imaginary Gell-Mann matrices. Therefore the space of real (in a fixed reference basis) positive matrices is $D(D+1)/2$ dimensional, and the closest separable state $\sigma$ to a real entangled state $\rho$ can be written as a convex combination of at most $D(D+1)/2$ real product states.

The original objective is as follows
\begin{equation}
    \begin{split}
        \textrm{Minimize} \quad & S(\rho||\sigma) = \Tr[\rho(\log(\rho)-\log(\sigma))]
        \\
        \textrm{with }\sigma \textrm{ subject to} \quad & \sigma = \sum_{i} p_i \sigma_A^{(i)} \otimes \sigma_B^{(i)}, \quad p_i\geq 0 \textrm{ and } \sum_{i}p_i = 1
    \end{split}
\end{equation}
However, this constraint on $\sigma$ cannot be directly incorporated into the scheme of semidefinite programming as the constraints are quadratic in the matrix variables due to the tensor product operation. Some efforts have been made by relaxing the set of seperable states $\mathcal{D}_{\textrm{sep}}$ to a larger superset $\mathcal{D}_{\textrm{PPT}}$, namely the positive partial transpose (PPT) states\cite{fawzi2018efficient}. This can be conveniently formulated as semidefinite constraints
\begin{equation}
\begin{split}
    \textrm{Minimize} \quad & S(\rho||\sigma) = \Tr[\rho(\log(\rho)-\log(\sigma))]
        \\
        \textrm{with }\sigma \textrm{ subject to} \quad & \sigma \geq 0 \quad \textrm{and} \quad \sigma^{T_B} \geq 0.
\end{split}
\end{equation}
Due to the existence of entangled PPT states\cite{horodecki1997separability,bennett1999unextendible} in dimensions $D > 6$, the quantity
\begin{equation}
    E_{\textrm{PPT}}(\rho) = \min_{\sigma \in \mathcal{D}_{\textrm{PPT}}} S(\rho||\sigma)
\end{equation}
is in general a lower-bound for the exact relative entropy of entanglement $E_{\textrm{RE}}$ (in this section we use the notation $E_\textrm{RE}$ for the relative entropy of entanglement to avoid confusion. The subscript is dropped in other sections). It is worth noting that $E_{\textrm{RE}} = E_{\textrm{PPT}}$ in the case $D = 2\times2$ and $D=2\times3$\cite{horodecki1997separability,peres1996separability}.

We propose the following optimization scheme for approximating $E_{\textrm{RE}}$.

\begin{algorithm}[H] \label{alg:zigzag}
\SetAlgoLined
\KwResult{Find the closest separable state $\sigma^\ast$}
initialize $\{\sigma_{A/B}^{(i)}\}_{i=1}^{D^2}$ as positive matrices

\While{$\textrm{improvement}>\textrm{tolerance}$}{
minimize $S(\rho||\sum_{i=1}^{D^2}\sigma_A^{(i)}\otimes\sigma_B^{(i)})$ w.r.t. $\{\sigma_B^{(i)}\}_{i=1}^{D^2}$  with $\sigma_B^{(i)}\geq0$ and $\Tr[\sum_{i=1}^{D^2}\sigma_A^{(i)}\otimes\sigma_B^{(i)}]=1$;
\\
minimize $S(\rho||\sum_{i=1}^{D^2}\sigma_A^{(i)}\otimes\sigma_B^{(i)})$ w.r.t. $\{\sigma_A^{(i)}\}_{i=1}^{D^2}$  with $\sigma_A^{(i)}\geq0$ and $\Tr[\sum_{i=1}^{D^2}\sigma_A^{(i)}\otimes\sigma_B^{(i)}]=1$;
 }
 $\sigma^\ast \leftarrow \sum_{i=1}^{D^2}\sigma_A^{(i)}\otimes\sigma_B^{(i)}$
 \caption{Minimizing $S(\rho||\sigma)$ w.r.t. $\sigma$ for a fixed $\rho$}
\end{algorithm}

In theory, Theorem \ref{thrm:cara} provides the argument that the closest separable state $\sigma$ can be (in principle exactly) found as such a finite convex combination. In practice the result of Algorithm \ref{alg:zigzag} is an upper-bound for the exact $E_{\textrm{RE}}$, as any $\sigma^\ast$ it finds (exact minimizer or not) is by construction a separable state. Since each minimization step can be cast as a semidefinite program, this algorithm can be realized efficiently. Typically the search for the closest separable state stops when no improvement of the relative entropy can be made beyond the set tolerance of accuracy, and the algorithm returns a local minimum. To approximate the global minimum we choose a range of random starting points.

The first thing we would like to check is whether Algorithm \ref{alg:zigzag} reproduces the known boundary of $\mathcal{D}_{\textrm{sep}}$ in the case of two qubits. We consider the following states
\begin{equation}
    \rho_p = p \frac{\mathbbm{1}}{4} + (1-p) |\Psi_+\rangle\langle \Psi_+|,\quad p\in[0,1],
\end{equation}
where $|\Psi_+\rangle$ is one of the Bell states $|\Psi_+\rangle = (|00\rangle+|11\rangle)/\sqrt{2}$. As $p$ goes from $0$ to $1$, $\rho_p$ morphs from the maximally entangled state $|\Psi_+\rangle\langle\Psi_+|$ to the maximally mixed state $\mathbbm{1}/4$, and it should cross the boundary of $\mathcal{D}_{\textrm{sep}}$ at some value of $p$. At that point the entanglement of $\rho_p$ drops to zero.

\begin{figure}[ht]
    \centering
    \includegraphics[scale=0.4]{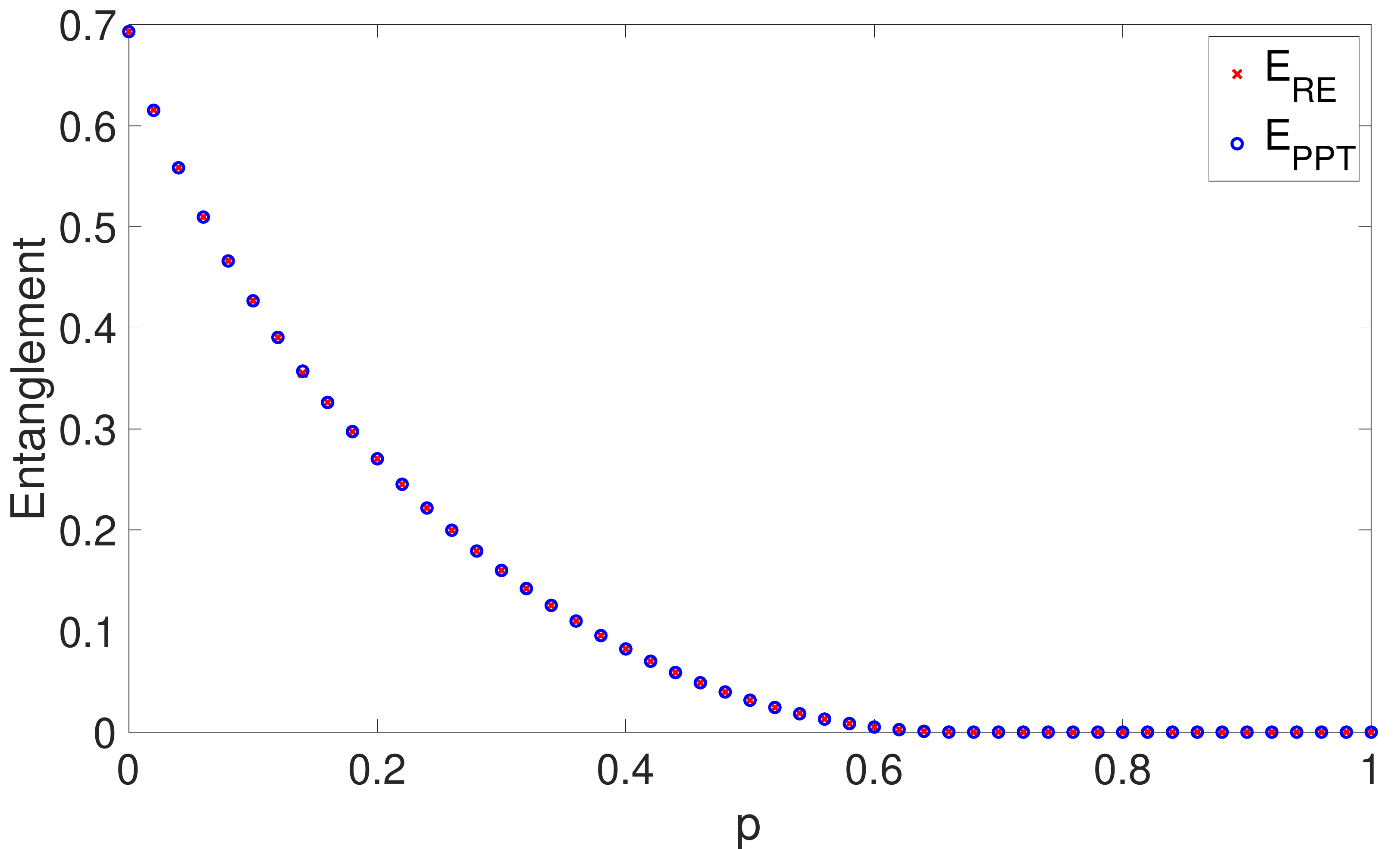}
    \caption{$E_{\textrm{RE}}(\rho_p)$ and $E_{\textrm{PPT}}(\rho_p)$ as functions of $p$.}
    \label{fig:EntREvsEntPPT}
\end{figure}

In Figure \ref{fig:EntREvsEntPPT} we plotted $E_{\textrm{RE}}(\rho_p)$ calculated using Algorithm \ref{alg:zigzag} and $E_{\textrm{PPT}}(\rho_p)$ obtained by semidefinite optimization against the parameter $p$. As we can see the two quantities match perfectly, and both drop to zero from $p = 0.68$ on (with minimum spacing $0.02$).

Now we turn to a certain family of entangled PPT states, the Horodecki bound entangled states\cite{horodecki1997separability} in $3\times 3$ dimension.
\begin{equation}
    \rho_a = \frac{1}{8a+1} \begin{pmatrix} a & 0 & 0 & 0 & a & 0 & 0 & 0 & a \\ 0 & a & 0 & 0 & 0 & 0 & 0 & 0 & 0 \\ 0 & 0 & a & 0 & 0 & 0 & a & 0 & 0 \\ 0 & 0 & 0 & a & 0 & 0 & 0 & 0 & 0 \\ a & 0 & 0 & 0 & a & 0 & 0 & 0 & a \\ 0 & 0 & 0 & 0 & 0 & a & 0 & 0 & 0 \\ 0 & 0 & a & 0 & 0 & 0 & \frac{1+a}{2} & 0 & \frac{\sqrt{1-a^2}}{2} \\ 0 & 0 & 0 & 0 & 0 & 0 & 0 & a & 0 \\ a & 0 & 0 & 0 & a & 0 & \frac{\sqrt{1-a^2}}{2} & 0 & \frac{1+a}{2} \end{pmatrix}, \quad a \in [0,1].
\end{equation}
This family of density matrices are constructed to be both PPT and entangled. Therefore $E_{\textrm{PPT}}(\rho_a)=0$ for all $a \in [0,1]$, yet their entanglement can be revealed by faithful measures such as $E_{\textrm{RE}}$. It was shown by approximations of other measures\cite{audenaert2001variational,chen2002matrix} that entanglement is maximal when $a$ is somewhere between $0.225$ and $0.236$.

\begin{figure}[ht]
    \centering
    \includegraphics[scale=0.4]{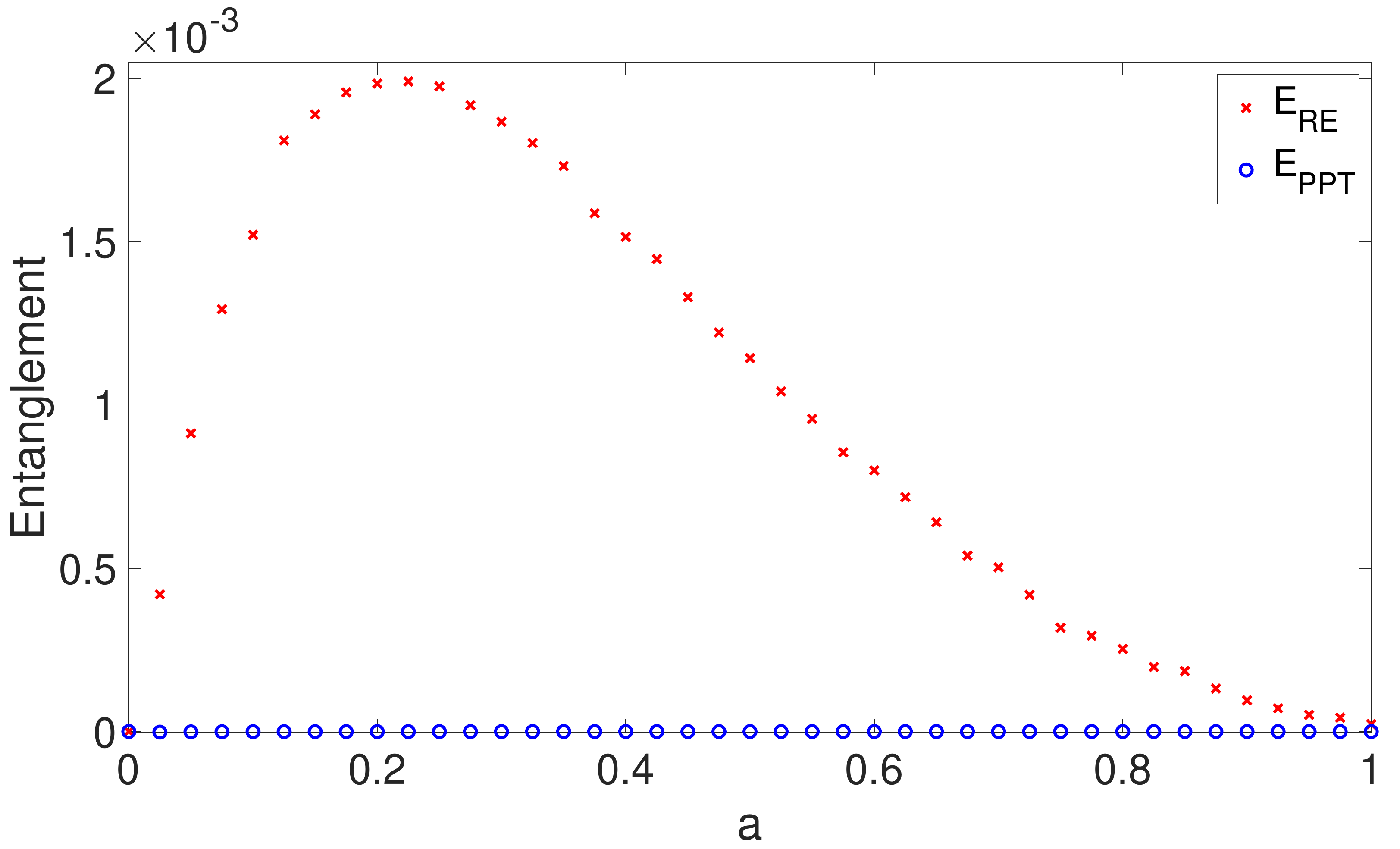}
    \caption{$E_{\textrm{RE}}(\rho_a)$ and $E_{\textrm{PPT}}(\rho_a)$ as functions of $a$.}
    \label{fig:BoundEnt}
\end{figure}

In Figure \ref{fig:BoundEnt} we plotted $E_{\textrm{RE}}$ and $E_{\textrm{PPT}}$ of the Horodecki bound entangled states. We can see that for $a \in (0,1)$ $E_{\textrm{RE}}$ detected finite bound entanglement. Although there are no reference for quantitative comparison, the shape of the curve and location of maximum agree with other findings\cite{audenaert2001variational,chen2002matrix} with different entanglement measures. Maximum occurs at $a=0.225$ (minimum spacing $0.025$), standing at $E_{\textrm{RE}} = 1.99\times10^{-3}$ ($1.81\times10^{-3}$ for normalised $E_\textrm{RE}$). On one hand previous findings showed that the entanglement of formation has a peak at about $E_{\textrm{F}} = 0.0109$\cite{audenaert2001variational}. On the other hand since $\rho_a$'s are PPT states they have zero distillable entanglement\cite{peres1996separability,horodecki1998mixed}, i.e. $E_{\textrm{D}} = 0$. The exact values of the normalised relative entropy of entanglement lie in between $E_{\textrm{D}}$ and $E_{\textrm{F}}$ (assuming the conjecture $E_F = E_C$ is correct)\cite{horodecki2000limits}. The fact that such curve can be obtained shows strong promise that Algorithm \ref{alg:zigzag} can quantitatively discern PPT entanglement measured by the relative entropy of entanglement with accuracy at least of order $10^{-4}$.

\chapter{Applications} \label{chap:App}

Having established proper notions of fermionic particle and mode correlation/entanglement in Chapter \ref{chap:fermion} and derived analytic formula for quantifying mode entanglement in Chapter \ref{chap:quantifying}, we are ready to apply these results to concrete systems. In Section \ref{sec:diss} (containing contents published on \textit{Journal of Chemical Theory and Computation} \cite{ding2020correlation}) we investigate the correlation paradox of the dissociation limit. In Section \ref{sec:qchem} (containing contents published on \textit{Journal of Chemical Theory and Computation} \cite{ding2020concept}) we apply our results to molecules, and study various type of correlations between molecular orbitals. 

\section{Correlation Paradox of the Dissociation Limit} \label{sec:diss}

Let us first recall the correlation paradox from a more general point of view. Whenever identical fermions do not interact, solving the $N$-particle Schr\"odinger equation simplifies to an effective one-fermion problem. Indeed, for any Hamiltonian $H\equiv \sum_{i,j=1}^d h_{ij} f_i^\dagger f_j $ one just needs to diagonalize the Hermitian matrix $(h_{ij})$, leading to $H \equiv \sum_{\alpha=1}^d h_\alpha \hat{n}_\alpha$ with some one-particle solutions $\ket{\alpha}$. The respective $N$-fermion eigenstates follow as configuration states $\ket{\alpha_1,\ldots,\alpha_N}\equiv \ket{\alpha_1}\wedge \ldots \wedge\ket{\alpha_N}$ obtained by distributing the $N$ fermions into $N$ different spin-orbitals $\ket{\alpha_i}$. Having said this, how can a non-degenerate fermionic ground state not take the form of a single configuration state in a limit process which marginalizes the interaction between the fermions? The existence of exactly such processes can be seen as paradoxical in that sense.

The dissociation limit of molecules often gives rise to such paradoxical situations which play an important role in the context of the general electron correlation problem\cite{Low58,Pop76,Bart78,Pop87,Mazz12}.
They are of course well-understood in quantum chemistry, in particular on a qualitative level. For instance, it is rather obvious that those paradoxical situations require the closing of the excitation gap $\Delta E(r)$ and at the limit $r\rightarrow \infty$ the system needs to have several configuration states as degenerate ground states. For very large but not infinite separation distances $r$ between the nuclei, those configurations are then typically superposed to form the non-degenerate correlated ground state. From the most rudimentary point of view, the paradox could therefore be resolved by just referring to the excitation gap $\Delta E(r)$ which reduces to zero at least as fast as the electron-electron interaction energy vanishes.

Yet, there is more to be said. For instance, why would one like to construct a measure for correlation\cite{Ziesche97} which vanishes for the dissociated hydrogen ground state
\begin{equation}\label{PsiH2inf}
    |\psi\rangle = \frac{1}{\sqrt{2}}(f^\dagger_{L\uparrow} f^\dagger_{R\downarrow} - f^\dagger_{L\downarrow} f^\dagger_{R\uparrow})|\Omega\rangle
\end{equation}
where $f^\dagger_{L/R \sigma}$ are creation operators associated with mode on the left/right nuclei with spin $\sigma$ and $|\Omega\rangle$ the vacuum state, despite the fact that the latter cannot be written as a single configuration state? It seems to us that there are partly self-contradicting definitions in place for what ``correlation'' actually might or should be:
On the one hand, a state is considered as being ``uncorrelated'' if it takes the form of as a single configuration state. On the other hand, one observes that the electron-electron interaction vanishes in the dissociation limit despite the fact that the ground state is not a configuration state. This apparent contradiction is based on a confusion between the notion of total correlation and the concept of correlation functions. Furthermore, how does the dissociated ground state \eqref{PsiH2inf} compare to the uncorrelated degenerate configuration states emerging at the limit $r\rightarrow \infty$ in terms of its robustness to perturbations? We will provide an answer to the latter question in Section \ref{sec:resol}. To be more specific, we illustrate and prove that thermal noise due to finite, possibly even just infinitesimally low, temperature $T$ will destroy the quantum correlations beyond a critical separation distance $r_{\mathrm{crit}}$($T$) entirely.
This rationalizes that ``correlation'' vanishes in the dissociated ground state in the sense that this perception is correct provided the presence of some (possibly infinitesimally low) temperature $T>0$. These considerations which are made precise in Section \ref{sec:resol}
reveal a conceptually new characterization of static and dynamic correlation in ground states by relating them
to the (non)robustness of correlation with respect to thermal noise.

\subsection{The Dissociated Hydrogen Molecule}\label{sec:dimer}

From a general point of view, the realization of the dissociation limit of the hydrogen molecule (or any other molecular system) in the laboratory requires the coupling of the molecule to another system. To present our theoretical argument on the (in)stability of correlation/entanglement with respect to thermal noise in the cleanest fashion we consider an experimental procedure which accesses the nuclei directly to moves them apart. In that sense, it also freezes the nuclear (vibrational) degrees of freedom and the Born-Oppenheimer approximation with a separation distance $r$ of both nuclei will be assumed. To discuss such realizations of the dissociation limit of the hydrogen molecule we thus begin with the electronic Hamiltonian, i.e., we consider two interacting electrons in the Coulomb potential generated by two nuclei separated by a distance $r$. Choosing large basis sets of atomic orbitals centered around both nuclei would allow one to obtain highly accurate descriptions of the behavior and the properties of the hydrogen molecule. Yet, in our case we restrict ourselves to very low temperatures and thus only the $1s$ orbital needs to be taken into account for capturing the main effects. This is due to the fact that the energy difference between the 1s and the higher shells significantly exceeds the energy scale of the electron-electron energy in atoms. After all, this approximation is getting exact in the limit of arbitrarily large separation distances $r$
since then the two electrons are getting arbitrarily far separated (and in particular the probability of finding them at the same nucleus tends to zero).

As a consequence, we can study the most relevant aspects of the dissociation limit of the hydrogen molecule in the Hubbard dimer model\cite{alvarez2001hubbard,chiappe2007hubbard}. This (and after all our initial choice of a simple two-electron system) will allow us to illustrate all relevant quantum information theoretical aspects without getting deflected by highly involved descriptions of correlated ground states.
From a general point of view, the Hubbard dimer is one of the simplest models for interacting fermions, while already exhibiting rich physical properties. It consists of two lattice sites ($L$ and $R$) corresponding to the 1s orbitals centered at both nuclei and the underlying Hamiltonian takes the form
\begin{equation}\label{Hubbard}
    H = - t\!\sum_{\sigma = \uparrow, \downarrow} (f^\dagger_{L\sigma} f_{R\sigma } + f^\dagger_{R\sigma} f_{L\sigma} ) + U \sum_{i=L,R}\hat{n}_{i\uparrow}\hat{n}_{i\downarrow}.
\end{equation}
Here, $t\geq 0$ describes the hopping between the both nuclei/sites, $U>0$ represents the on-site repulsion (originating from the Coulomb interaction between two electrons in a 1s shell) and $\hat{n}_{L/R}$ denotes the particle number operator at the left/right site. Since the eigenstates of \eqref{Hubbard} depend only on the ratio $t/U$  we set in the following $U\equiv 1$. Moreover, the hopping $t$ decays exponentially as function of $r$, in agreement with the overlap of two 1s atomic orbitals separated by a distance $r$\cite{krauss1965interaction}. Depending on the context, we will choose in the following either $r$ or $t=e^{-r}$ as the parameter of the system.

\subsection{Exact Diagonalisation of the Hubbard Dimer}

To diagonalize the Hamiltonian \eqref{Hubbard} of the Hubbard dimer it is instructive to exploit
its spin symmetries and the reflection symmetry $L \leftrightarrow R$. Those manifest themselves in the form of the total spin $S$, the magnetization $M$ along the $z$-axis and the refection parity $p$ as good quantum numbers. The corresponding eigenvalue problem decouples according to
\begin{equation}
H = \bigoplus_{S=0}^1 \bigoplus_{M=-S}^S \bigoplus_{p=\pm} H_{S,M,p}.
\end{equation}
As a matter of fact, this almost completes the diagonalization and it remains to diagonalize $H_{0,0,-}$ on the corresponding two-dimensional space $\mathcal{H}_{0,0,-}$. The details of those calculations are presented in Appendix \ref{spectrum} and we just present here the well-known results for the six eigenenergies \cite{hasegawa2011thermal} (with $U\equiv 1$)
\begin{eqnarray}\label{dimerspectrum}
&&E_0 = \frac{1}{2} - \sqrt{\frac{1}{4}+4t^2}\,, \quad E_1 = E_2 = E_3 = 0, \nonumber \\
&&E_4 = 1, \quad E_5 = \frac{1}{2} + \sqrt{\frac{1}{4}+4t^2}.
\end{eqnarray}
It is crucial to notice that the ground state is always non-degenerate and the first excited energy corresponds to the threefold degenerate  triplet states. The energy spectrum \eqref{dimerspectrum} is also shown in Figure \ref{HubbardSpectrum} (recall $t = \exp(-r)$) and the corresponding six eigenstates are listed in Appendix \ref{spectrum}.
\begin{figure}[h!]
    \centering
    \includegraphics[scale=0.35]{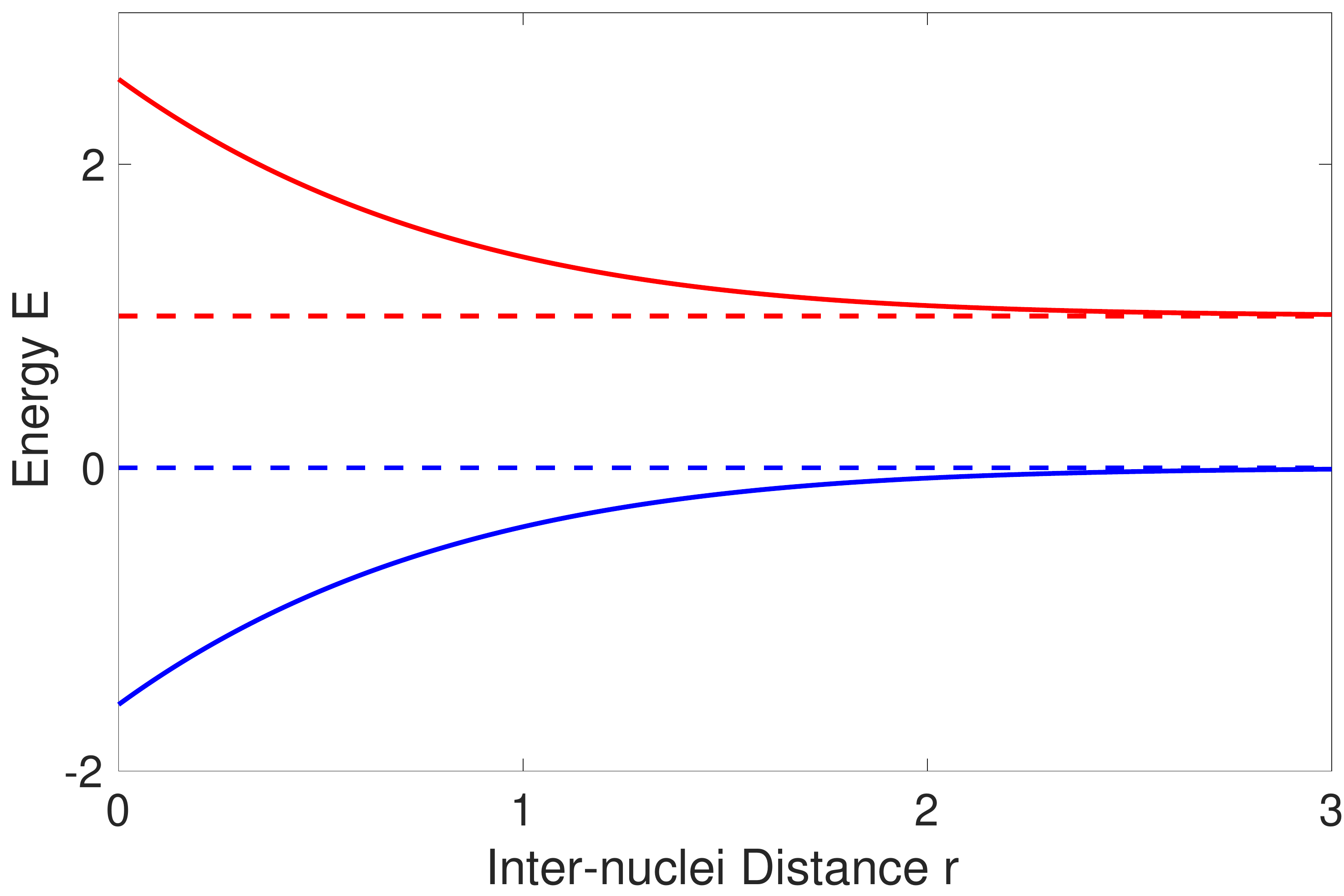}
    \caption{Energy spectrum \eqref{dimerspectrum} of the Hubbard dimer \eqref{Hubbard} in dimensionless units ($U$ and the Bohr radius are set to one). All energy levels are non-degenerate except for the first excited level (blue dashed), where three triplet states reside.}
    \label{HubbardSpectrum}
\end{figure}
In particular this confirms the closing of the excitation gap for $r\rightarrow \infty$, as described by $\Delta E(r) \sim 4 t^2=4 e^{-2 r}$.

At temperature \(T=0\) the system takes the energetically favorable ground state,
\begin{eqnarray}\label{ground}
    |\Psi_0(r)\rangle &=& \frac{a(r)}{\sqrt{2}}\big(f^\dagger_{L\uparrow}f^\dagger_{R\downarrow} -  f^\dagger_{L\downarrow} f^\dagger_{R\uparrow}\big)\ket{0} \nonumber  + \frac{b(r)}{\sqrt{2}} \big(f^\dagger_{L\uparrow}f^\dagger_{L\downarrow} -  f^\dagger_{R\downarrow} f^\dagger_{R\uparrow}\big) \ket{0},
\end{eqnarray}
where the coefficients $a(r)$ and $b(r)$ are functions of the inter-nuclear distance $r$ (explicit expressions can be found in Appendix \ref{spectrum}). In particular, one has $a(r)=\sqrt{1-b^2(r)}$ and $b(r)\sim 2 \,t = 2 e^{-r}$ for $r \rightarrow \infty$. The latter confirms
that the probability of finding both electrons at the same site/nucleus tends to zero for large separation distances $r$ and small hopping rates $t$, respectively.
Consequently, at the limit \(r \rightarrow \infty\), the ground state follows indeed as
\begin{equation}
    |\Psi_0(r=\infty)\rangle = \frac{1}{\sqrt{2}}(f^\dagger_{L\uparrow}f^\dagger_{R\downarrow} - f^\dagger_{L\downarrow} f^\dagger_{R\uparrow})|0\rangle,
\end{equation}
which is \textit{not} a configuration state.

At finite temperature, the state of interest is the thermal Gibbs state (we set for simplicity $k_B \equiv 1$),
\begin{equation}
    \rho(T,r) = \frac{1}{Z(T,r)} e^{-H(r)/ T} , \label{Gibbs}
\end{equation}
where $Z(T,r)\equiv \Tr\left[e^{-H(r)/T}\right]$ is the partition function.
In addition to the standard Boltzmann-Gibbs statistics which we use here, there have also been proposals of other distributions for systems of non-extensive size\cite{salinas1999special,hasegawa2011thermal}. Although it is somewhat debatable to say which statistics is more appropriate for a small system like ours, we shall stick to the Boltzmann-Gibbs distribution and use the thermal equilibrium state as defined in Eq.~\eqref{Gibbs}.

\subsection{Resolution of the Correlation Paradox} \label{sec:resol}

\subsubsection{Mode Picture}

We are now in a position to calculate the mode correlation and mode entanglement in the Gibbs state in Eq.~\eqref{Gibbs} for all temperatures $T$ and all separation distances $r$. The respective results for the cases $T=0, 0.1$ are presented
in Figure \ref{fig:ModeCorr}.

\begin{figure}[H]
    \centering
    \includegraphics[scale=0.35]{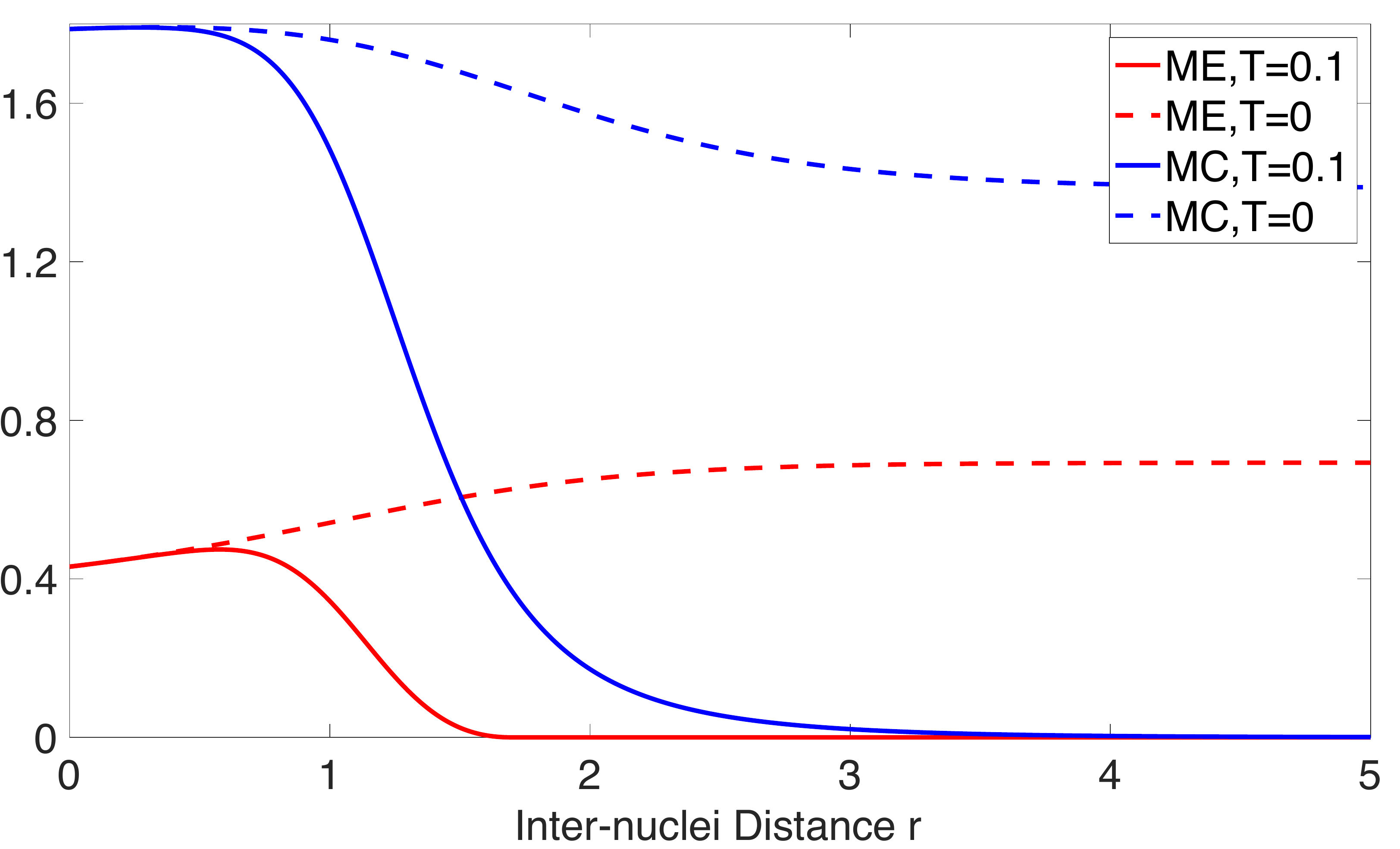}
    \caption{Mode correlation (MC, blue) and mode entanglement (ME, red) plotted for the Gibbs state (solid) at finite temperature $T=0.1$ and the ground state (dashed), the equilibrium state at $T=0$, with particle number superselection rule.}
    \label{fig:ModeCorr}
\end{figure}

First, we observe that both mode correlation and mode entanglement in the Gibbs state with $T=0.1$ and the ground state ($T=0$) coincide at smaller distances $r$. This is due to the fact that for small $r$ the energy gap $\Delta E(r)$ between the ground state and the first excited states is much larger than the thermal energy scale $k_B T$ such that both states essentially coincide (the contribution of the excited states to the Gibbs ensemble are exponentially suppressed according to Eq.~\eqref{Gibbs}).
Second, the presence of a correlation paradox is confirmed since the mode correlation (blue dashed) and mode entanglement (red dashed) of the ground state remain finite even at the dissociation limit. Third, for finite temperature, this is quite different. When the inter-site distance $r$ becomes larger, and the gap $\Delta E(r)$ between ground state and first excited state closes, both correlation (blue solid) and entanglement (red solid) at finite temperature start to deviate more from the ground state ones. They get smaller and smaller, and they are eventually completely wiped out at the dissociation limit. This asymptotic behavior at $r \rightarrow \infty$ is present at any finite temperature $T>0$, regardless of how small it is. In particular, this means that the mode correlation of the ground state is highly unstable against thermal noise, and finite mode entanglement or mode correlation at the dissociation limit can \textit{never} be observed in a laboratory.

Remarkably, in Figure \ref{fig:ModeCorr} the mode entanglement in the Gibbs state drops to zero already at a \emph{finite} distance, $r_{crit}^{(m)}(T=0.1)=1.70$, unlike the usual asymptotic behavior of correlation. In other words, for any temperature $T$, there exists a minimal distance $r_{\mathrm{crit}}^{(m)}(T)$ beyond which the mode entanglement vanishes entirely. Such a decaying  behavior of the entanglement, sometimes referred to as ``sudden death'', is not uncommon in quantum systems\cite{yu2009sudden}, and is a unique feature of quantum correlation. Fascinating as it is, this finite parameter point at which the entanglement vanishes is nothing mysterious if one considers the geometric picture as shown in Figure \ref{fig:states}: The Gibbs state $\rho(T,r)$ simply entered the convex set of separable states as the inter-nuclei distance $r$ increases. In fact, because of this, the point $r_{\mathrm{crit}}^{(m)}(T)$ is of course independent of the measure employed for quantifying the entanglement.
\begin{figure}[h!]
    \centering
    \includegraphics[scale=0.40]{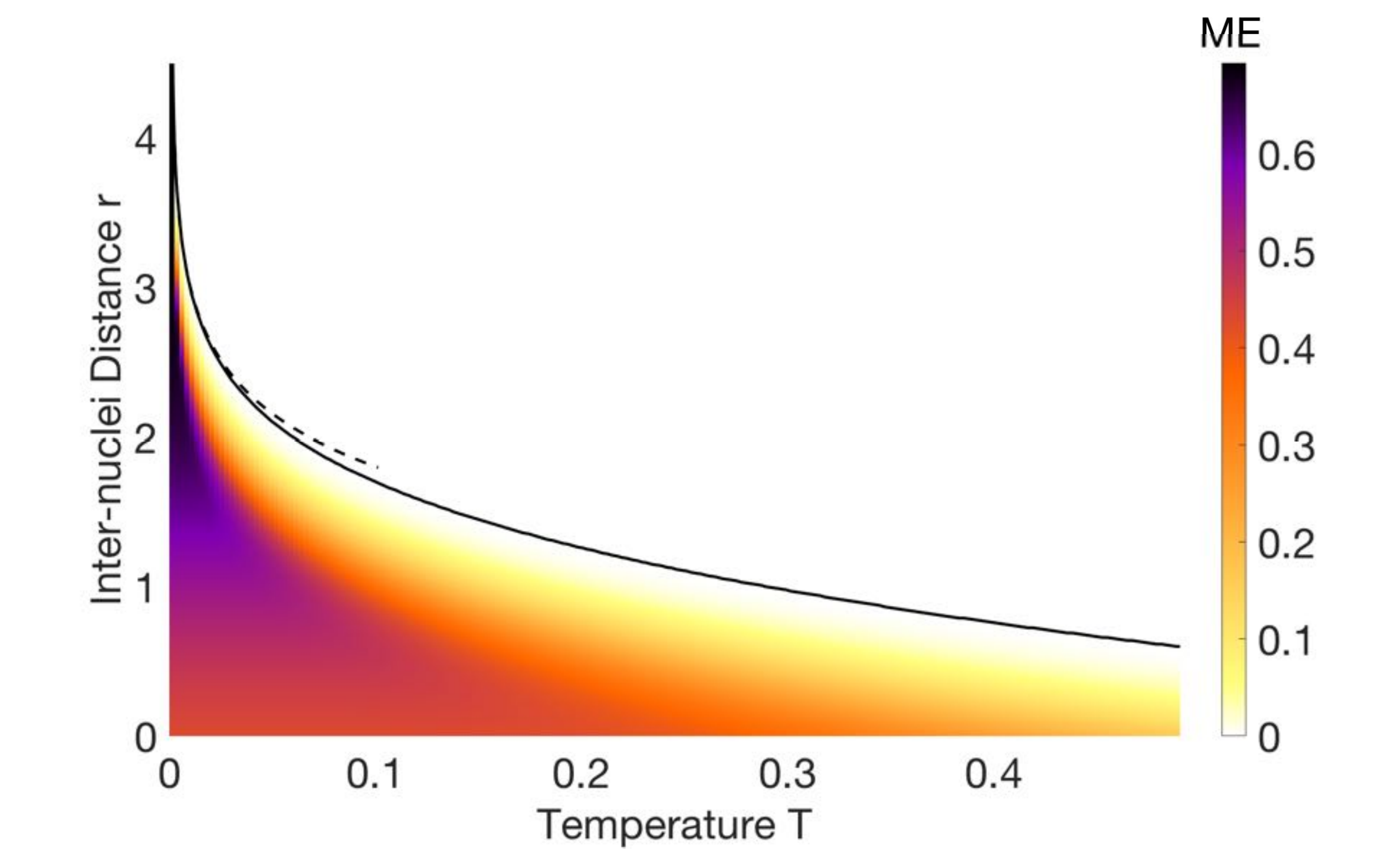}
    \caption{Mode entanglement (ME) as a function of temperature $T$ and inter-site distance $r$. It vanishes entirely above the black curve  $r_{\mathrm{crit}}^{(m)}$. The dashed line represents the asymptotic result \eqref{rm} for small $T$.}
    \label{fig:ME_sudden}
\end{figure}

To see how temperature affects this phenomenon, we present the mode entanglement in Figure \ref{fig:ME_sudden} as a function of the temperature $T$ and the inter-nuclei distance $r$. The critical distance $r_{\mathrm{crit}}^{(m)}(T)$ is shown as black curve.
For all parameter points $(T,r)$ above the black curve the mode entanglement vanishes, while it is finite for all points below it. As the temperature increases, the minimum distance $r_{\mathrm{crit}}^{(m)}(T)$ required to disentangle the left and right side becomes smaller. When $T \rightarrow 0$, the Gibbs state $\rho(T,r)$ approaches the ground state $\ket{\Psi_0}\!\bra{\Psi_0}$, and $r_{\mathrm{crit}}^{(m)}$ approaches infinity. In fact, the divergence of $r_{\mathrm{crit}}^{(m)}$ at small $T$ is logarithmic,
\begin{equation}
    r_{\mathrm{crit}}^{(m)}(T) = - \frac{1}{2} \log(T) + c_0 +c_1 T+\mathcal{O}(T^2), \quad T \rightarrow 0, \label{rm}
\end{equation}
where $c_0\equiv \log(2) - \frac{1}{2} \log(\log(3))$, $c_1\equiv -\frac{1}{2}(1+\log (3))$ are constants.
This asymptotic result is shown as dashed line in Figure \ref{fig:ME_sudden} and its derivation is included in Appendix \ref{divergence} for the interested readers.

\subsubsection{Particle Picture}

In the particle picture we recall that the analogues for correlation and entanglement are the nonfreeness and quantum nonfreeness, respectively. To determine the nonfreeness we just need to calculate the one-particle reduced density matrix of the state \eqref{Gibbs} and plug it into the formula \eqref{PartCorr}. To calculate the quantum nonfreeness we can resort to the analytic procedure outlined in Section \ref{sec:part} since our model consists indeed of two fermions and a four-dimensional one-particle Hilbert space. We also would like to recall that in contrast to the other measures employed in our work the respective measure \eqref{PartEnt} for the quantum nonfreeness is not of the form \eqref{eqn:rel_ent}. It namely does not involve the relative entropy as a distance function and is based on a so-called convex roof construction instead. Nonetheless, the used measure for the quantum nonfreeness quantifies how close a state is to the convex set $\mathcal{D}^{(p)}_{sep}$ given as the convex hull of single configuration states.

\begin{figure}[h!]
    \centering
    \includegraphics[scale=0.35]{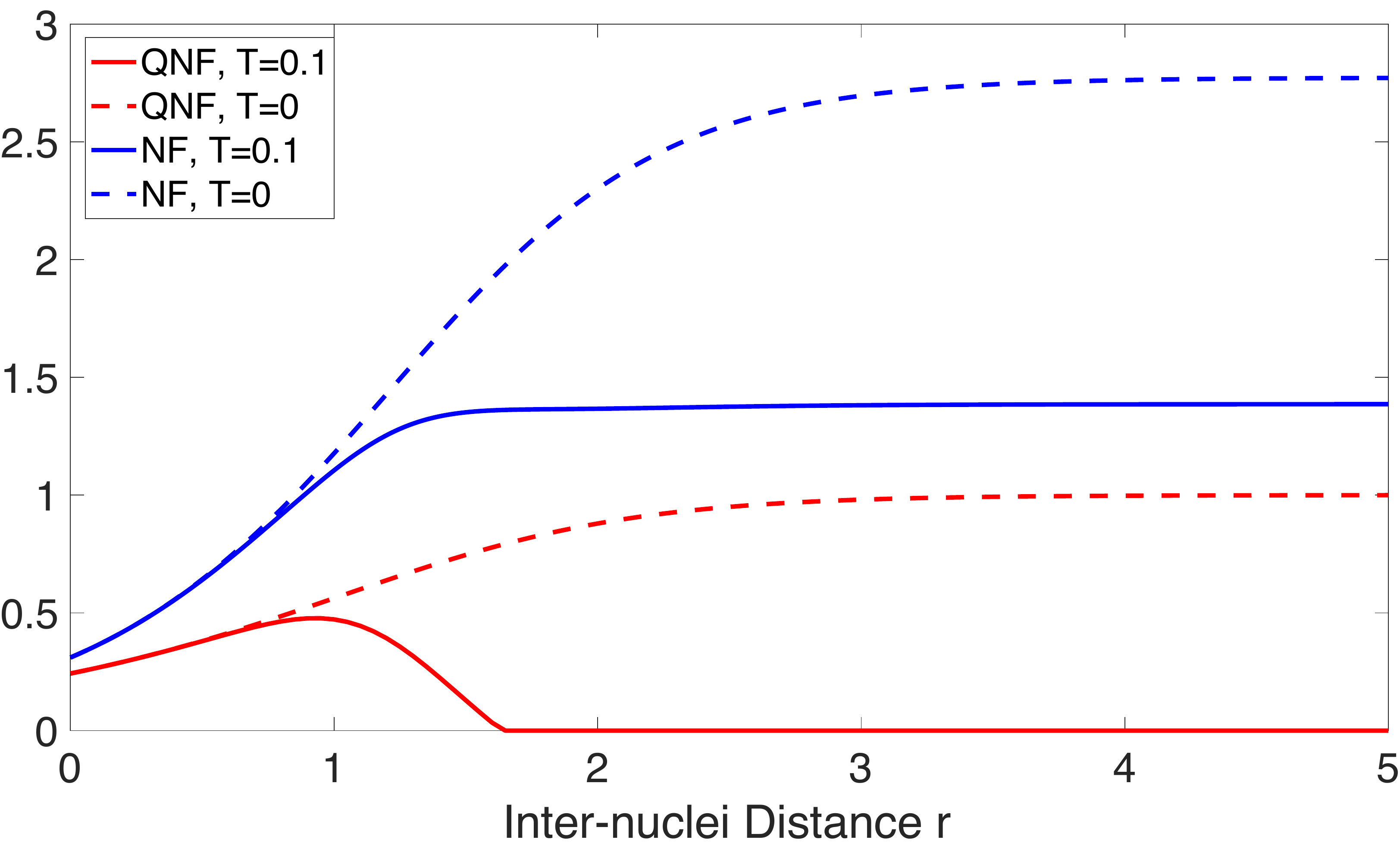}
    \caption{Nonfreeness (NF, blue) and quantum nonfreeness (QNF, red) as a function of inter-site distance. The correlation and entanglement in the ground state at zero temperature are also plotted (in dashed line).}
    \label{fig:ParticleCorr}
\end{figure}

In Figure \ref{fig:ParticleCorr} we present the nonfreeness \eqref{PartCorr} (blue) and its quantum part \eqref{PartEnt} (red) for the Hubbard dimer. 
The results for finite temperature $T=0.1$ are represented by the solid lines, and the dashed lines are reserved for the ground state Eq.~\eqref{ground} ($T=0$). As already discussed in the previous section, the Gibbs state at sufficiently low temperature approximately coincides with the ground state for smaller $r$, and therefore the (quantum) nonfreeness of both states approximately coincide as well. Things become very interesting as the two nuclei move further apart. First of all, the nonfreeness is reduced by introducing a small temperature, but remains finite at the dissociation limit. To be more specific, we already know that for any finite $T>0$, the Gibbs state approximates in the limit $r\rightarrow \infty$ better and better an equally weighted classical mixture of four configuration states,
\begin{equation}\label{PsiMix}
\rho(T,r) \approx \frac{1}{4}\sum_{\sigma,\sigma'=\uparrow/\downarrow} \ket{L\sigma,R\sigma'}\!\bra{L\sigma, R\sigma'}.
\end{equation}
This is also reflected by the fact that the 1RDM is perfectly mixed,
\begin{equation}
    \rho_1=\frac{1}{2}\mathbbm{1}_4= \frac{1}{2}\sum_{i=L/R}\sum_{\sigma=\uparrow/\downarrow} \ket{i\sigma}\!\bra{i\sigma}.
\end{equation}
This means that it is equally probable to find an electron on left or right, which has spin up or down. Moreover, as it can directly been inferred from the purely classical mixture \eqref{PsiMix} of configuration states, the quantum part of the nonfreeness decays to zero as we increase the inter-nuclei distance $r$. Remarkably, also the quantum nonfreeness in the Gibbs state experiences a ``sudden death'' as the mode entanglement, at a critical distance $r_{\mathrm{crit}}^{(p)}(T=0.1)=1.65$. As pointed out before, this phenomenon is a unique feature of quantum correlation, and it emphasizes that the quantum nonfreeness \eqref{PartEnt} captures something truly non-classical.
\begin{figure}[h!]
    \centering
    \includegraphics[scale=0.40]{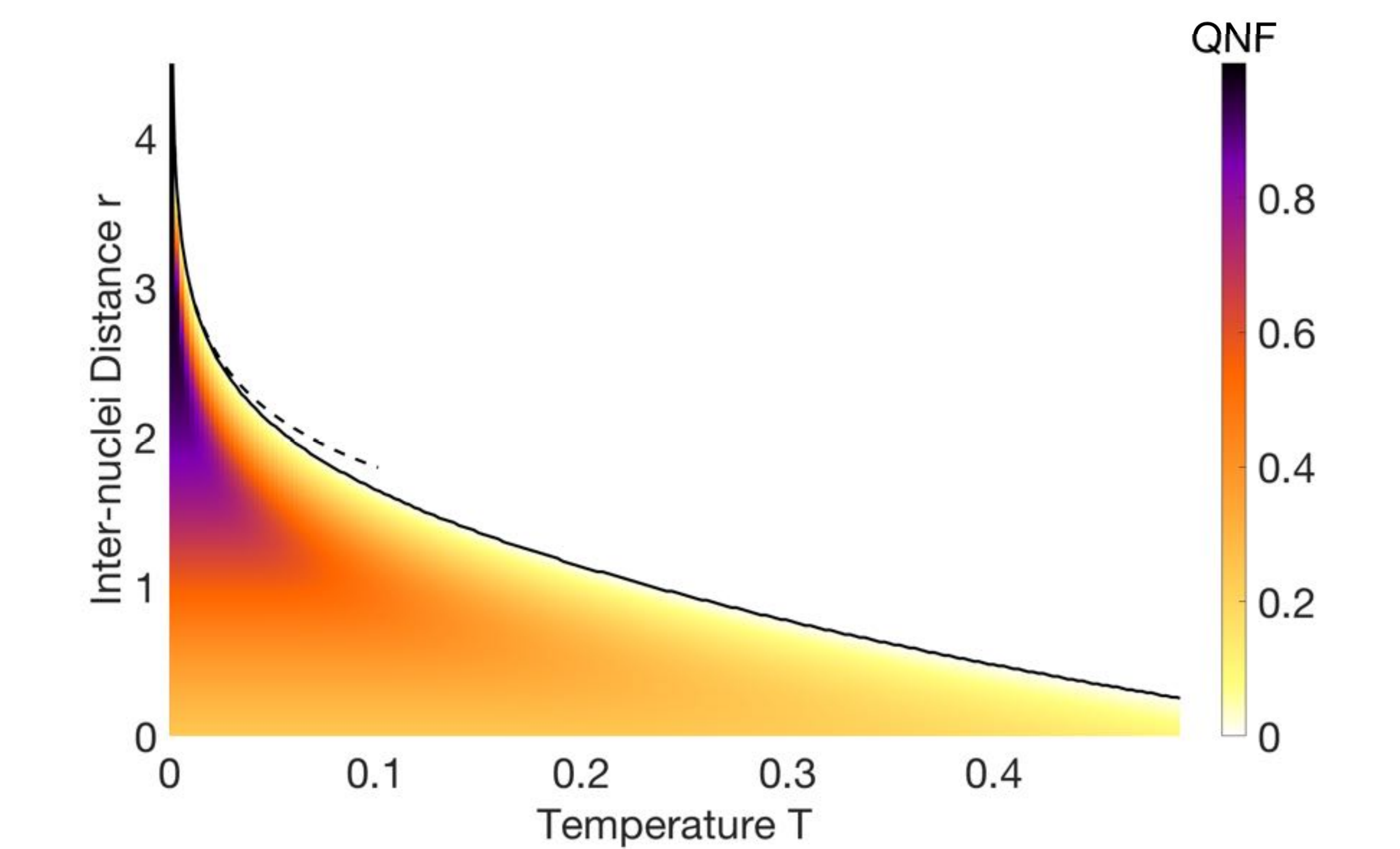}
    \caption{Quantum nonfreeness (QNF) as a function of temperature $T$ and inter-site distance $r$. It vanishes entirely above the black curve  $r_{\mathrm{crit}}^{(p)}$. The dashed line represents the asymptotic result \eqref{PartEntDiv} for small $T$.}
    \label{fig:SuddenDeath}
\end{figure}

To see how temperature affects the destruction of the quantum nonfreeness, we present the latter as a function of both distance and temperature in Figure \ref{fig:SuddenDeath}. The black line depicts $r_{\mathrm{crit}}^{(p)}$ as a function of $T$.
As the temperature increases, the minimum distance needed to kill the entanglement entirely is lowered. Similarly, the critical separation $r_{\mathrm{crit}}^{(p)}$ diverges logarithmically at small temperature,
\begin{equation}
    r_{\mathrm{crit}}^{(p)}(T) = - \frac{1}{2} \log(T) + d_0 +d_1 T+\mathcal{O}(T^2), \quad T \rightarrow 0, \label{PartEntDiv}
\end{equation}
where $d_0\equiv \log(2) - \frac{1}{2} \log(\log(3))$, $d_1\equiv -\frac{1}{2}(2+\log (3))$ are constants.
This asymptotic result is shown as dashed line in Figure \ref{fig:SuddenDeath} and its derivation is included in Appendix \ref{divergence} for the interested readers.

In the form of those results referring to the particle picture we have resolved the correlation paradox in the dissociation limit in the most concise way: For any finite temperature $T$, regardless of how close to zero it might be, there always exists a finite separation distance $ r_{\mathrm{crit}}^{(p)}(T)$ beyond which the quantum state $\rho(T,r)$ of the system does not contain quantum nonfreeness anymore. Instead, $\rho(T,r)$ is given as a purely classical mixture of configuration states. In particular, this means that the quantum nonfreeness in the ground state of the hydrogen molecule is highly unstable against thermal noise, and finite quantum nonfreeness at the dissociation limit can \textit{never} be observed in a laboratory.

\subsection{Correlation Paradox of the Generalized Dissociation Limit}\label{sec:general}

All above discussions of the correlation paradox of the dissociation limit are based on the assumption that only the $1s$ shell orbitals of the two hydrogen nuclei are active, and that there is exactly one electron at each center at the dissociation limit.
\begin{figure}[ht]
    \centering
    \subfigure[Two-center dissociation]{\includegraphics[scale=0.2]{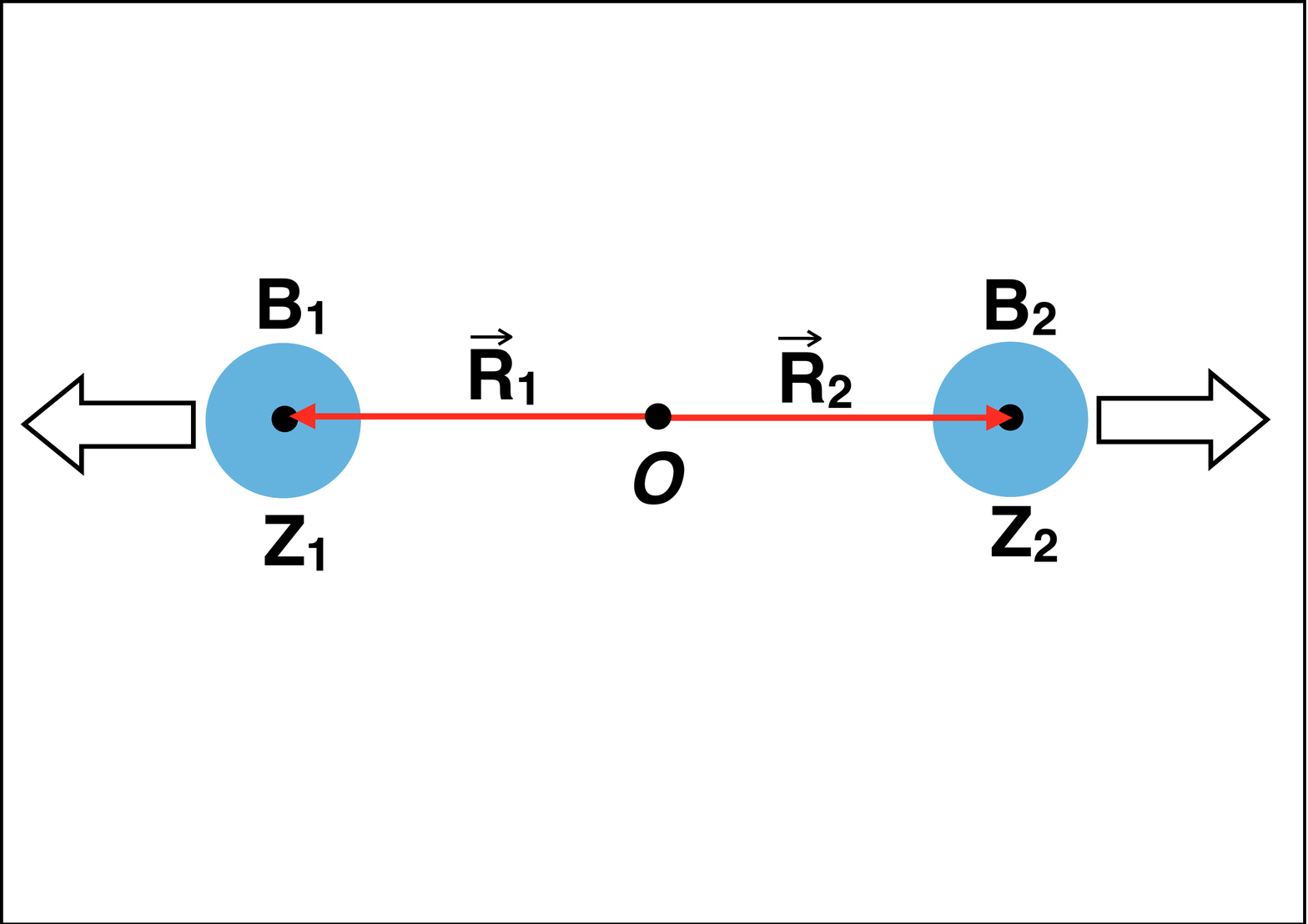}}
    \quad
    \subfigure[Five-center dissociation]{\includegraphics[scale=0.2]{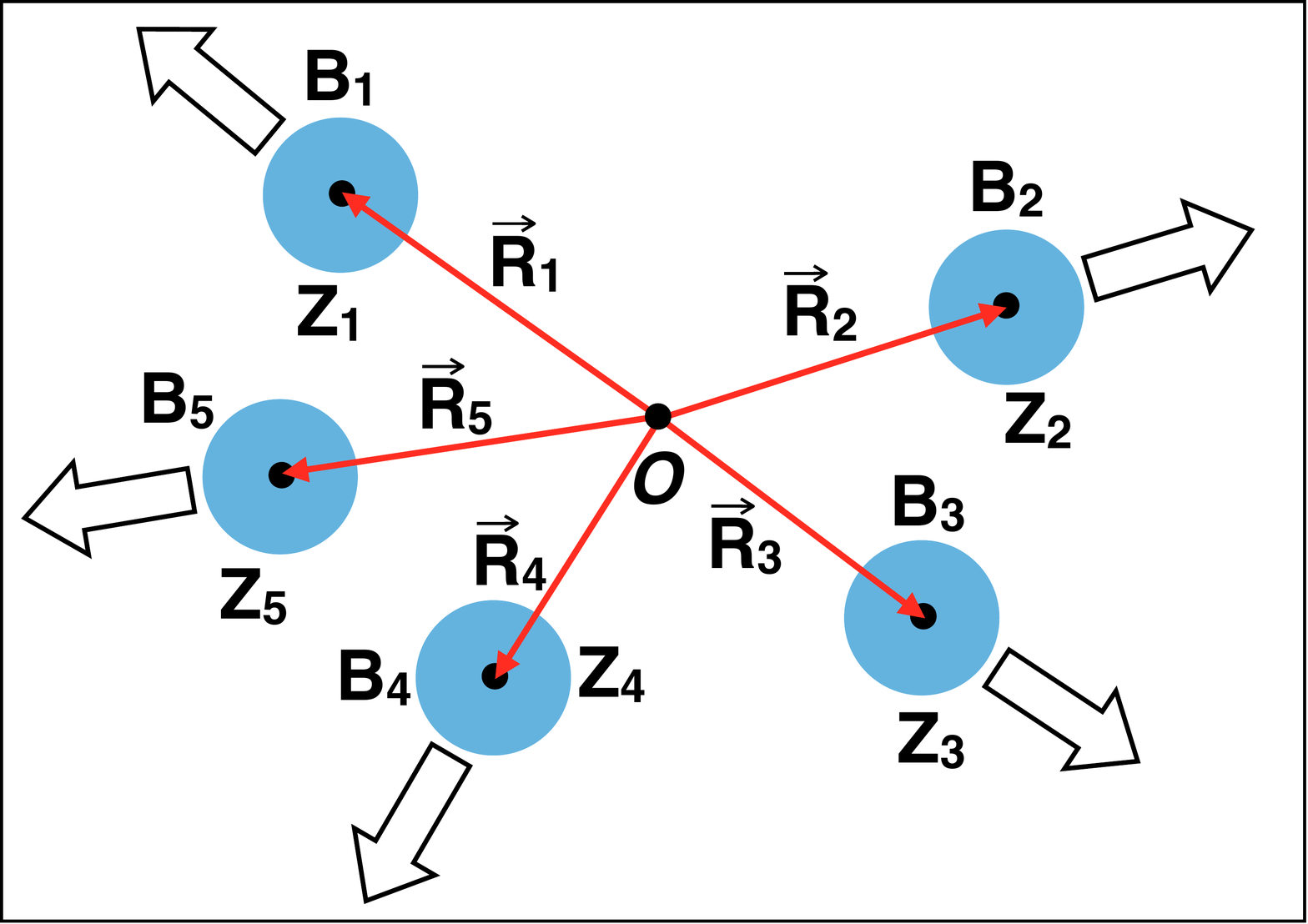}}
    \caption{Schematic illustration of dissociation in general: Various nuclei together with their local bases $B_i$ of atomic spin-orbitals centered at $\vec{R}_i$ are separated from each other (see text for more details).}
    \label{Dissociation}
\end{figure}
In the following we successively relax these assumptions to formulate a hierarchy of generalized correlation paradoxes in the dissociation limit.
In analogy to Sections \ref{sec:diss}, \ref{sec:resol}, we resolve those paradoxes by referring to thermal noise which will destroy in the dissociation limit the mode entanglement between different nuclear centers (and if applicable the quantum nonfreeness).

As illustrated in Figure \ref{Dissociation}, we consider a general molecular system with $\nu$ nuclear centers $Z_i$ at positions $\vec{R}_i$, $i=1,\ldots, \nu$. We then choose finite local basis sets $B_i = \{|\varphi^{(k)}_i\rangle\}_{k=1}^{d_{i}}$
of atomic spin-orbitals which are localized mainly around the respective center $Z_i$. The general dissociation limit can then be formally described, e.g., by the process $\Vec{R}_i \rightarrow \lambda \Vec{R}_i$ as $\lambda \rightarrow \infty$, where $\lambda$ is a scale parameter, e.g., $\lambda=1$ could correspond to the equilibrium geometry of the molecule. Actually, 
it is only crucial for the following considerations that the nuclei separate from each other more and more in the dissociation limit, i.e., $|\vec{R}_i-\vec{R}_j|\rightarrow \infty$ for all $i,j$. Scenarios in which two or more nuclei remain at finite separation distances are included as well and in that case we would merge them to form one joint (more complex) center.
Moreover, we denote the local one-particle Hilbert space at center $Z_i$ by $\mathcal{H}_i = \mathrm{Span}(B_i)$. The corresponding local Fock spaces $\mathcal{F}_i$ follow as the (direct) sums of different fixed particle number sectors $\wedge^N[\mathcal{H}_i]$ generated by $\mathcal{H}_i$,
\begin{equation}
    \mathcal{F}_i = \bigoplus_{N_i=0}^{\scriptsize{\mbox{dim}}(\mathcal{H}_i)} \wedge^{N_i}[\mathcal{H}_i]. \label{Fock}
\end{equation}
For each center $Z_i$ a natural basis $\overline{B}_i$ for its Fock space $\mathcal{F}_i$ is given by the family of configuration states Eq.~\eqref{config1st} which involve only spin-orbitals belonging to $B_i$. The corresponding local observables, i.e., Hermitian operators acting on  $\mathcal{F}_i$ form a local algebra, $\mathcal{A}_i$. Physically admissible operators obey number parity (or particle number) superselection rules, that is, when represented in $\overline{B}_i$, they are block diagonal with respect to the even and odd particle number sector (or all particle number sectors) and thus preserve the number parity (or particle number).

The Fock space $\mathcal{F}$ of the total system is given by the  tensor product of all local Fock spaces,
\begin{equation}\label{Focktotal}
    \mathcal{F} = \bigotimes_{i=1}^\nu \mathcal{F}_i.
\end{equation}
Since the molecular system has a fixed particle number $N$, we could restrict to the corresponding particle number sector of $\mathcal{F}$.
The electronic Hamiltonian $H$ of the molecular system expressed in second quantization can be decomposed into local terms $H_i$ and ``coupling'' terms $H_{ij}$,
\begin{equation}\label{Hdecomp}
    H = \sum_{i=1}^\nu H_i + \sum_{1 \leq i < j \leq\nu} H_{ij}.
\end{equation}
The local terms $H_i$ involve only creation and annihilation operators referring to the spin-orbitals of center $Z_i$ while the coupling terms $H_{ij}$ refer to two centers $Z_i,Z_j$. The latter ones describe the Coulomb pair interaction of electrons/nuclei at center $Z_i$ with those at center $Z_j$ and the hopping of the electrons between those centers (kinetic energy). Consequently, they decay in the dissociation limit, fulfilling
\begin{equation}\label{couplingdecay}
  \|H_{ij}\| \leq \frac{q_{ij}}{|\vec{R}_i-\vec{R}_j|}
\end{equation}
with some appropriate constants $q_{ij}$. In contrast to those coupling terms $H_{ij}$, the local terms $H_i$ are effectively independent of the dissociation limit (only their reference points $\vec{R}_i$ change).

After having formally introduced the general physical system, we can now present a hierarchy of generalized dissociation limits and their resolution on a \emph{qualitative} level:

\begin{enumerate}
\item \emph{Full Electron Separation}. In this scenario, we assume at the dissociation limit that each nuclear center will be occupied by exactly one electron. This is, only the one-particle sectors in the local Fock spaces $\mathcal{F}_i$ in Eq.~\eqref{Fock} are occupied. Prime examples of this situation are the dissociation of the hydrogen molecule $\mathrm{H}_2$, its isotopic variations $\mathrm{HD}$ and $\mathrm{D}_2$ and just any chain or ring of hydrogen atoms.

    Since the dissociation limit spatially separates all electrons, their interaction is also marginalized. Naively, one may therefore expect that the dissociated ground state would take the form of a configuration state,
    \begin{equation}\label{PsiIa}
    \ket{\Psi_0} = \ket{\phi_1}\wedge \ldots \wedge \ket{\phi_{\nu}},
    \end{equation}
    involving from each center $Z_i$ its local one-electron ground state $\ket{\phi_i}$. In fact, this is the case as long as the local ground states $\ket{\phi_i}$ are nondegenerate for all (or all except one) centers $Z_i$. Otherwise, by adding to each $\ket{\phi_i^{(m_i)}}$ a superindex reflecting its possible  degeneracies, the respective $N$-electron ground state will typically take the form of a coherent superposition
    \begin{equation}\label{PsiI}
    \ket{\Psi_0} = \sum_{m_1,\ldots,m_\nu} a_{m_1,\ldots,m_\nu} \,\ket{\phi_1^{(m_1)}}\wedge \ldots \wedge \ket{\phi_{\nu}^{(m_{\nu})}}.
    \end{equation}
    For centers with a unique, non-degenerate ground state the respective sum collapses to just one term, $m_i=1$.

    Since at the dissociation limit the couplings $H_{ij}$ between any two centers vanish \eqref{couplingdecay}, all configurations involved in \eqref{PsiI} have the same energy in that limit. This implies, that the presence of any finite temperature $T$ would turn  the state \eqref{PsiI} into a classical mixture of its configuration states and in that sense resolve the correlation paradox. Since for each of those configuration states the involved one-particle states $\ket{\phi_i^{(m_i)}}\in \mathcal{H}_i$ belong to a definite center (i.e., they are not coherent superpositions of spin-orbitals of different centers), exactly the same will hold true for the mode entanglement and mode correlation between any two centers. 

    \item \emph{Fixed Local Particle Number}. We relax the restriction of having only one electron per center $Z_i$ at the dissociation limit. Yet, we still assume fixed local electron numbers $N_i$, where $\sum_{i=1}^\nu N_i = N$. This type of situations arises, e.g., when a molecule dissociates into two (or more) neutral identical atoms, e.g., $\mathrm{N}_2$ and $\mathrm{O}_2$. Since the $N_i$ electrons at each center are not getting separated there is no reason to expect that the dissociated $N$-electron state would  be a configuration state. Indeed, the non-vanishing interaction between the electrons at each center can give rise to finite correlations.

        Nonetheless, the vanishing of the coupling terms implies that the ground state problem decouples into those of the individual centers.  To be more specific, one just needs to determine the local $N_i$-electron ground state $\ket{\Phi_i}\in \wedge^{N_i}[\mathcal{H}_i]$ of $H_i$ at each center $Z_i$. Naively, due to the vanishing of the coupling terms $H_{ij}$ at the dissociation limit \eqref{couplingdecay} one may then expect a dissociated ground state of the form
        \begin{equation}\label{PsiIIa}
        \ket{\Psi_0} = \ket{\Phi_1}\wedge \ldots \wedge \ket{\Phi_\nu},
        \end{equation}
        which could be seen as a generalized configuration state. 
        In fact, this is the case as long as the local $N_i$-electron ground states $\ket{\Phi_i}$ are nondegenerate for all (or all except one) centers $Z_i$. Otherwise, by adding to each $\ket{\Phi_i^{(m_i)}}$ a superindex reflecting its possible degeneracies, the respective $N$-electron ground state will typically take the form of a coherent superposition
        \begin{equation}\label{PsiIIb}
        \ket{\Psi_0} = \sum_{m_1,\ldots,m_\nu} a_{m_1,\ldots,m_\nu} \,\ket{\Phi_1^{(m_1)}}\wedge \ldots \wedge \ket{\Phi_{\nu}^{(m_{\nu})}}.
        \end{equation}
        Since at the dissociation limit the couplings $H_{ij}$ between any two centers vanish \eqref{couplingdecay}, all generalized configuration states $\ket{\Phi_1^{(m_1)}}\wedge \ldots \wedge \ket{\Phi_{\nu}^{(m_{\nu})}}$ involved in \eqref{PsiIIb} have the same energy in that limit. This implies that the presence of any finite temperature $T$ would turn the state \eqref{PsiIIb} into a classical mixture of those wedge products with equal weights and in that sense resolves this generalized correlation paradox. Again, since each element $\ket{\Phi_i^{(m_i)}}$ belongs in the mode/orbital picture to a local Fock space $\mathcal{F}_i$, exactly the same applies to the mode entanglement and mode correlation between any two centers.

    \item \label{generalparadox} \emph{Mixed Local Particle Numbers}. We relax the assumptions even further, and now allow for mixed local particle numbers. This may even include cases in which the total system is coupled to an environment and therefore may have an indefinite total electron number. Typical example for isolated systems are molecules with an excess or shortage of electrons, such as  $\mathrm{N}_2^+$: At the dissociation limit of $\mathrm{N}_2^+$ the total thirteen-electron state will (in the simplest case) be an equal superposition of two generalized configuration states \eqref{PsiIIa}, one with seven electrons on the left and six on the right, and one with seven electrons on the left and six on the right,
    \begin{equation}\label{PsiIIIa}
        \ket{\Psi_0} = \frac{1}{\sqrt{2}}\left[\ket{\Psi_6^{(L)}} \wedge \ket{\Psi_7^{(R)}} + \ket{\Psi_7^{(L)}} \wedge \ket{\Psi_6^{(R)}}\right].
    \end{equation}
    Since even the simplest possible state \eqref{PsiIIIa} assumed for $\mathrm{N}_2^+$ does not take the form of a generalized configuration state \eqref{PsiIIa}, we have to give up the particle picture (wedge product-based notation) and consider exclusively the mode/orbital picture which is based on second quantization.
    Just to reiterate, the mode reduced density operators are defined with respect to the tensor product structure \eqref{Focktotal}. For instance, the mode reduced density operator for the mode subsystem $B_i$ and $\mathcal{H}_i$, respectively, at center $Z_i$ is obtained by taking the partial trace of the total state $\rho$ with respect to all factors $\mathcal{F}_j$, $j\neq i$. This leads to a reduced density operator acting on the local space $\mathcal{F}_i$ which in general does not have a definite particle number anymore. Yet, as long as the total state $\rho$ has a fixed particle number, any mode reduced density operator is block-diagonal with respect to the different particle number sectors $\wedge^{N_i}[\mathcal{H}_i]$.

    In the mode/ortbial picture, the consequences of the decay \eqref{couplingdecay} of the coupling terms $H_{ij}$ are obvious: Since at the dissociation limit, the spin-orbitals belonging to different centers $Z_i$ do not couple anymore, one may naively expect that the corresponding $Z_i$-mode reduced density operators $\rho_i$ would be uncorrelated, and the total state would take the form
    \begin{equation}\label{PsiIIIb}
        \rho = \rho_1 \otimes \cdots \otimes \rho_\nu.
    \end{equation}
    In case the system is isolated, each $\rho_i$ would be a pure state (with possibly indefinite particle number) on the local Fock space $\mathcal{F}_i$.
     As the example \eqref{PsiIIIa} already illustrates, this is not necessarily the case whenever the dissociated total system has a degenerate ground state space spanned by generalized configuration states \eqref{PsiIIa} with varying local particle numbers. Consequently, in case of a finite temperature $T$, the same happens as in the previous scenario of fixed local particle numbers: All contributing generalized configurations are classically mixed with equal weights (yet those configurations have no definite local particle numbers anymore). Consequently, the mode entanglement and mode correlation between any two centers vanishes in the dissociation limit, regardless of how small $T>0$ is.
\end{enumerate}

In the following we resolve those generalized correlation paradoxes also in a \emph{quantitative} way, at least in the mode/orbital picture. In particular, this will illustrate why and how a finite temperature in combination with the decaying behaviour of the coupling terms $H_{ij}$ affects and eventually kills the mode entanglement and mode correlation in the dissociation limit.
First, it suffices to resolve the most general version of the correlation paradox, since the one for mixed local particle numbers contains
the other two scenarios as special cases. Second, for the sake of providing a resolution of those paradoxes, one can ignore superselection rules, as the entanglement and correlation \textit{without} superselection rules serve as upper bounds for the physically accessible entanglement and correlation\cite{bartlett2003entanglement}. To be more specific, we only need to show that this upper-bound goes to zero at the dissociation limit to resolve the paradox. Third, given a consistent correlation and entanglement measure, the total correlation is always greater or equal to its quantum part. Then, in order to show that the entanglement vanishes at the dissociation limit, it suffices to show the same for the total correlation (as quantified by the quantum mutual information \eqref{eqn:MI} without superselection rules). Combining the above arguments, we can safely claim to resolve in the following various correlation paradoxes by proving that the quantum mutual information between any two centers becomes zero at the dissociation limit at any finite temperature.

In the important work \cite{wolf2008area}, a universal relation has been found between the correlation in multipartite quantum systems and the system's temperature and individual coupling terms. Here we give a demonstration of the underlying ideas by applying it first to the Hubbard dimer. Afterwards we repeat those steps in the context of general molecular system to resolve the correlation paradox in a quantitative way. The Hubbard dimer Hamiltonian \eqref{Hubbard} can be written as
\begin{equation}
    H = H_{L} + H_{R}+ H_{LR},
\end{equation}
where $H_L$ and $H_R$ are terms that act only on the left or right nucleus/site, and $H_{LR}$ denotes the ``coupling'' between both mode subsystems, i.e., the hopping term. To refer more to the previous sections we have replaced here the indices of the nuclear centers according to $1 \mapsto L$ and $2 \mapsto R$.

The thermal equilibrium state $\rho$ \eqref{Gibbs} follows as the minimizer of the free energy
\begin{equation}\label{thermo1}
    F = E - T S.
\end{equation}
Here, $S$ denotes the von Neumann entropy of the total state $\rho$,  $E \equiv \langle H\rangle_{\rho}\equiv \Tr[H\rho]$ and we
denote in the following the mode reduced density operators of $\rho$ for the left and right mode subsystem by $\rho_L$ and $\rho_R$, respectively. As a consequence of the characterization of $\rho$, the free energy of the state $\rho_L \otimes \rho_R$ is larger than that of $\rho$,
\begin{equation}\label{thermo2}
    \Tr[H\rho] - TS(\rho) \leq \Tr[H(\rho_L \otimes \rho_R)] - TS(\rho_L \otimes \rho_R).
\end{equation}
Equivalently, this can be stated as
\begin{equation}\label{thermo3}
    \Tr\left[H\big(\rho-\rho_L \otimes \rho_R\big)\right]   \leq T \left[S(\rho) - S(\rho_L \otimes \rho_R)\right].
\end{equation}
Since the right-hand side of Eq.~\eqref{thermo3} is up to a prefactor $-T$ nothing else than the quantum mutual information of $\rho$
it follows that (for the sake of clarity we make explicit the chosen decomposition of the mode system into left and right)
\begin{equation}
    \begin{split}
        I_{\rho}(L:R) &\leq \frac{1}{T} \Tr[H(\rho_L \otimes \rho_R - \rho)]= \frac{1}{T} \Tr[H_{LR} (\rho_L \otimes \rho_R - \rho)]\leq \frac{2\|H_{LR}\|_F}{T}. \label{wolflaw}
    \end{split}
\end{equation}
In the second line we have used $\Tr[H_{i} (\rho_L \otimes \rho_R - \rho)]=0$ for $i=L/R$.
The remarkable relation \eqref{wolflaw} states that the quantum mutual information is bounded by the strength of the coupling/hopping term $H_{LR}$, and decays to zero in the dissociation limit.

This whole illustration of the work \cite{wolf2008area} in the Hubbard dimer model and in particular \eqref{wolflaw} can also be generalized to a multivariate setting. For this one first defines the generalized quantum mutual information as\cite{watanabe1960information}
\begin{equation}
    I_\rho(1:2:\cdots : \nu) \equiv S(\rho || \rho_1 \otimes \rho_2 \otimes \cdots \otimes \rho_\nu).
\end{equation}
It quantifies the quantum information of the total state $\rho$ which is not yet contained in the single-center reduced density operators $\rho_1,\ldots,\rho_\nu$.  Since it is concerned with a decomposition of the total system into  several subsystems it will be particularly useful for our generalized dissociation limit beyond diatomic molecules.
Various steps of the derivation of \eqref{wolflaw} for the Hubbard dimer can be repeated in a similar fashion to any molecular system in its multi-nuclear dissociation limit as defined at the beginning of this section. To explain this, we first decompose the Hamiltonian of the molecular system according to \eqref{Hdecomp}.
Then, by recalling the characterization of the Gibbs state as the minimizer of \eqref{thermo1} and comparing its free energy to the one of the
product state $\rho_1 \otimes \rho_2 \otimes \cdots \otimes \rho_\nu$, we find
\begin{equation}
    F(\rho) \leq F(\rho_1 \otimes \rho_2 \otimes \cdots \otimes \rho_\nu).
\end{equation}
Plugging in the definition of the free energy and repeating the steps below \eqref{thermo2} immediately leads to the desired final result
\begin{equation} \label{generallaw}
    I_{\rho}(1:2:\cdots:\nu) \leq \frac{2}{T}  \sum_{1 \leq i < j \leq\nu} \|H_{ij}\|_F.
\end{equation}
Relation \eqref{generallaw} in combination with the decay \eqref{couplingdecay} of the coupling terms $H_{ij}$ implies that for any finite temperature $T>0$ the thermal state of the molecular system converges to the mode uncorrelated state $\rho_1 \otimes \rho_2 \otimes \cdots \otimes \rho_\nu$ in the dissociation limit. Hence, the correlation paradox is completely resolved in the mode/orbital picture:
The mode correlation between any two nuclear centers and thus also any correlation function vanishes in the dissociation limit for all molecular quantum systems under realistic experimental conditions. This is due to the presence of thermal noise, regardless of how close the temperature is to the absolute minimum of zero Kelvin.

\section{Orbital Correlation and Entanglement in Quantum Chemistry}\label{sec:qchem}

\subsection{Relevance of Information Tools in Quantum Chemistry}

The difficulty of the ground state problem lies intrinsically in the strong correlation of the ground state. In a lattice system, the local structure of the lattice sites and decay of entanglement allow one to efficiently search for the lowest energy state with the matrix product state ansazt. This method known as density matrix renormalisation group\cite{white1992density,schollwock2011density} (DMRG) has been widely applied and proven to be extremely efficient and accurate.

The success of DMRG attracted attentions from the field of quantum chemistry, and has shown a great deal of potential. In order to apply DMRG to solve for the ground state of a molecule, a number of orbitals are heuristically selected for constructing the ground state candidate. Such orbitals comprise the so-called \textit{active space}. The active space excludes both orbitals that are always fully occupied and are considered ``frozen'', and those that are never occupied and are called ``virtual''. The orbitals in the active space are then the quantum chemistry analogue of lattice sites. In contrast to physical lattice sites, these active orbitals are not localised. The nonlocality of these orbitals are manifested by the highly nonlocal molecular Hamiltonian
\begin{align}\label{eqn:ham}
    H &= \sum\limits_{ij\sigma}
    T_{ij} f^\dagger_{i\sigma} f_{j\sigma}^{\phantom{\dagger}} + \sum\limits_{ijkl\sigma \tau} V_{ijkl} f^\dagger_{i\sigma} f^\dagger_{j\tau} f_{k\tau}^{\phantom{\dagger}} f_{l \sigma}^{\phantom{\dagger}}\,.
\end{align}
parametrised by the one- and two-particle Hamiltonian $T$ and $V$, respectively.
In order to get rid of this nonlocal structure, these active orbitals are typically reordered or in general linearly transformed. Such process is typically guided by two important quantities, namely the single orbital entanglement and orbital-orbital entanglement, as the cost function of orbital transformation. Therefore correctly quantifying these entanglement quantities may potentially improve the performance of DMRG in molecular systems.

In the following sections, we demonstrate our results in Chapter \ref{chap:fermion} and \ref{chap:quantifying} and quantify the single orbital and orbital-orbital entanglement in molecular ground states obtained with DMRG. Hartree-Fock method is applied to the molecules to obtain the Hartree-Fock orbitals, among which the active space is selected for DMRG calculation. All correlation quantities are calculated for the ground state of $\mathrm{H_2O}$ constructed with $8$ orbitals, that of $\mathrm{C_{10}H_8}$ with $10$ orbitals and that of $\mathrm{Cr_2}$ with $28$ orbitals (for details see Ref.\cite{ding2020concept}). Interaction between these Hartree-Fock orbitals are highly nonlocal. Nonetheless, due to the geometric symmetries of the individual molecules, the active spaces can be further divided into subgroups of orbitals, with respect to which the one-particle Hamiltonian $T$ in \eqref{eqn:ham} is block diagonal. For example, in Figure \ref{fig:OneHam} we plotted the matrix elements of $T$ for $\mathrm{H_2O}$ and $\mathrm{C_{10}H_8}$. This grouping of orbitals may be used as a first tool to predict correlation structure within the active space, although later we will see that it is not always consistent with reality. From the ground state $|\Psi\rangle$ we need the single orbital and orbital-orbital reduced density matrices for quantifying entanglement, which will be defined in Section \ref{sec:single} and \ref{sec:orb_orb} respectively. Additionally, the one-particle reduced matrix
\begin{align}
   \gamma_{i\sigma, j\tau} \equiv \bra{\Psi} c^\dagger_{j\tau} c^{\phantom{\dagger}}_{i\sigma} \ket{\Psi}, \quad \sigma,\tau = \uparrow, \downarrow, \label{eqn:one_part}
\end{align}
whose eigenvalues are also called \emph{natural occupation numbers} provides insights into the reference basis-independent intrinsic correlation (as we will discuss below). Notice that $\gamma$ is trace-normalized to the particle number $N$.

\begin{figure}[t]
    \centering
    \includegraphics[scale=0.29]{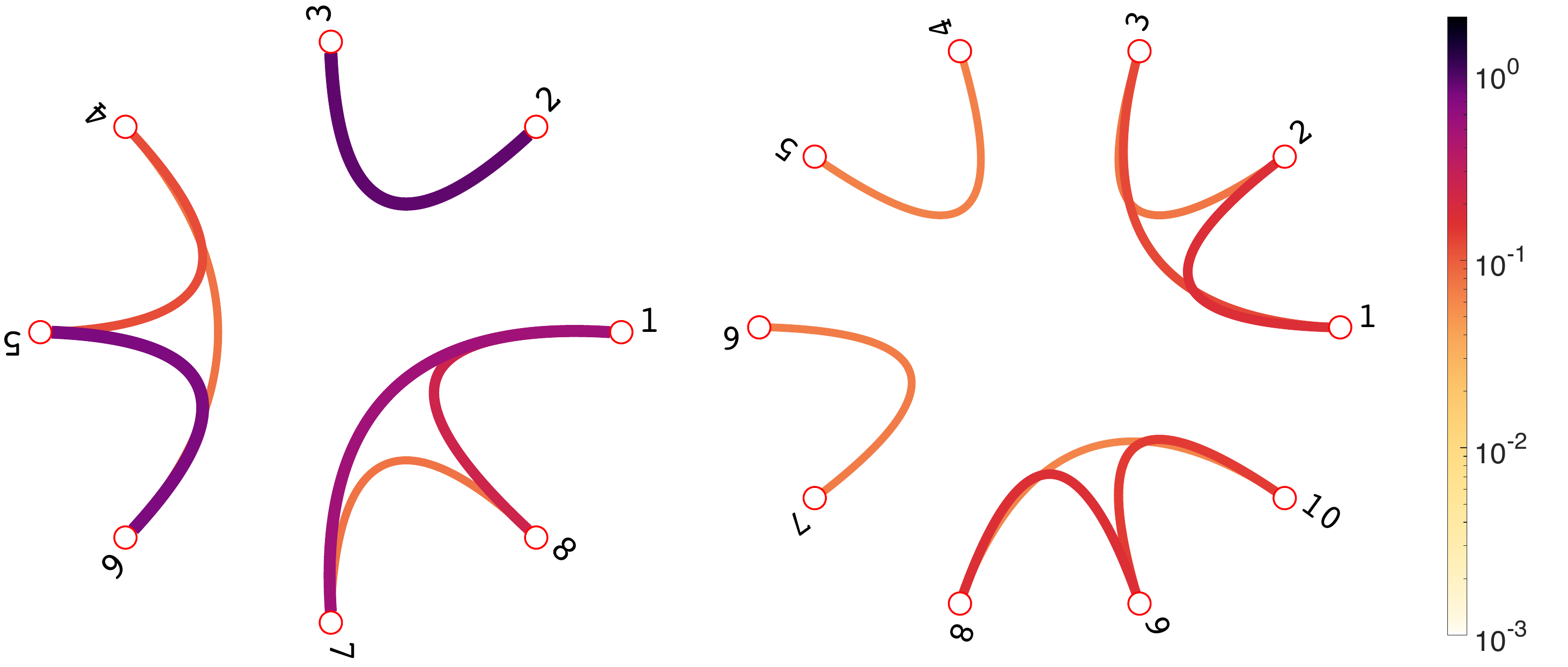}
    \caption{Matrix elements $|T_{ij}|$ of the one-particle Hamiltonian with respect to the Hartree-Fock orbitals for $\mathrm{H_2O}$ (left) and $\mathrm{C_{10}H_8}$ (right).}
    \label{fig:OneHam}
\end{figure}

\subsection{Single Orbital Correlation and Entanglement} \label{sec:single}

We first look at the single-orbital total correlation and entanglement and assume a pure quantum state $\rho =|\Psi\rangle\bra{\Psi}$ for the total $N$-electron system (typically it will be the ground state or an excited state of a molecular system). The one-orbital reduced density matrix associated with the orbital $\ket{\chi_j}$ is obtained by tracing out all remaining orbitals\cite{Sergey17}
\begin{equation}
\rho_j = \Tr_{\setminus \{j\}} [|\Psi\rangle\langle \Psi|] \label{eqn:one_orb_rdm}\,.
\end{equation}
We reiterate that the partial trace $\Tr_{\setminus \{j\}}[\cdot]$ does not mean to trace out particles but instead refers to the tensor product in the second quantization, i.e., it exploits the structure $\mathcal{F} = \mathcal{F}_j \otimes \mathcal{F}_{\setminus \{j\}}$.
From a practical point of view, the non-vanishing entries of the single-orbital reduced density matrix can be determined by calculating
expectation values of $\ket{\Psi}$ involving only fermionic creation and annihilation operators referring to orbital $\ket{\chi_j}$. For more details we refer the reader to Refs.~\cite{boguslawski2013orbital,boguslawski2015orbital,Sergey17}. Due to the fixed particle number and the spin symmetry of $\ket{\Psi}$
the one-orbital reduced density matrix will be always diagonal in the local reference basis
$\{\ket{\Omega},\ket{\!\uparrow},\ket{\!\downarrow},\ket{\!\uparrow\downarrow}\}$
of orbital $\ket{\chi_j}$:
\begin{equation}
\rho_j = \begin{pmatrix}
p_1 & 0 & 0 & 0
\\
0 & p_2 & 0 & 0
\\
0 & 0 & p_3 & 0
\\
0 & 0 & 0 & p_4
\end{pmatrix}. \label{eqn:1rdm}
\end{equation}
By referring to the so-called Schmidt decomposition, the total state then takes the form
\begin{eqnarray}
|\Psi\rangle &=&\sqrt{ p_1} |\Omega\rangle \otimes |\Psi_{N,M}\rangle + \sqrt{p_2} |\!\uparrow\rangle \otimes |\Psi_{N-1,M-\frac{1}{2}}\rangle
\\
& &+ \sqrt{p_3} |\!\downarrow\rangle \otimes |\Psi_{N-1,M+\frac{1}{2}}\rangle + \sqrt{p_4} |\! \uparrow\downarrow\rangle \otimes |\Psi_{N-2,M}\rangle,\nonumber
\end{eqnarray}
where $|\Psi_{n,m}\rangle$ is a quantum state with particle number $n$ and magnetization $m$ of the complementary subsystem comprising the remaining $D-2$ spin-orbitals. Now we can readily determine the physical part $\rho^\textrm{Q}$ in the presence of P-SSR or N-SSR. In the absence of SSRs, the single-orbital entanglement of $\ket{\Psi}$ is simply given by the von Neumann entropy of $\rho_j$, and the single-orbital total correlation is simply twice the entanglement,
\begin{equation}
\begin{split}
E(|\Psi\rangle\langle\Psi|) &= S(\rho_j) = - \sum_{i=1}^4 p_i \log(p_i). \label{eqn:single_noSSR}
\\
I(|\Psi\rangle\langle\Psi|) &= 2 E(|\Psi\rangle\langle\Psi|).
\end{split}
\end{equation}

In the case of Q-SSR ($\text{Q}=\text{P},\text{N}$), we need to consider the physical part $\rho^\text{Q}$ of $\rho = |\Psi\rangle\langle \Psi|$, which is no longer a pure state. Consequently the single-orbital entanglement cannot be quantified by the von Neumann entropy of $\rho_j$ anymore. Instead we have to invoke the geometric picture in Figure \ref{fig:states}. We first calculate the physical states with respect to P-SSR and N-SSR according to \eqref{eqn:tilde}, and then their correlation and entanglement are quantified using \eqref{eqn:MI} and \eqref{eqn:rel_ent}. Explicit derivations are presented in Appendix \ref{app:single}. Remarkably, despite the fact that $\rho^{\textrm{Q}}$ is not a pure state anymore
the single-orbital correlation and entanglement under P-SSR and N-SSR still involves the spectrum of $\rho_j$ only:
\begin{equation}
\begin{split}
I(\rho^\textrm{P}) &= (p_1 + p_4) \log(p_1 + p_4) + (p_2 + p_3)\log(p_2+p_3)
\\
& - 2(p_1 \log(p_1) + p_2 \log(p_2) + p_3 \log(p_3) + p_4 \log(p_4)),
 \\
 I(\rho^\textrm{N}) & = p_1 \log(p_1) + (p_2 + p_3)\log(p_2+p_3) + p_4 \log(p_4)
\\
 &- 2(p_1 \log(p_1) + p_2 \log(p_2) + p_3 \log(p_3) + p_4 \log(p_4)),
 \\
 E(\rho^\textrm{P}) &= (p_1 + p_4) \log(p_1 + p_4) + (p_2 + p_3)\log(p_2+p_3)
\\
& - p_1 \log(p_1) - p_2 \log(p_2) - p_3 \log(p_3) - p_4 \log(p_4),
\\
E(\rho^\textrm{N}) &= (p_2 + p_3) \log(p_2 + p_3) - p_2 \log(p_2) - p_3 \log(p_3). \label{eqn:single_SSR}
\end{split}
\end{equation}
In particular, this implies immediately for both SSRs ($\text{Q}=\text{P},\text{N}$)
\begin{equation}\label{eqn:IvsEpure}
I^\textrm{Q-SSR}(\rho) = E^\textrm{Q-SSR}(\rho) + E(\rho).
\end{equation}
For the case of no SSR, this is consistent with Eq.~\eqref{eqn:single_noSSR}.

After having obtained the ground states of the desired molecules, we can now explore single-orbital correlation and entanglement by applying the respective formulas \eqref{eqn:single_noSSR} and \eqref{eqn:single_SSR}. Since the states $\rho$ at hand are all \emph{pure} states, the single-orbital total correlation without any SSR is always exactly twice the single-orbital entanglement, as stated in \eqref{eqn:single_noSSR}. When P-SSR or N-SSR is taken into account, the respective physical states $\rho^\textrm{P}$ and $\rho^\textrm{N}$ are no longer pure, but in general mixtures of fixed parity or particle number states. However, in the form of Eq.~\eqref{eqn:IvsEpure} there still exists an exact relation between total correlation and entanglement in the presence of SSRs. Because of this, we focus in this section on the entanglement.

\begin{figure}[t]
\centering
\includegraphics[scale=0.35]{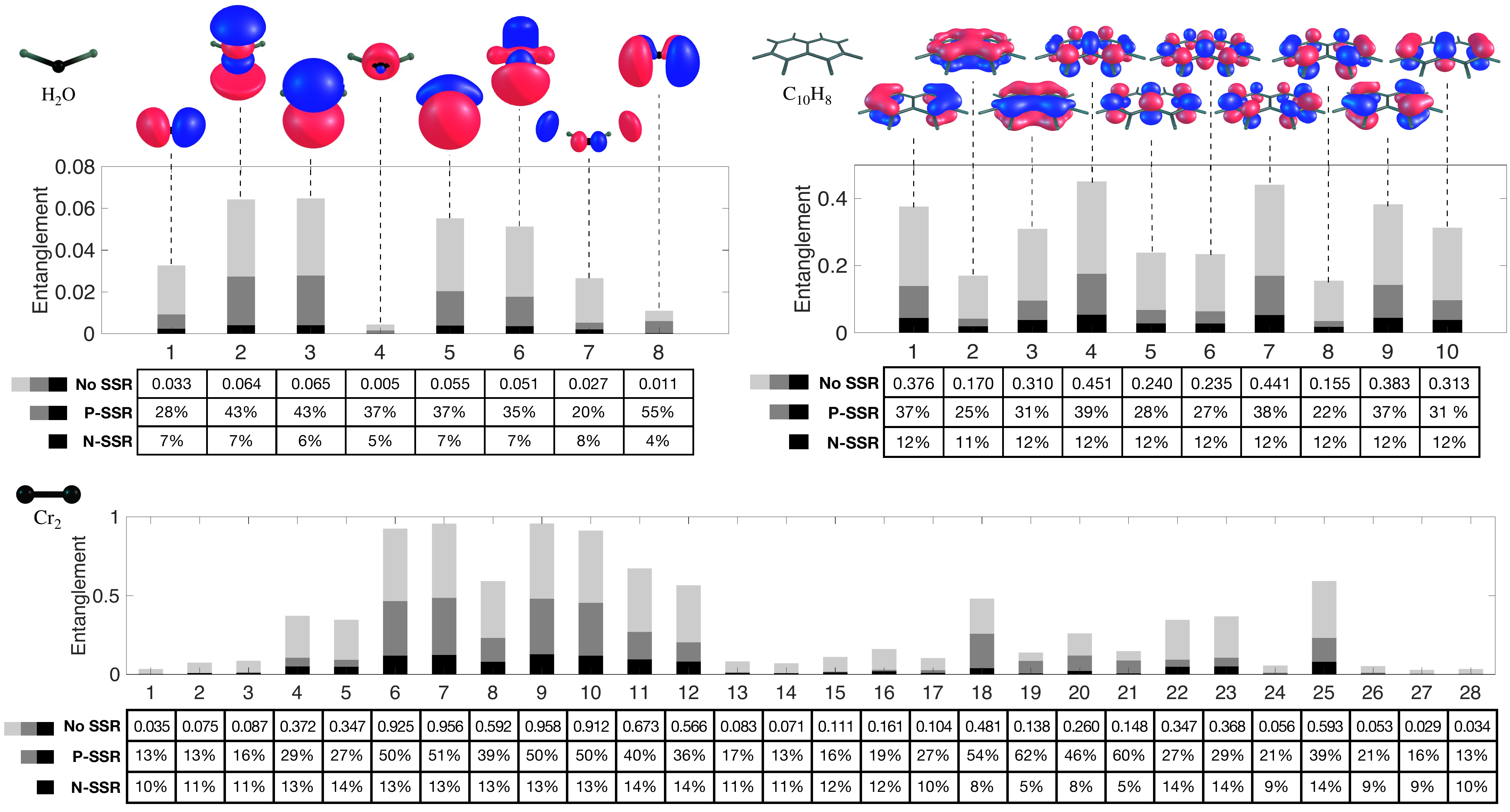}
\caption{Single-orbital entanglement of the Hartree-Fock orbitals (as visualized) in the ground states of $\mathrm{H_2O}$, $\mathrm{C_{10}H_8}$ and $\mathrm{Cr_2}$ for the three cases of no, P- and N-SSR. Exact values of entanglement and the remaining entanglement in terms of percentage of the No SSR case in the presence of P-SSR and N-SSR are listed in the table below each plot.}
\label{fig:1rdm}
\end{figure}

In Figure \ref{fig:1rdm} we plotted the single-orbital entanglement in the ground state of the $\mathrm{H_2O}$, $\mathrm{C_{10}H_8}$ and $\mathrm{Cr_2}$, respectively, for the case without SSR, with P-SSR and with N-SSR, using the analytic formulas \eqref{eqn:single_noSSR} and \eqref{eqn:single_SSR}. Below each figure we listed the exact values of single-orbital entanglement in the absence of SSRs, and also the remaining entanglement in the presence of P-SSR and N-SSR, in percentage.
All these results refer here and in the following to the Hartree-Fock orbitals which are for the sake of completeness also visualized for $\mathrm{H_2O}$ and $\mathrm{C_{10}H_8}$.

Generally speaking, the single-orbital entanglement of Hartree-Fock orbitals is quite small compared to the one of \emph{atomic} orbitals in a bond (typically $\log(2)$, the entanglement of a singlet), particularly for $\mathrm{H_2O}$ and $\mathrm{C_{10}H_8}$. This confirms that the Hartree-Fock orbitals give rise to a much more local structure than that the atomic orbitals and in that sense define a much better starting point for high precision ground state methods.  Comparing the three systems, the water molecule contains the weakest single-orbital entanglement, less than $10^{-1}$ for all eight orbitals, whereas the strongest single-orbital entanglement in naphthalene and the chromium dimer have the values $0.451$ and $0.958$, respectively. This already emphasizes the different levels of correlation in those systems. Yet, it is worth noticing that any type of orbital entanglement and correlation (e.g., single- or two-orbital entanglement) strongly depends on the chosen reference basis. Even for a configuration state \eqref{config} one could find large orbital entanglement and correlation if one referred to orbitals which differ a lot from the natural orbitals.

To dwell a bit more on the concept of correlation, we emphasize that a basis set-independent notion can be defined in terms of the one-particle reduced density matrix $\gamma$ \eqref{eqn:one_part}. To explain this, we first observe that for configuration states \eqref{config} $\gamma$ has eigenvalues (\emph{natural occupation numbers}) all identical to one and zero, respectively. Arranging them in decreasing order gives rise to the so-called ``Hartree-Fock point'' $\vec{\lambda}^{\textrm{HF}} \equiv (1,\ldots,1,0,\ldots,0)$, where the first $N$ entries are $1$, and the remaining $D-N$ are $0$. The distance of the decreasingly ordered natural occupation numbers $\vec{\lambda}\equiv (\lambda_\alpha)_{\alpha=1}^D\equiv \mbox{spec}^\downarrow(\gamma)$ to $\vec{\lambda}^{\mathrm{HF}}$,
\begin{eqnarray}\label{eqn:lambdadist}
\mbox{dist}_1(\vec{\lambda},\vec{\lambda}^{\textrm{HF}})  &\equiv& \sum_{\alpha=1}^D |\lambda_\alpha - \lambda^{\textrm{HF}}_\alpha| \nonumber \\
&=& \sum_{\alpha=1}^N \left(1-\lambda_\alpha\right) + \sum_{\alpha=N+1}^D \lambda_\alpha\,,
\end{eqnarray}
defines an elementary measure for the \emph{intrinsic} (i.e., reference basis-independent) correlation of a quantum state $\rho$.
As a matter of fact, one easily proves\cite{CSthesis,CScorrmeas} that the overlap of an $N$-electron pure states $\ket{\Psi}$ with a configuration state built up from its $N$ first natural spin-orbitals $\ket{\varphi_\alpha}$ fulfills,
\begin{equation}
\frac{1}{2N} \mbox{dist}_1(\vec{\lambda},\vec{\lambda}^{\textrm{HF}}) \leq 1-\big|\!\langle \varphi_1,\ldots,\varphi_N\ket{\Psi}\!\big|^2 \leq \frac{1}{2} \mbox{dist}_1(\vec{\lambda},\vec{\lambda}^{\textrm{HF}})\,.
\end{equation}
This means that the maximized overlap of $\ket{\Psi}$ with a configuration state (Slater determinant) approaches the value one whenever $\vec{\lambda}$ is close to the ``Hartree-Fock point''. The closer a ground state $\ket{\Psi}$ is to the closest configuration state, the more accurate will be the Hartree-Fock approximation for the respective molecular system. Applying the measure \eqref{eqn:lambdadist} of intrinsic correlation to the ground states of water, naphthalene and chromium dimer yields the values $0.004$, $0.025$ and $0.084$, respectively. Those results comprehensively confirm that the systems at hand are not that strongly correlated and water in particular is weakly correlated.
As our analysis in the following section will show, the pairwise orbital entanglement and correlation pattern will be dominated by the point group symmetries of the one-particle Hamiltonian $T$ of the molecule \emph{as long as} the intrinsic correlation of the ground state is small enough.

From a quantum information perspective, the effect of SSRs on the single-orbital entanglement is drastic. The presence of P-SSR and N-SSR considerably reduces the amount of physical entanglement. According to the accompanying tables in Figure \ref{fig:1rdm}, P-SSR eliminates at least $45\%$ of it and occasionally even up to $87\%$. Taking into account the more relevant N-SSR eliminates between $86\%$ and $96\%$. Intriguingly, the entanglement hierarchy, however, remains almost intact. That is, if one orbital is more entangled with the rest than another orbital, the same will likely hold in the presence of P-SSR and N-SSR. It is also worth noting that even the stronger N-SSR does never wipe out the entire entanglement, which we shall see below can happen in the context of orbital-orbital entanglement.

From a quantum chemistry point of view, in Figure \ref{fig:1rdm} the single-orbital entanglement varies significantly from orbital to orbital. In particular, some orbitals are barely correlated with the others. This is a good indicator that our chosen actives spaces were large enough to cover most of the correlation contained in the three molecules. On the other hand, if most orbitals were strongly entangled, the respective active space probably would have been too small. This is also the reason why the single-orbital correlation could help to automate the selection of active orbital spaces in quantum chemistry, as has been suggested and worked out in Refs.~\cite{barcza2011quantum,boguslawski2012entanglement}. Our refined analytic results \eqref{eqn:single_noSSR} and \eqref{eqn:single_SSR} demonstrated in Figure \ref{fig:1rdm} are able to identify exactly the quantum part of the total correlation while also taking into account the important superselection rules. These additional facets make precise the usage of quantum information theoretic concepts in the context of quantum chemistry, and may offer new perspectives into the selection of active space.

\subsection{Orbital-Orbital Correlation and Entanglement} \label{sec:orb_orb}

\begin{figure}[!t]
    \centering
    \includegraphics[scale=0.65]{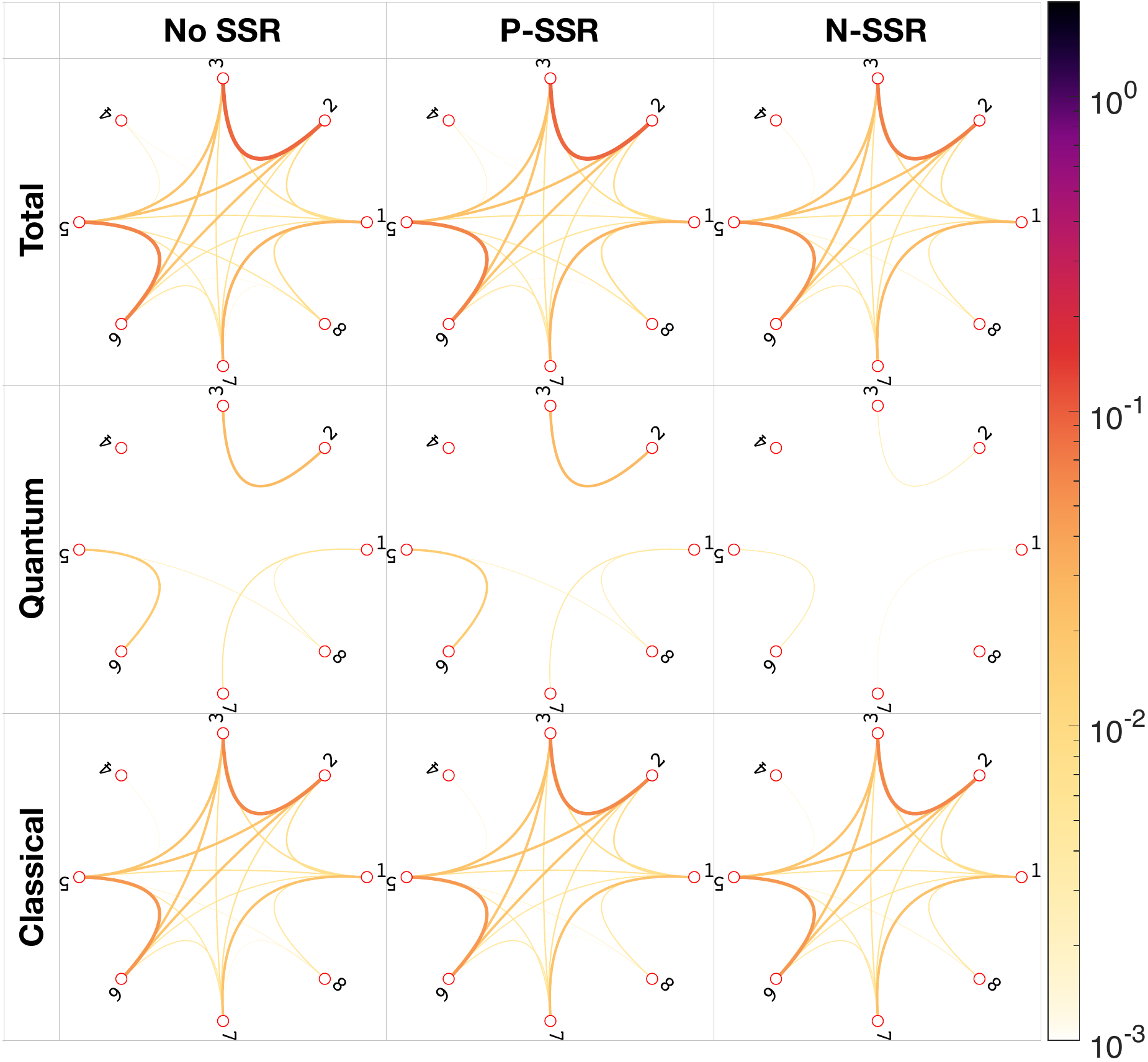}
    \caption{Total correlation, entanglement (``Quantum'') and classical correlation between any two Hartree-Fock orbitals in the ground state of $\mathrm{H_2O}$ for the case with no, P- and N-SSR.}
    \label{fig:CorrWater}
\end{figure}

To provide more detailed insights into the correlation and entanglement structure of molecular ground states, we also study the pairwise correlation and entanglement between two orbitals. This can be done in general in three steps: 1. Obtain the two-orbital reduced density matrix $\rho_{i,j}$ by tracing out all orbital degrees of freedom but orbital $i$ and $j$ as
\begin{equation}
\rho_{i,j} = \Tr_{\setminus \{i,j\}} [|\Psi\rangle\langle\Psi|]. \label{eqn:two_orb_rdm}
\end{equation}
2. Apply the suitable projection to obtain the physical part $\rho_{i,j}^\textrm{Q}$ of $\rho_{i,j}$ under Q-SSR, as explained in Section \ref{sec:SSRIncorp}. 3. Calculate the correlation and entanglement between the two orbital using \eqref{eqn:SSRmeasures}.

For the specific case of a total system consisting of just two orbitals, the only two-orbital ``reduced'' density operator is given by the total (pure) state. Consequently, the orbital-orbital correlation and entanglement thus coincide with the single-orbital ones and the above results \eqref{eqn:single_noSSR},\eqref{eqn:single_SSR} immediately apply.
Due to the electron interaction, the two-orbital reduced density matrices $\rho_{i,j}$ of \emph{general} systems are, however, not pure anymore. For this we will use our results on quantifying mode entanglement in Chapter \ref{chap:quantifying}.

\begin{figure}[!t]
    \centering
    \includegraphics[scale=0.65]{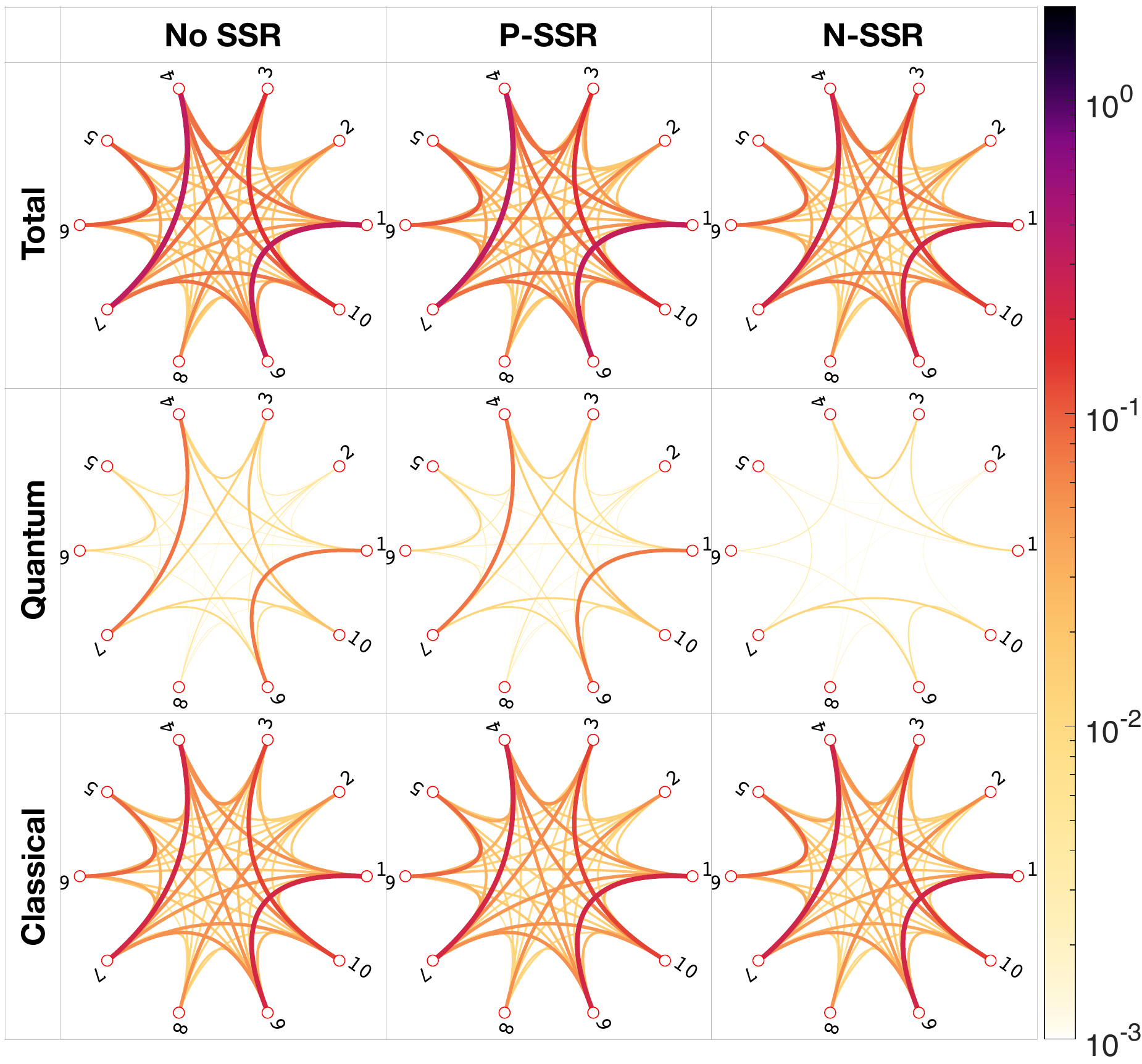}
    \caption{Total correlation, entanglement (``Quantum'') and classical correlation between any two Hartree-Fock orbitals in the ground state of $\mathrm{C_{10}H_8}$ for the case with no, P- and N-SSR.}
    \label{fig:CorrNaph}
\end{figure}

\begin{figure}[t]
    \centering
    \includegraphics[scale=0.65]{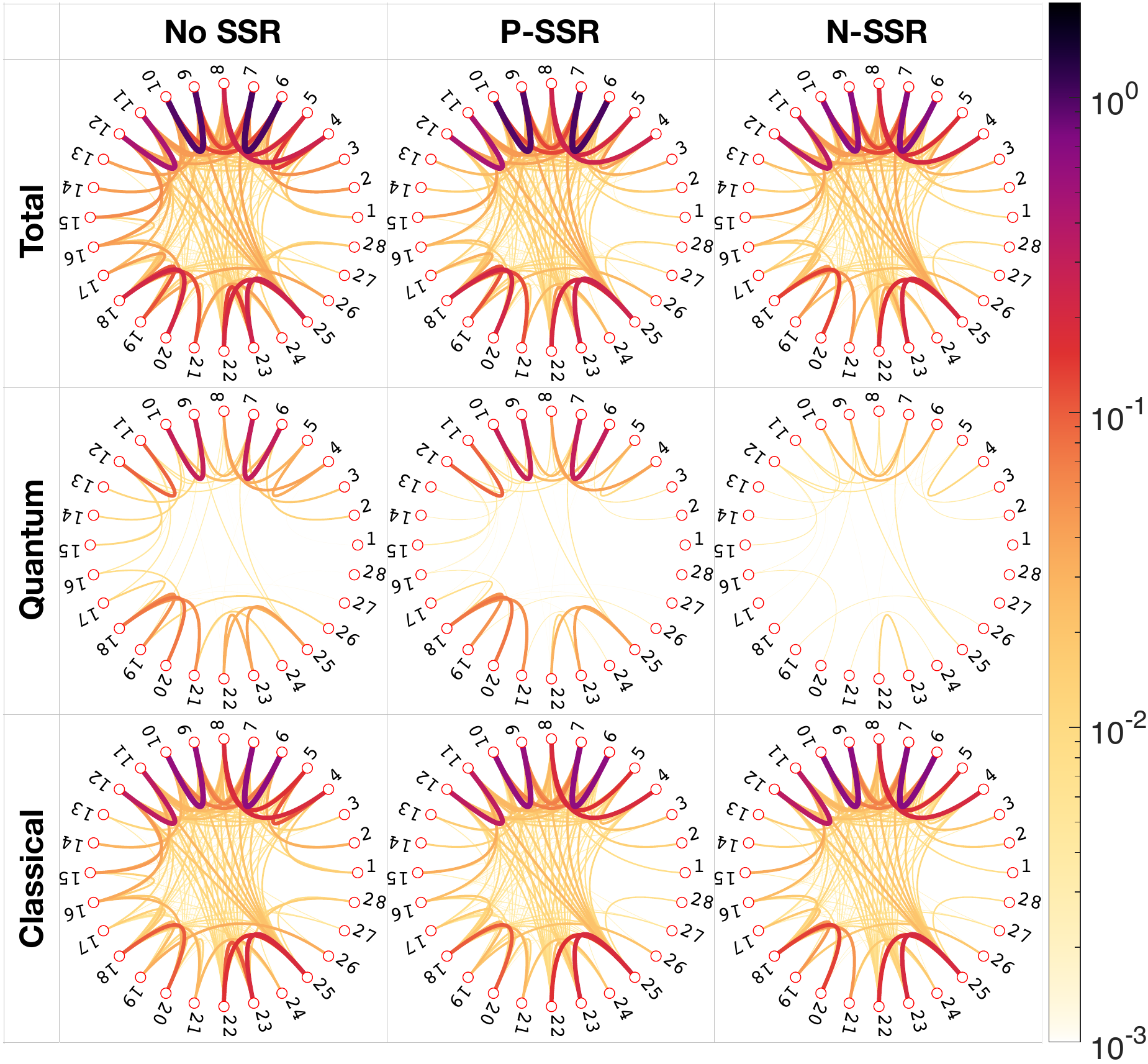}
    \caption{Total correlation, entanglement (``Quantum'') and classical correlation between any two Hartree-Fock orbitals in the ground state of $\mathrm{Cr_2}$ for the case with no, P- and N-SSR.}
    \label{fig:CorrCr2_28}
\end{figure}

The quantities we calculate are the total correlation, entanglement and classical correlation between two Hartree-Fock orbitals, for the case without SSR, with P-SSR and with N-SSR. All those nine quantities are calculate for all pairwise combinations of orbitals, for the ground states of all three molecules $\mathrm{H_2O}$, $\mathrm{C_{10}H_8}$ and $\mathrm{Cr_2}$. Since each ground state is a singlet with a fixed electron number, any two-orbital reduced state $\rho_{i,j}$ is also symmetric with respect to the total two-orbital spin (by Theorem \ref{thrm:symminherit}), magnetization and particle number. Using the symmetry argument\cite{vollbrecht2001entanglement}, the closest separable state $\sigma^\ast_{i,j}$ is block diagonal in the simultaneous eigenbasis of the respective two-orbital spin and particle number operators (as also illustrated in Figure \ref{fig:sectors}). In the case of N-SSR, the projections used for calculating the physical state further increase the symmetry of $\sigma^\ast_{i,j}$, which eventually allows us to determine it analytically, as shown in Section \ref{sec:analytic}. For the case without SSR and with P-SSR, we used semidefinite programming to find the closest separable state and calculate the entanglement to high accuracy, as explained in Section \ref{sec:SDP}.

\begin{figure}[ht]
\centering
\includegraphics[scale=0.42]{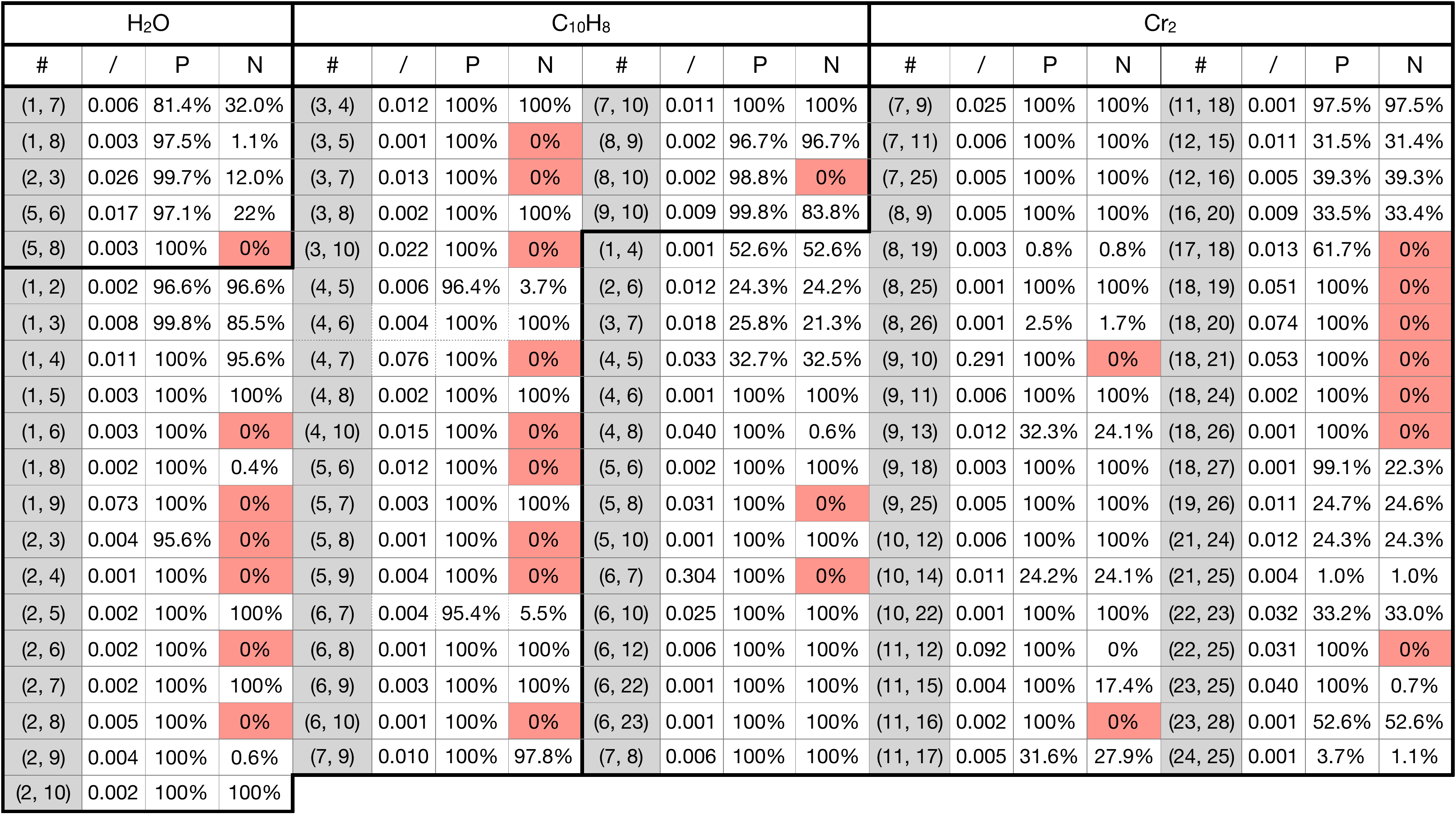}
\caption{Exact values of the orbital-orbital entanglement ($\geq0.001$) without SSR (/) and the fraction of it (in \%) remaining in the presence of P-SSR (P) and N-SSR (N), for the ground states of $\mathrm{H_2 O}$, $\mathrm{C_{10}H_8}$ and $\mathrm{Cr_2}$.}
\label{fig:2rdm}
\end{figure}

In Figure \ref{fig:CorrWater}, \ref{fig:CorrNaph} and \ref{fig:CorrCr2_28} we present the different types of correlation of the ground state of $\mathrm{H_2O}$ constructed with 8 orbitals, $\mathrm{C_{10}H_8}$ with 10 orbitals and $\mathrm{Cr_2}$ with 28 orbitals, respectively. In Figure \ref{fig:2rdm} we list the exact value of orbital-orbital entanglement and the fraction of the entanglement (in \%) which is remaining in the presence of the P-SSR and N-SSR, respectively.

There are several important messages to get across. First of all, similar to the results for the single-orbital entanglement, the water molecule contains the weakest orbital-orbital correlation, and the chromium dimer the strongest. Most importantly, our comprehensive analysis reveals that the quantum part of the total correlation plays only a minor role. In fact, the orbital-orbital entanglement is usually one order of magnitude smaller than the total correlation, and the molecular structure is thus dominated by classical correlation. This key result of our analysis emphasizes that the quantum mutual information \eqref{eqn:MI} is not a suitable tool for quantifying orbital entanglement, as it leads to a gross overestimation. From a general point of view, our findings raise questions about the significance of entanglement in chemical bonding  and quantum chemistry in general.
\begin{figure}[t]
    \centering
    \includegraphics[scale=0.29]{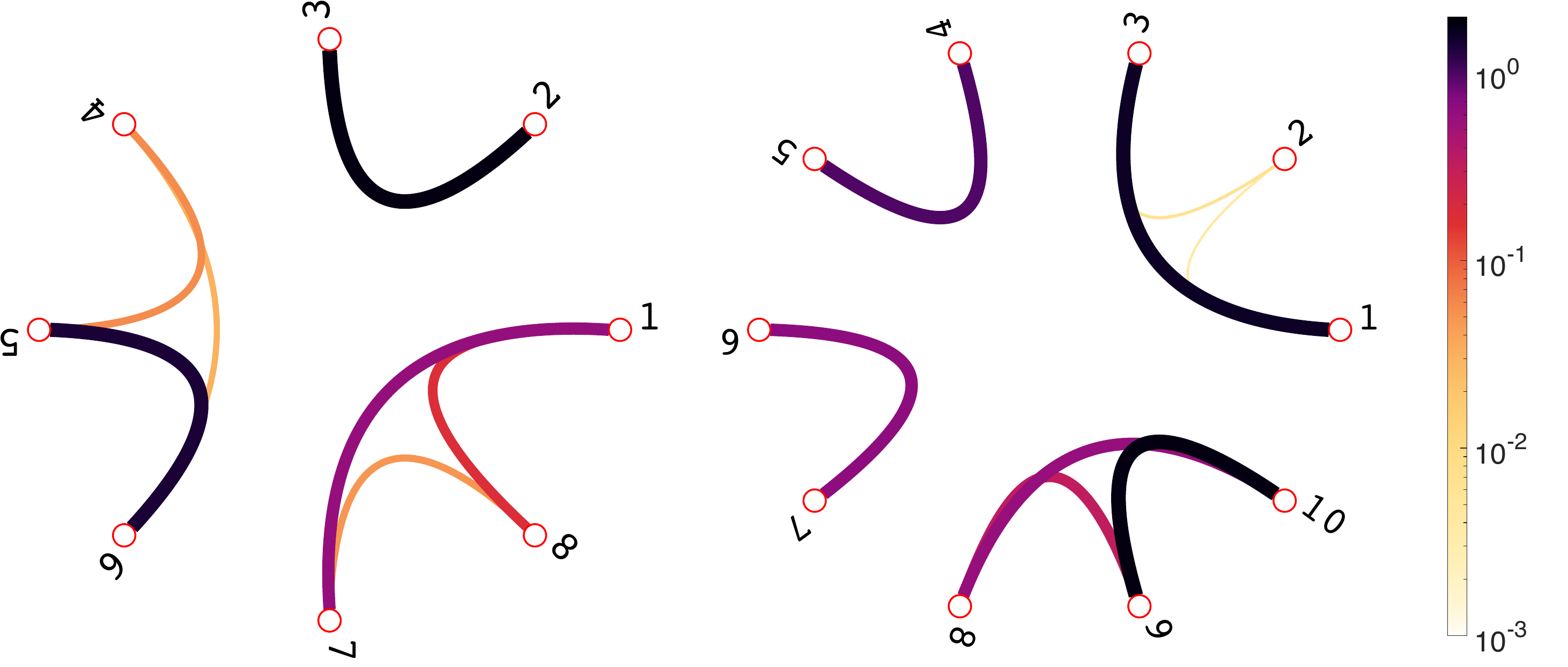}
    \caption{Orbital-orbital total correlation without SSR for $\mathrm{H_2O}$ (left) and $\mathrm{C_{10}H_8}$ (right) with the electron-electron interaction switched off.}
    \label{fig:noint}
\end{figure}

Similar to the single-orbital entanglement, SSRs also have a drastic effect on the orbital-orbital entanglement, yet in a qualitatively different way. In the molecular systems we considered, P-SSR preserved almost all of the orbital-orbital entanglement, whereas in the case of N-SSR, almost no orbital-orbital entanglement is left, and consequently almost all orbital-orbital correlation is classical. Furthermore, in some instances even the \emph{entire} orbital-orbital entanglement is destroyed by the N-SSR. Referring to Figure \ref{fig:sectors}, this indicates that most of the contribution to orbital-orbital correlation and entanglement comes from superposing $f^\dagger_{i\downarrow}f^\dagger_{i\uparrow}|{\Omega}\rangle$ and $f^\dagger_{j\downarrow}f^\dagger_{j\uparrow}|{\Omega}\rangle$, which are marked as the dark grey blocks. These states describe either empty or doubly occupied orbitals. In fact, in all three molecules, single excitations are highly suppressed in any of the molecular orbitals we consider. This is qualitatively different to the entanglement in a singlet bond, which refers to \emph{localized atomic} orbitals, each singly occupied. In agreement with valence bonding theory, this observation confirms that two-orbital correlation and entanglement are suitable tools for describing bonding orders only if they are applied to  localized atomic orbitals.

Lastly, we would like to relate the orbital-orbital correlation pattern to the one-particle Hamiltonian $T$. Due to the point group symmetry arising from the molecular geometry, $T$ represented with respect to the active molecular orbitals is block-diagonal. This is well illustrated in  Figure \ref{fig:OneHam} not showing any coupling in $T$ between Hartree-Fock orbitals belonging to different irreducible symmetry sectors. Exploiting this structure can improve the implementation of numerical methods such as DMRG. Yet, the orbital-orbital correlation patterns inherit that structure only in case the respective ground state is weakly correlated. For the three molecules studied in our work, this is only
the case for the water molecule in agreement to its weak intrinsic correlation as quantified by \eqref{eqn:lambdadist}.
For the other two molecules the more dominant electron-electron interaction results in a major deviation of the orbital-orbital correlation patterns. To confirm these claims, we plotted the orbital-orbital total correlation similar to the top-left cells in Figure \ref{fig:CorrWater} and \ref{fig:CorrNaph}, with the same reference orbitals but with the electron-electron interaction switched off in Figure \ref{fig:noint}. We can see that the correlation patterns now match well the respective structures of the one-particle Hamiltonians in Figure \ref{fig:OneHam}. As a result, much care is needed when using the one-particle Hamiltonian to achieve a localized orbital arrangement, as the unperturbed orbital-orbital correlation and entanglement patterns might be completely scrambled by electron-electron interaction.

  \chapter{Conclusions and Outlook}

The aim of this work is to define precise concepts of correlation and entanglement in fermionic systems, namely mode- and particle- correlation and entanglement, and provide operationally meaningful quantification of entanglement for common electronic systems. The mode picture refers to tensor product structure on the total Fock space on the total set of modes, where the $N$-fermion Hilbert space is embedded. This naturally recovers a tensor product structure between distinct subsets of modes, where we quantify the entanglement using the relative entropy of entanglement. In the particle picture, the concept of ``nonfreeness'' was discussed, and quantified as the distance to the closest quasifree state\cite{gottlieb2005new}. Although conceptually different than the definition of correlation based on a tensor product structure, nonfreeness is a promising tool for measuring the intrinsic complexity of the ground state problem\cite{Pachos17free,Pachos18free,Pachos19free}. As our main result, in the mode picture we took the fundamental superselection rules into account and derived analytic formula for the physical mode entanglement between two sites/orbitals using symmetry. With the system of two sites/orbitals as a building block, we demonstrated our results to concrete systems.

First, we resolved the correlation paradoxes in the dissociation limits of molecular systems in a quantitative way: We have proven that thermal noise due to temperature will destroy the mode entanglement beyond a critical separation distance $r_{\mathrm{crit}}^{(m)}$($T$) and the total mode correlation at the dissociation limit entirely. This means that all correlation functions referring to different nuclear centers vanish in the dissociation limit, provided the temperature is finite. Similarly, we confirmed that in the particle picture thermal noise turns coherent superpositions of (quasi)degenerate configuration states into classical mixtures of them. As a matter of fact, the more general result \eqref{generallaw} emphasizes that any form of perturbation of the system would have the same effect in the dissociation limit as thermal noise. Hence, from a practical point of view, our findings emphasize that neither finite mode entanglement nor finite quantum nonfreeness can ever be observed in the dissociation limit in a laboratory. This also rationalizes and clarifies the perception that the ``correlation'' of the dissociated ground state vanishes.

Secondly, we quantified in Section \ref{sec:qchem} the different correlation types exemplarily in the ground states of the water, naphthalene and dichromium molecule in a numerically exact way. Our findings as presented in Figures \ref{fig:1rdm}-\ref{fig:CorrCr2_28} reveal the following: (i) Compared to the correlation between two (orthonormalized) atomic orbitals in single bonds (order $2\log{2}$), the correlation between most Hartree-Fock orbitals is quite small. This highlights the well-known fact that Hartree-Fock orbitals are a much better starting point for high precision ground state methods than atomic orbitals. (ii) Taking into account the important N-SSR has a drastic effect. It reduces the correlation and entanglement of one orbital with the remaining ones by about $86$-$96$\%. The effect on the two orbital level is significant as well but varies a lot more namely between no reduction and total cancellation (see Figure \ref{fig:2rdm}). Those particular findings raise first doubts about the usefulness of molecular systems as a source for correlation and entanglement. This conclusion may change to some extent, however, if one refers to localized atomic orbitals instead of delocalized molecular orbitals, which is yet to be explored in a precise manner. (iii) The overwhelming part of the total correlation between molecular orbitals is classical. This immediately raises questions about the role of entanglement in the description of chemical bonds and quantum chemistry in general, at least for the entanglement between highly nonlocal molecular orbitals.

There are of course many open questions on fermionic entanglement and correlation. On the fundamental level, one natural question one could ask is if there is a connection between the mode and particle picture. For example, does the existence of entanglement in one picture indicates the same in the other? On the applicational aspect, what entanglement pattern would we observe if we change the Hartree-Fock orbitals in Section \ref{sec:qchem} to localised atomic orbitals? Would this type of orbital entanglement be better suited for describing chemical bonds? In general, atomic orbitals on different nuclei have finite overlaps. To recover the notion of subsystems (similar to orthogonal and localised physical sites) one need to orthogonalise these atomic orbitals in a way that their locality is well preserved. Lastly, exactly quantifying all these correlation and entanglement quantities discussed in Section \ref{sec:ent_meas} and \ref{sec:quantvsclass} is of great independent interest and there are still many gaps to fill. For example, does a closest classical state to $\rho$ exhibit the same local symmetries as $\rho$ as its closest separable state does? If we managed to apply to the quantum discord \eqref{eqn:disc} the same symmetry argument used to quantify the relative entropy of entanglement in Section \ref{sec:symm}, despite the set of classical states $\mathcal{D}_\text{cl}$ being non-convex, it would greatly simplify the calculation of quantum discord and could even lead to a general analytic formula. 
  \begin{appendix}

\chapter{Proof of Proposotion \ref{prop:totalspin}}\label{app:totalspin}

In this section we prove Proposition \ref{prop:totalspin} in Section \ref{sec:derive}. Namely, we show that including $\vec{S}^2$ symmetry in the twirl when particle number $N$ and magnetisation $S^z$ symmetry is already present is not entanglement generating. We denote the unitary group associated with conserved quantities ${N}_A, {N}_B$ and ${S}^z$ as $G'$, and the one with the additional quantity ${\vec{S}}^2$ as $G$. The objective is to check if there exist an entangled states of the form $\sigma = T_{G}(\sigma)$, and a separable state $\sigma' = T_{G'}(\sigma')$ such that $T_{G}(\sigma') = \sigma$.

According to Table \ref{tab:sym} a state $\sigma'$ that satisfies $\sigma' = T_{G'}(\sigma')$ has its restriction to the sector $M = \mathrm{Span}\{|\psi_8\rangle, |\psi_9\rangle, |\psi_{10}\rangle, |\psi_{11}\rangle\}$ of the form
\begin{equation}
    \sigma'|_M = \begin{pmatrix}
        q_{10} & 0 & 0 & 0
        \\
        0 & q_8 & b & 0
        \\
        0 & \overline{b} & q_9 & 0
        \\
        0 & 0 & 0 & q_{11}
        \end{pmatrix},
\end{equation}
in the same eigen-basis $|\Psi_i\rangle$'s in sector $M$ which is solely responsible for all the entanglement in $\sigma'$, by Proposition \ref{prop:M}. Using the Peres-Horodecki criterion, $\sigma'|_M$ is separable if and only if
\begin{equation}
    \left(\frac{q_8-q_9}{2}\right)^2 + \mathrm{Im}(b)^2 \leq q_{10} q_{11}. \label{eqn:G'}
\end{equation}
When we apply $T_G$ to $\sigma'$, i.e. including the generator $|\hat{\vec{S}}|$, the coherence term $b$ vanishes,
\begin{equation}
    T_G(\sigma')|_M = \begin{pmatrix}
        q_{10} & 0 & 0 & 0
        \\
        0 & q_8 & 0 & 0
        \\
        0 & 0 & q_9 & 0
        \\
        0 & 0 & 0 & q_{11}
        \end{pmatrix},
\end{equation}
which satisfies the separability criterion \eqref{eqn:separability} due to \eqref{eqn:G'}. Therefore, if $\sigma'= T_{G'}(\sigma')$ is separable, then $T_G(\sigma')$ is also separable. We conclude that for a N-SSR covariant state $\rho^\text{N}$ with $SU(2)$ spin symmetry, the closest separable state $\sigma^\ast$ has the form of \eqref{eqn:coeff} and restricted by \eqref{eqn:separability}.

\chapter{Spectrum of Hubbard Dimer}\label{spectrum}
The Hubbard dimer model contains four spin-orbitals $\{\ket{L\!\uparrow}, \ket{L\!\downarrow}, \ket{R\!\uparrow}, \ket{R\!\downarrow}\}$ which span together the underlying one-particle Hilbert space $\mathcal{H}_1$. The total Fock space is given as the (direct) sum of various particle number sectors $\mathcal{H}_N=\wedge^N[\mathcal{H}_1]$,
\begin{equation}
    \mathcal{F} = \bigoplus_{N=0}^4 \mathcal{H}_N
\end{equation}
Since we consider the Hubbard dimer as an effective model for the hydrogen molecule in the dissociation limit, we restrict ourselves to the \(N=2\) sector $\mathcal{H}_2$ which has dimension $\binom{4}{2}=6$. We can divide $\mathcal{H}_2$ into spin sectors with magnetization \(M = -1, 0, 1\).
\begin{enumerate}
    \item \(M = \pm1\). Only one possible state in each sector:
    \begin{equation}
        |\Psi_{\uparrow/\downarrow}\rangle  = f^\dagger_{L\uparrow/\downarrow} f^\dagger_{R\uparrow/\downarrow} |0\rangle,
    \end{equation}
    which is therefore also an eigenstates of the Hamiltonian \eqref{Hubbard}. Its  energy is $0$ since no hopping is allowed and harboring two electrons with the same spin is forbidden by the Pauli exclusion principle.

    \item \(M = 0\). A basis of this sector contains four states, which can be grouped into different reflection parity sectors, denoted by $p=\pm$:
    \begin{eqnarray}
        |1_\pm \rangle = \frac{1}{\sqrt{2}}(f^\dagger_{L\uparrow}f^\dagger_{R \downarrow} \mp f^\dagger_{L\downarrow}f^\dagger_{R\uparrow})|0\rangle, \nonumber
        \\
        |2_\pm \rangle = \frac{1}{\sqrt{2}}(f^\dagger_{L\uparrow}f^\dagger_{L \downarrow} \pm f^\dagger_{R\uparrow}f^\dagger_{R\downarrow})|0\rangle.
    \end{eqnarray}
    The state $\ket{1_-}$ belongs to the triplet ($S=1$) while the other three are singlets ($S=0$).
    Energy eigenstates can be found via exact diagonalization:
    \begin{eqnarray}
    E_0 &=& \frac{U}{2} - W, \quad |\Psi_0\rangle = a |1_+\rangle + b |2_+\rangle, \nonumber
    \\
    E_1 &=& 0, \hspace{1.3cm} |\Psi_1\rangle = |1_-\rangle, \nonumber
    \\
    E_2 &=& 0, \hspace{1.3cm} |\Psi_2\rangle = |2_-\rangle, \nonumber
    \\
    E_3 &=& \frac{U}{2} + W, \quad |\Psi_3\rangle = c |1_+\rangle + d |2_+\rangle,
    \end{eqnarray}
    where
    \begin{equation}
     W = \sqrt{\frac{U^2}{4}+4t^2}
    \end{equation}
    and
    \begin{eqnarray}
    a = \sqrt{\frac{W+\frac{U}{2}}{2W}}, \quad b = \frac{2t}{\sqrt{2W\left(W+\frac{U}{2}\right)}}, \nonumber
    \\
    c = -\sqrt{\frac{W-\frac{U}{2}}{2W}}, \quad d = \frac{2t}{\sqrt{2W\left(W-\frac{U}{2}\right)}}.
    \end{eqnarray}

\end{enumerate}

\chapter{Divergences in Disentangling Separations}\label{divergence}

In Figure \ref{fig:ME_sudden} and Figure \ref{fig:SuddenDeath} we presented the curve $r_{\mathrm{crit}}^{(m/p)}(T)$ above which the mode entanglement and the quantum nonfreeness vanished. In particular, we observed a diverging behavior when $T$ approaches zero. In this section we will determine the leading order of these divergences.

\textit{1. Mode/orbital Picture.} When $T$ is small, only the ground state and the first excitation level is activated. The local particle number superselected Gibbs state $\rho^\text{N}$, whose entanglement gives the physical entanglement of the original Gibbs state under superselection rule, can be written as a sum of a separable state and a four-dimensional matrix which can be represented as
\begin{equation}
\begin{split}
    \rho^\text{N}|_{M_1} = \begin{pmatrix}
    e^{- \frac{\Delta E}{ T}} & 0 & 0 & 0
    \\
    0 & A & B & 0
    \\
    0 & B & A & 0
    \\
    0 & 0 & 0 & e^{- \frac{\Delta E}{T}}
    \end{pmatrix} \label{app:restricted}
\end{split}
\end{equation}
referring to the ordered basis states $f^\dagger_{L\uparrow} f^\dagger_{R\uparrow}|0\rangle$, $f^\dagger_{L\uparrow} f^\dagger_{R\downarrow}\ket{0}$, $f^\dagger_{L\downarrow} f^\dagger_{R\uparrow}\ket{0}$, $f^\dagger_{L\downarrow} f^\dagger_{R\downarrow}\ket{0}$ whose span is denoted by $M_1$. Here, we introduced
\begin{equation}
    A = \frac{1}{2}|a|^2 + \frac{1}{2} e^{- \frac{\Delta E}{ T}}, \quad B = -\frac{1}{2}|a|^2 + \frac{1}{2} e^{- \frac{\Delta E}{T}},
\end{equation}
and $a$ is as defined in Appendix \ref{spectrum} and $\Delta E = E_1 - E_0$. It is clear that $\tilde{\rho}$ is separable if and only if the expression in Eq.~\eqref{app:restricted} is separable. By the Peres-Horodecki criterion, Eq.~\eqref{app:restricted} is separable if and only if it has positive partial transpose. Then given a small temperature $T$, $r_{\mathrm{crit}}^{(m)}$ is the inter-nuclei distance such that the partial transpose of Eq.~\eqref{app:restricted} becomes rank deficient. That is, $r = r_{\mathrm{crit}}^{(m)}$ when
\begin{equation}
   A^2 (e^{- \frac{2\Delta E}{ T}} - B^2) = 0. \label{criticalcond}
\end{equation}
Since the factor $A^2$ cannot vanish we just need to solve $e^{- \frac{\Delta E}{ T}} = B$, leading to
\begin{equation}
3e^{- \frac{\Delta E}{T}}=a^2=\frac{1}{2} \left(\frac{1}{\sqrt{16 t^2+1}}+1\right).
\end{equation}
Resorting to the software Mathematica and recalling $t\equiv e^{-r}$ then yields the final result
\begin{equation}
    r_{\mathrm{crit}}^{(m)} = - \frac{1}{2} \log(T) + c_0 +c_1 T+\mathcal{O}(T^2), \quad T \rightarrow 0,
\end{equation}
where $c_0\equiv \log(2) - \frac{1}{2} \log(\log(3))$, $c_1\equiv -\frac{1}{2}(1+\log (3))$ are constants.

\textit{2. Particle Picture.} We recall the ``separability'' criterion from Eq.~\eqref{PartEnt}. We can similarly make the approximation by neglecting higher excitation contributions in the Gibbs state $\rho(T,r)$ as $T$ is small enough compare to $U\equiv 1$. In that case, the matrix $K(\rho)$ defined in Eq.~\eqref{C_SL} follows as
\begin{equation}
    K(\rho) = \begin{pmatrix}
    (b^2-a^2)p^2 & 0 & 0 & 0
    \\
    0 & q^2 & 0 & 0
    \\
    0 & 0 & 0 & -q^2
    \\
    0 & 0 & -q^2 & 0
    \end{pmatrix},
\end{equation}
where $a,b$ are defined as in Appendix \ref{spectrum}, and
\begin{equation}
    p = \sqrt{\frac{1}{1+3e^{- \frac{\Delta E}{ T}}}}, \quad q = \sqrt{\frac{e^{- \frac{\Delta E}{ T}}}{1 + 3 e^{- \frac{\Delta E}{ T}}}},
\end{equation}
with $\Delta E = E_1 - E_0$ as before. Recall the quantum nonfreeness $E^{(p)}(\rho)$ \eqref{PartEnt} is calculated as the largest absolute eigenvalue of $K(\rho)$ minus the absolute values of the rest. It is not hard to see that the absolute eigenvalues of $C(\rho)$ are $|a^2-b^2|p^2$ and $q^2$ where the latter is triply degenerate. At the point $r = r_{\mathrm{crit}}^{(p)}$ at which $E^{(p)}(\rho)$ reaches the value zero, the largest absolute eigenvalue can only be $|a^2-b^2|p^2$, and it must satisfies
\begin{equation}
    |a^2-b^2|p^2 = 3q^2.
\end{equation}
This leads to (using again the software Mathematica)
\begin{equation}
    r_{\mathrm{crit}}^{(p)} = - \frac{1}{2} \log(T) + d_0 +d_1 T+\mathcal{O}(T^2), \quad T \rightarrow 0,
\end{equation}
where $d_0\equiv \log(2) - \frac{1}{2} \log(\log(3))$, $d_1\equiv -\frac{1}{2}(2+\log (3))$ are constants.

\chapter{Single-Orbital Correlation and Entanglement} \label{app:single}
This section is devoted to deriving the formulas in Eq.~\eqref{eqn:single_SSR} for the single-orbital correlation and entanglement under P-SSR and N-SSR. In Section \ref{sec:SSRIncorp} we defined the physical part of a quantum state $\rho$ under P-SSR and N-SSR using the projections
\begin{equation}
\begin{split}
\rho^\textrm{P} &= \sum_{\tau, \tau' = \textrm{odd},\textrm{even}} P_\tau \otimes P_{\tau'} \rho P_\tau \otimes P_{\tau'},
\\
\rho^\textrm{N} &= \sum_{m=0}^{\nu} \sum_{n=0}^{\nu'} P_m \otimes P_n \rho P_m \otimes P_n,
\end{split}
\end{equation}
where $\nu$ and $\nu'$ are the maximal particle numbers allowed on the local subsystems. The total correlation and entanglement available in $\rho$ given the local algebras of observables are restricted by P-SSR and N-SSR, are quantified as the total correlation and entanglement in $\rho^\textrm{P}$ and $\rho^\textrm{N}$ respectively, without the restrictions of superselection rules, according to \eqref{eqn:SSRmeasures}.

By referring to the splitting between orbital $j$ and the remaining ones, resulting in factorizing the total Fock space as $\mathcal{F}= \mathcal{F}_j \otimes \mathcal{F}_{\setminus \{j\}}$, and also assuming particle number and spin symmetries, the ground state $|\Psi\rangle$ of the total system admits the following Schmidt decomposition
\begin{eqnarray}
|\Psi\rangle &=&\sqrt{ p_1} |\Omega\rangle \otimes |\Phi_{N,M}\rangle + \sqrt{p_2} |\!\uparrow\rangle \otimes |\Phi_{N-1,M-\frac{1}{2}}\rangle
\\
&& \quad + \sqrt{p_3} |\!\downarrow\rangle \otimes |\Phi_{N-1,M+\frac{1}{2}}\rangle + \sqrt{p_4} |\! \uparrow\downarrow\rangle \otimes |\Phi_{N-2,M}\rangle. \nonumber
\end{eqnarray}
If we consider the P-SSR, the coherent terms between different local parity sectors are excluded according to Eq.~\eqref{eqn:tilde} and Figure \ref{fig:sectors}, leading to the physical state
\begin{equation}
\rho^\textrm{P} = (p_1 + p_4) |\Psi_{\textrm{even}}\rangle\langle\Psi_{\textrm{even}}| + (p_2 + p_3) |\Psi_{\textrm{odd}}\rangle \langle \Psi_{\textrm{odd}}|,
\end{equation}
where $\rho = |\Psi\rangle\langle\Psi|$ and
\begin{equation}
\begin{split}
|\Psi_{\textrm{even}}\rangle &\equiv \sqrt{\frac{p_1}{p_1+p_4}} |\Omega\rangle \otimes |\Phi_{N,M}\rangle + \sqrt{\frac{p_4}{p_1+p_4}} |\! \uparrow\downarrow\rangle \otimes |\Phi_{N-2,M}\rangle,
\\
 |\Psi_{\textrm{odd}}\rangle &\equiv \sqrt{\frac{p_2}{p_2+p_3}}  |\!\uparrow\rangle \otimes |\Phi_{N-1,M-\frac{1}{2}}\rangle  + \sqrt{\frac{p_3}{p_2+p_3}}  |\!\downarrow\rangle \otimes |\Phi_{N-1,M+\frac{1}{2}}\rangle.
 \end{split}
\end{equation}
Similarly for N-SSR, the physical state is
\begin{equation}
\rho^\textrm{N} = p_1 |\Psi_0\rangle\langle\Psi_0| + (p_2+p_3) |\Psi_1\rangle\langle\Psi_1| + p_4 |\Psi_2\rangle\langle\Psi_2|,
\end{equation}
where
\begin{equation}
\begin{split}
|\Psi_0\rangle &\equiv  |\Omega\rangle \otimes |\Phi_{N,M}\rangle ,
\\
|\Psi_1\rangle & \equiv |\Psi_\textrm{odd}\rangle,
\\
|\Psi_2\rangle & \equiv |\! \uparrow\downarrow\rangle \otimes |\Phi_{N-2,M}\rangle.
\end{split}
\end{equation}
To calculate the total correlation, we first determine the spectra of the respective reduced density matrices $\rho^\textrm{P}_j$, $\rho^\textrm{P}_{\setminus \{j\}}$ and $\rho^\textrm{N}_j$, $\rho^\textrm{N}_{\setminus \{j\}}$. Due to the highly symmetric total state, all four matrices are isospectral as $\rho_1$ in \eqref{eqn:1rdm}. Using the definition of the quantum mutual information in \eqref{eqn:MI} we obtain
\begin{equation}
\begin{split}
I(\rho^\textrm{P}) &= (p_1 + p_4) \ln(p_1 + p_4) + (p_2 + p_3)\ln(p_2+p_3)
\\
 -& 2(p_1 \ln(p_1) + p_2 \ln(p_2) + p_3 \ln(p_3) + p_4 \ln(p_4)),
 \\
 I(\rho^\textrm{N}) & = p_1 \ln(p_1) + (p_2 + p_3)\ln(p_2+p_3) + p_4 \ln(p_4)
\\
 -& 2(p_1 \ln(p_1) + p_2 \ln(p_2) + p_3 \ln(p_3) + p_4 \ln(p_4)).
\end{split}
\end{equation}
For the single-orbital entanglement, we use Theorem \ref{thm:trace}, which allows us, given certain criteria are met, to separate the single-orbital entanglement into the entanglement of its pure state decomposition,
\begin{equation}
\begin{split}
E(\rho^\textrm{P}) &= (p_1+p_4) E(|\Phi_\textrm{even}\rangle\langle\Phi_\textrm{even}|) +  (p_2+p_3) E(|\Phi_\textrm{odd}\rangle\langle\Phi_\textrm{odd}|),
\\
E(\rho^\textrm{N}) & = (p_2+p_3) E(|\Phi_\textrm{odd}\rangle\langle\Phi_\textrm{odd}|).
\end{split}
\end{equation}
Using the von Neumann entropy as the entanglement measure for pure states, the single-orbital entanglement in the presence of P-SSR and N-SSR is also determined solely by the spectrum of the one-orbital reduced density matrix,
\begin{equation}
\begin{split}
E(\rho^\textrm{P}) &= (p_1 + p_4) \ln(p_1 + p_4) + (p_2 + p_3)\ln(p_2+p_3)
\\
& - p_1 \ln(p_1) - p_2 \ln(p_2) - p_3 \ln(p_3) - p_4 \ln(p_4).
\\
E(\rho^\textrm{N}) &= (p_2 + p_3) \ln(p_2 + p_3) - p_2 \ln(p_2) - p_3 \ln(p_3).
\end{split}
\end{equation}

\end{appendix}

  \backmatter
\bibliographystyle{jkthesis}
\bibliography{diss}
  \markboth{}{}

  \addcontentsline{toc}{chapter}{\protect Acknowledgment}

\chapter*{Acknowledgments}

I would like to express my gratitude to my supervisor Dr. Christian Schilling from whom I have learned a great deal, for his guidance and tremendous support. I would also like to thank Dr. Zoltán Zimborás, Dr. Sreetama Das and Sam Mardazad for many, many insightful discussions. Special thank goes to Sam Mardazad for providing the the molecular ground states data in Section \ref{sec:qchem} which laid the ground work for my entanglement analysis. I would like to thank my parents, who have always supported my pursuits in every possible way. Finally I want to thank my partner Gian Rossini, who has effortlessly bettered my best and worst days.

  \include{lebenslauf}

\end{document}